\documentclass[acmsmall,screen]{acmart}

\acmJournal{PACMPL}
\acmVolume{1}
\acmNumber{POPL} % CONF = POPL or ICFP or OOPSLA
\acmArticle{1}
\acmYear{2018}
\acmMonth{1}
\acmDOI{} % \acmDOI{10.1145/nnnnnnn.nnnnnnn}
\startPage{1}

\setcopyright{none}
\usepackage{booktabs}   %% For formal tables:
\usepackage{subcaption} %% For complex figures with subfigures/subcaptions
\usepackage{mathrsfs}
\usepackage{blindtext}
\usepackage[frozencache]{minted} 
\usepackage{eqnarray}
\usepackage{mathtools}
\usepackage{booktabs}
\usepackage{bussproofs}

\newif\ifdraft\draftfalse
\ifdraft
\newcommand{\anthony}[1]{\color{red} {AL: #1 :LA} \color{black}}
\newcommand{\zhilin}[1]{\color{brown} {ZL: #1 :LZ} \color{black}}
\newcommand{\tl}[1]{\color{blue} {TL: #1 :LT} \color{black}}
\newcommand{\mat}[1]{\color{cyan} {MH: #1 :HM} \color{black}}
\newcommand{\philipp}[1]{\color{magenta} {PR: #1 :PR} \color{black}}
\newcommand{\zhilei}[1]{\color{violet} {ZLH: #1 :HZL} \color{black}}
\else
\newcommand{\anthony}[1]{}
\newcommand{\zhilin}[1]{}
\newcommand{\tl}[1]{}
\newcommand{\mat}[1]{}
\newcommand{\zhilei}[1]{}
\fi

\newif\ifproceeding\proceedingtrue

\newcommand{\set}[1]{\{ #1 \}}

\newcommand{\Nat}{\ensuremath{\mathbb{N}}}

\newcommand{\Int}{\ensuremath{\mathbb{Z}}}
\newcommand{\eqdef}{\stackrel{\mbox{\begin{tiny}def\end{tiny}}}{=}} % =def=
\newcommand {\pspace} {\textsc{pspace}}

\newcommand{\cut}[1]{}

\mathchardef\mhyphen="2D % hyphen while in math mode

\newcommand{\hide}[1]{}

\newcommand\dom{\mathsf{dom}}

\newcommand\nat{\mathbb{N}}

\newcommand\cA{\mathcal{A}}
\newcommand\cB{\mathcal{B}}
\newcommand\cC{\mathcal{C}}
\newcommand\cE{\mathcal{E}}

\newcommand\cP{\mathcal{P}}
\newcommand\cR{\mathcal{R}}

\newcommand\cT{\mathcal{T}}

\newcommand\replace{\mathsf{replace}}
\newcommand\replaceall{\mathsf{replaceAll}}

\newcommand\extract{\mathsf{extract}}

\newcommand\strline{\mathsf{STR}}

\newcommand\strlinesl{\mathsf{STR}_{\mathsf{SL}}}

\newcommand{\ASSERT}[1]{\textsf{assert}\left(#1\right)}

\newcommand{\concat}{\cdot}

\newcommand{\FA}{FA}

\newcommand{\ialphabet}{\Sigma}

\newcommand{\Lang}{\mathscr{L}}

\newcommand{\controls}{\ensuremath{Q}}
\newcommand{\finals}{\ensuremath{F}}
\newcommand{\transrel}{\ensuremath{\delta}}

\newcommand{\defn}[1]{\emph{#1}}

\newcommand\bigO{\mathcal{O}}

\newcommand\Aut{\mathcal{A}}

\newcommand{\tup}[1]{\left( #1 \right)}

\newcommand\ap[2]{{#1}\mathord{\brac{#2}}}

\newcommand{\opset}{\mathscr{O}}

\newcommand\brac[1]{\left(#1\right)}

\newcommand\lang[1]{\mathcal{L}\mathord{\brac{#1}}}

\newtheorem{remark}{Remark}

\newcommand{\regexp} {{RegEx}}
\newcommand{\regexps} {{RegExes}}
\newcommand{\pcre} {RegEx}

\newcommand{\tmtextit}[1]{{\itshape{#1}}}

\newcommand\PSST{{\sf PSST}}
\newcommand\refexp{{\sf REP}}

\newcommand\pnfa{\mathcal{A}}
\newcommand\psst{\mathcal{T}}

\newcommand\pat{\mathsf{pat}}
\newcommand\rep{\mathsf{rep}}

\newcommand\refbefore{\$^{\leftarrow}}
\newcommand\refafter{\$^{\rightarrow}}

\newcommand\nullchar{\mathsf{null}}

\newcommand\ostrich{OSTRICH}
\newcommand\expose{ExpoSE}

\newcommand{\OMIT}[1]{}

\newcommand{\seq}[1]{\ensuremath{#1}}
\newcommand{\seqq}[1]{\seq{\Gamma\ifx#1\relax\else,#1\fi}}

\newcommand{\ruleName}[1]{\textsc{#1}}

\newcommand{\infer}[3][]{%
  \AxiomC{$#3$}
  \ifx#1\relax\else\LeftLabel{\ruleName{#1}}\fi
  \UnaryInfC{$#2$}
  \DisplayProof
}
\newcommand{\inferC}[4][]{%
  \AxiomC{$#4$}
  \RightLabel{$~~#2$}
  \ifx#1\relax\else\LeftLabel{\ruleName{#1}}\fi
  \UnaryInfC{$#3$}
  \DisplayProof
}
\newcommand{\inferii}[4][]{%
  \AxiomC{$#3$}
  \AxiomC{$#4$}
  \ifx#1\relax\else\LeftLabel{\ruleName{#1}}\fi
  \BinaryInfC{$#2$}
  \DisplayProof
}

\begin{document}

%% Title information
\title[Solving String Constraints With Regex-Dependent Functions]{Solving String Constraints With Regex-Dependent Functions Through Transducers With Priorities And Variables (Technical Report)}
%
%With Regex-Dependent Functions Through Transducers With Priorities And Variables]{Solving String Constraints With Regex-Dependent Functions Through Transducers With Priorities And Variables}         %% [Short Title] is optional;
                                        %% when present, will be used in
                                        %% header instead of Full Title.
%\titlenote{}             %% \titlenote is optional;
                                        %% can be repeated if necessary;
                                        %% contents suppressed with 'anonymous'
%\subtitle{}                     %% \subtitle is optional
%\subtitlenote{}       %% \subtitlenote is optional;
                                        %% can be repeated if necessary;
                                        %% contents suppressed with 'anonymous'

%% Author information
%% Contents and number of authors suppressed with 'anonymous'.
%% Each author should be introduced by \author, followed by
%% \authornote (optional), \orcid (optional), \affiliation, and
%% \email.
%% An author may have multiple affiliations and/or emails; repeat the
%% appropriate command.
%% Many elements are not rendered, but should be provided for metadata
%% extraction tools.

%% Author with single affiliation.

\author[T. Chen]{Taolue Chen}
\orcid{0000-0002-5993-1665}
\affiliation{
	%\position{Position1}
	\department{Department of Computer Science}              %% \department is recommended
	\institution{Birkbeck, University of London}            %% \institution is required
	\streetaddress{Malet Street}
	\city{London}
	%\state{State1}
%	\postcode{WC1E 7HX}
	\country{United Kingdom}
}
\email{t.chen@bbk.ac.uk}          %% \email is recommended

%%%%%%%%%%%%%%%%%%%%%%%%%%%%%%%%%%%%%%%%%%%%%%%%%%%%%%%%%%%%%%%

\author[A. Flores-Lamas]{Alejandro Flores-Lamas}
\affiliation{
	\department{Department of Computer Science} %% \department is recommended
	\institution{Royal Holloway, University of London} %% \institution is required
	\streetaddress{Egham Hill}
	\city{Egham}
	\state{Surrey}
	\postcode{TW20 0EX}
	\country{United Kingdom}
}
\email{ Alejandro.Flores-Lamas@rhul.ac.uk}

%%%%%%%%%%%%%%%%%%%%%%%%%%%%%%%%%%%%%%%%%%%%%%%%%%%%%%%%%%%%%%%%%%%

\author[M. Hague]{Matthew Hague}
\orcid{0000-0003-4913-3800}
\affiliation{
	%\position{Position1}
	\department{Department of Computer Science} %% \department is recommended
	\institution{Royal Holloway, University of London} %% \institution is required
	\streetaddress{Egham Hill}
	\city{Egham}
	\state{Surrey}
	\postcode{TW20 0EX}
	\country{United Kingdom}
}
\email{matthew.hague@rhul.ac.uk}          %% \email is recommended

%%%%%%%%%%%%%%%%%%%%%%%%%%%%%%%%%%%%%%%%%%%%%%%%%%%%%%%%%%%%%%%%%%%
\author[Z. Han]{Zhilei Han}
\orcid{0000-0001-9171-4997}
\affiliation{
	\department{School of Software}
	\institution{Tsinghua University}
	\country{China}
}
\email{hzl21@mails.tsinghua.edu.cn}

%%%%%%%%%%%%%%%%%%%%%%%%%%%%%%%%%%%%%%%%%%%%%%%%%%%%%%%%%%%%%%%%%%%

\author[D. Hu]{Denghang Hu}
\affiliation{
	\department{State Key Laboratory of Computer Science} %% \department is recommended
	\institution{Institute of Software, Chinese Academy of Sciences  \& University of Chinese Academy of Sciences} %% \institution is required
	%\streetaddress{Street1 Address1}
	%\city{City1}
	%\state{State1}
	%\postcode{Post-Code1}
	\country{China}
}
\email{hudenghang15@mails.ucas.ac.cn}

%%%%%%%%%%%%%%%%%%%%%%%%%%%%%%%%%%%%%%%%%%%%%%%%%%%%%%%%%%%%%%%%%%%

\author[S. Kan]{Shuanglong Kan}
\affiliation{
\institution{University of Kaiserslautern}
\city{Kaiserslautern}
\country{Germany}
}
\email{shuanglong@cs.uni-kl.de}

%%%%%%%%%%%%%%%%%%%%%%%%%%%%%%%%%%%%%%%%%%%%%%%%%%%%%%%%%%%%%%%%%%%

\author[A. W. Lin]{Anthony W. Lin}
\orcid{0000-0003-4715-5096}
\affiliation{
	%\position{Position1}
%	\department{ }              %% \department is recommended
    \institution{University of Kaiserslautern \& Max-Planck Institute}            %% \institution is required
%	\streetaddress{}
	\city{Kaiserslautern}
	%\state{State1}
%	\postcode{}
	\country{Germany}
}
\email{lin@cs.uni-kl.de}          %% \email is recommended

%%%%%%%%%%%%%%%%%%%%%%%%%%%%%%%%%%%%%%%%%%%%%%%%%%%%%%%%%%%%%%%%%%%

\author[P. R\"ummer]{Philipp R\"ummer}
\orcid{0000-0002-2733-7098}
\affiliation{
	%\position{Position1}
	\department{Department of Information Technology}              %% \department is recommended
	\institution{Uppsala University}            %% \institution is required
	\streetaddress{Box 337}
	\city{Uppsala}
	%\state{State1}
	\postcode{SE-751 05}
	\country{Sweden}
}
\email{philipp.ruemmer@it.uu.se}     

%%%%%%%%%%%%%%%%%%%%%%%%%%%%%%%%%%%%%%%%%%%%%%%%%%%%%%%%%%%%%%%%%%%

\author[Z. Wu]{Zhilin Wu}
\orcid{0000-0003-0899-628X}
\affiliation{
	%\position{Position1}
	\department{State Key Laboratory of Computer Science} %% \department is recommended
	\institution{Institute of Software, Chinese Academy of Sciences \& University of Chinese Academy of Sciences} %% \institution is required
	%\streetaddress{Street1 Address1}
	%\city{City1}
	%\state{State1}
	%\postcode{Post-Code1}
	\country{China}
}
\email{wuzl@ios.ac.cn }  

\begin{abstract}
%% some background on regular expressions
Regular expressions are a classical concept in formal language theory.
%, which are expressions built from characters by the operators of concatenation, union, and Kleene star.
%Real-world regular expressions (RWRE) 
%On the other hand, 
Regular expressions in programming languages ({\regexp}) such as JavaScript, feature non-standard semantics of operators 
(e.g. greedy/lazy Kleene star), as well as additional features such as capturing
groups and references.
While symbolic execution of programs containing \regexps{} appeals to string solvers natively supporting important features of \regexp{}, such a string solver is hitherto missing.  
In this paper, we propose the first string theory and string solver that natively provides such support.
%of natively supporting these features of REs in string constraint solving. 
%Our approach relies on a string constraint language including \regexp{}-dependent string functions such as extract and replace(all). 
%
The key idea of our string solver is to introduce a new automata model, called \emph{prioritized streaming string transducers} (PSST), to formalize the semantics of \regexp{}-dependent string functions.
PSSTs combine \emph{priorities}, which have previously been introduced
in prioritized finite-state automata to capture greedy/lazy semantics,
with \emph{string variables} as in streaming string transducers to
model capturing groups. We validate the consistency of the formal semantics with the actual JavaScript
semantics by extensive experiments.
%
%non-standard semantics of regular expression operators can be modeled by priorities and new features of capturing groups and back references can be modeled by string variables.
%the non-standard semantics of regular expression operators can be modeled by priorities and new features of capturing groups and back references can be modeled by string variables.
%Based on PSSTs, we design a decision procedure for string constraints with
%RWREs and provide its implementation.
%% implementation and experiments
%We implement the decision procedure
%and do extensive experiments to evaluate its performance.
%
%We formulate a string constraint language including complex string functions such as extract and replace(all), and 
Furthermore, to solve the string constraints, we show that {\PSST}s enjoy nice closure and algorithmic properties, 
in particular, the regularity-preserving property (i.e., pre-images of regular constraints under {\PSST}s are regular), 
and introduce a sound sequent calculus that 
exploits these properties and 
performs  propagation of regular constraints by means of 
taking post-images or pre-images. %, as well as to measure the performance of our solver., 
%In particular, we show that {\PSST}s preserve regularity, i.e. pre-images of regular constraints under {\PSST}s are regular and can be effectively computed. 
Although the satisfiability of the string constraint language is generally 
undecidable, we show that our approach is complete for the so-called straight-line fragment.
We evaluate the performance of our string solver on over 195\,000 string constraints generated from an open-source {\regexp} library. The experimental results show the efficacy of our approach, drastically improving the existing methods (via symbolic execution) in  both precision and efficiency.

\end{abstract}

%% 2012 ACM Computing Classification System (CSS) concepts
%% Generate at 'http://dl.acm.org/ccs/ccs.cfm'.
 
\begin{CCSXML}
	<ccs2012>
	<concept>
	<concept_id>10003752.10003790.10003794</concept_id>
	<concept_desc>Theory of computation~Automated reasoning</concept_desc>
	<concept_significance>500</concept_significance>
	</concept>
	<concept>
	<concept_id>10003752.10010124.10010138.10010142</concept_id>
	<concept_desc>Theory of computation~Program verification</concept_desc>
	<concept_significance>500</concept_significance>
	</concept>
	<concept>
	<concept_id>10003752.10003766.10003776</concept_id>
	<concept_desc>Theory of computation~Regular languages</concept_desc>
	<concept_significance>500</concept_significance>
	</concept>
	<concept>
	<concept_id>10003752.10003790.10002990</concept_id>
	<concept_desc>Theory of computation~Logic and verification</concept_desc>
	<concept_significance>300</concept_significance>
	</concept>
	<concept>
	<concept_id>10003752.10003777.10003778</concept_id>
	<concept_desc>Theory of computation~Complexity classes</concept_desc>
	<concept_significance>100</concept_significance>
	</concept>
	</ccs2012>
\end{CCSXML}

\ccsdesc[500]{Theory of computation~Automated reasoning}
\ccsdesc[500]{Theory of computation~Program verification}
\ccsdesc[500]{Theory of computation~Regular languages}
\ccsdesc[300]{Theory of computation~Logic and verification}
\ccsdesc[100]{Theory of computation~Complexity classes}

%% End of generated code

%% Keywords
%% comma separated list
\keywords{String Constraint Solving, Regular Experssions, Transducers, Symbolic Execution}  
\maketitle

%%%%%%%%%%%%%%%%%%%%%%%%%%%%%%%%%%%%%%%%%%%%%%%%%%%%%%%%%%%%%%%%%%%%%%%%%%%%%%%%%%
%25 pages, excluding bibliography

%
% Add an explanation for the need for separating the set of accepting states in two parts (as mentioned in the rebuttal).
%
%- Add a discussion on extending PSST to handle lookahead and backreferences (as mentioned in the rebuttal). [Added in Conclusion - MH]
%
%- Expand the complexity analysis (as mentioned in the rebuttal). [Semi-done in "complexity analysis" paragraph - MH]
%
%- Add a discussion on the proof search approach for the general non-SL fragment (as mentioned in the rebuttal).
%
%- Fix the minor errors mentioned in the reviews.  A couple of additional minor errors were found after the initial reviews:

%%%%%%%%%%%%%%%%%%%%%%%%%%%%%%%%%%%%%%%%%%%%%%%%%%%%%%%%%%%%%%%%%%%%%%%%%%%

\section{Introduction}\label{sec-intro}

% general intro on string constraint solving

%
%Strings are among the most important data types. 
In modern programming languages---such as JavaScript, Python, Java, and PHP---the string data type plays a crucial role. 
%, whereby it is commonplace to use
%strings to represent all kinds of data
A quick look at the string libraries for these languages is enough to convince
oneself how well supported string manipulations are in these languages, in that
a wealth of string operations and functions are readily available for the
programmers.
%natively 
%support a wealth of string operations. 
%strings to store and process virtually all kinds of data or code.
Such operations include usual operators like concatenation, length, substring, 
but also complex functions such as 
%the functions dependent on regular expressions 
match, replace, split, and parseInt.
%and character encoding/decoding.  
Unfortunately, it is well-known that string manipulations are error-prone and
could even give rise to
%As a result, string-manipulating 
%in programs are
%subtle  error-prone, and their potential 
security vulnerabilities (e.g.\ cross-site scripting, a.k.a. XSS).
%. A typical example is cross-site
%scripting (XSS), which is among the OWASP Top 10 Application Security
%Risks.
%Regular expressions are widely used in string-manipulating programs. 
One powerful method for identifying such bugs in programs is \emph{symbolic 
execution} (possibly in combination with dynamic analysis), which
analyses symbolic paths in a program by viewing them as constraints %$\phi$, 
whose feasibility is checked by constraint solvers. 
%As far as string data
%types, 
Together with the challenging problem of string analysis,
this interplay between program analysis and constraint solvers has motivated 
the highly active research area of \emph{string solving}, resulting in the
development of numerous string solvers in the last decade or so including
%\cite{???}.
%As a result of the fundamental role of regular expressions in string processing, the state-of-the-art string constraint solvers (e.g.
Z3~\cite{Z3}, CVC4~\cite{cvc4}, Z3-str/2/3/4~\cite{Z3-str,Z3-str2,Z3-str3,BerzishMurphy2021},
 ABC~\cite{ABC}, Norn~\cite{Abdulla14},
Trau~\cite{Z3-trau,AbdullaACDHRR18-trau,Abdulla17}, OSTRICH~\cite{CHL+19}, S2S~\cite{DBLP:conf/aplas/LeH18}, Qzy~\cite{cox2017model}, Stranger~\cite{Stranger}, Sloth~\cite{HJLRV18,AbdullaA+19},
Slog~\cite{fang-yu-circuits}, Slent~\cite{WC+18}, Gecode+S~\cite{DBLP:conf/cpaior/ScottFPS17}, G-Strings~\cite{DBLP:conf/cp/AmadiniGST17}, HAMPI~\cite{HAMPI}, among many others. %\cite{??}
%... \anthony{make sure we add all}) 
 %support regular expressions. Nevertheless, what they support is only the regex-string matching (aka membership constraints) for \emph{classical} regular expressions, i.e. whether a string belongs to the language defined by a regular expression, as we know it from formal language theory.  

One challenging problem in the development of string solvers is the need
to support an increasing number of real-world string functions, especially because the
initial stage of the development of string solvers typically assumed only simple
functions (in particular, concatenation, regular constraints, and sometimes also
length constraints). For example, the importance of supporting functions like
the replaceAll function (i.e.\ replace with global flag) in a string solver was 
elaborated in~\cite{CCH+18};
ever since, quite a number of string solvers support this operator.
Unfortunately, the gap between the string functions that are supported by 
current string solvers and those supported by modern programming languages 
is still too big. As convincingly argued in~\cite{LMK19} in the context of 
constraint solving, 
%performed in programming languages depend on 
the widely used \emph{Regular Expressions} in modern programming
languages
(among others, JavaScript, Python, etc.)---which we call \emph{RegEx} 
in the sequel---are one important and frequently occurring feature in
programs that are difficult for existing SMT theories over 
strings to model and solve, especially because their syntaxes and semantics 
substantially differ from the notion of regular expressions in formal 
language theory~\cite{HU79}. Indeed,
many important string functions in programming languages---such as exec, test,
search, match, replace, and split in JavaScript, as well as match, findall,
search, sub, and split in Python---can and often do exploit RegEx, giving 
rise to path constraints that are difficult (if not impossible) to precisely 
capture in existing string solving frameworks. We illustrate these difficulties
in the following two examples.
\begin{example}\label{exmp-name-swap}
%    We shall give a more extensive example in Section \ref{sec:mot}, which 
%    simultaneously involves both match and replace. 
    We briefly mention the challenges posed by the replace
    function in JavaScript; a slightly different but more detailed example can be found
    in Section~\ref{sec:mot}. Consider the Javascript code snippet
    \begin{minted}{javascript}
        var namesReg = /([A-Za-z]+) ([A-Za-z]+)/g;
        var newAuthorList = authorList.replace(nameReg, "$2, $1");
    \end{minted}
    Assuming \texttt{authorList} is given as a 
    list of \texttt{;}-separated author names --- first name, followed by a last name ---
    the above program would convert this to last name, followed by first name
    format. For instance, \texttt{"Don Knuth; Alan Turing"} would
    be converted to \texttt{"Knuth, Don; Turing, Alan"}.
    A natural post condition for this code snippet one would like to check is the existence of at least one ``,'' between two occurrences of ``;''.
\OMIT
{   
    Suppose that the number of people in \texttt{authorList} and 
   \texttt{newAuthorList} is capped to 4. In this case, we would want to
   check if the string \texttt{newAuthorList} could contain
   more than 4 \texttt{';'}, assuming that \texttt{authorList} contains at most
   than 4 \texttt{';'}.
}
    %\qed
\end{example}

\begin{example}\label{ex:normalize}
    We consider the match function in JavaScript,
    in combination with replace. %; this example will be discussed in more
    %detail in Section \ref{sec:mot}. 
    Consider the code snippet in Figure~\ref{fig-run-exmp-normalize}.
 The function {\tt normalize}   removes leading and trailing zeros from a decimal string with the input %a string variable 
{\tt decimal}. For instance, 
%we get results
 \texttt{normalize("0.250") == "0.25"},
 \texttt{normalize("02.50") == "2.5"},
 \texttt{normalize("025.0") == "25"},
and  finally \texttt{normalize("0250") == "250"}. As the reader might have
    guessed, the function match actually returns an array of strings,
    corresponding to those that are matched in the \emph{capturing groups} (two 
    in our example) in the RegEx using the \emph{greedy} semantics of the Kleene star/plus operator. One might be
    interested in checking, for instance, that there is a way to generate a
    the string \texttt{"0.0007"}, but not the string \texttt{"00.007"}.
\end{example}

\begin{figure}[htbp]
\begin{center}
\begin{minted}[linenos]{javascript}
function normalize(decimal) {
  const decimalReg = /^(\d+)\.?(\d*)$/;
  var   decomp     = decimal.match(decimalReg);
  var   result     = "";
  if (decomp) {
    var integer    = decomp[1].replace(/^0+/, "");
    var fractional = decomp[2].replace(/0+$/, "");
    if (integer    !== "") result = integer; else result = "0"; 
    if (fractional !== "") result = result + "." + fractional;
  }
  return result;
}
\end{minted}
\end{center}
\caption{Normalize a decimal by removing the leading and trailing zeros
    \label{fig-run-exmp-normalize}}
\end{figure}

The above examples epitomize the difficulties that have arisen from the
interaction between RegEx and string functions in programs. Firstly,
RegEx uses deterministic semantics for pattern matching (like greedy
semantics in the above example, but the so-called \emph{lazy} matching is
also possible), and allows features that do not exist in regular expressions in
formal language theory, e.g., capturing groups (those in brackets) in the above
example. Secondly, string functions in programs can exploit RegEx in an
intricate manner, e.g., by means of references \$1 and \$2 in Example~\ref{exmp-name-swap}. Hitherto, no existing string solvers
can support any of these features. This is despite the fact that \emph{idealized
versions} of regular 
constraints and the replace functions are allowed in modern string solvers
(e.g.\ see~\cite{AbdullaACDHRR18-trau,HJLRV18,cvc4,TCJ16,YABI14,CHL+19}), i.e., 
features that can be found in the above examples like capturing groups, 
greedy/lazy matching, and references are not supported.
%in that the semantics of regular expressions\philipp{not sure what this is supposed to tell} is used and that 
%RegEx features like capture groups are not allowed. 
This limitation of existing
string solvers was already mentioned in the recent paper~\cite{LMK19}.

In view of the aforementioned limitation of string solvers, what solutions are
possible? One recently proposed solution is to map the path constraints
generated by string-manipulating programs that exploit RegEx into constraints
in the SMT theories
supported by existing string solvers. In fact, this was done in recent papers~\cite{LMK19}, where the path constraints are mapped to constraints
in the theory of strings with concatenation and regular constraints in Z3~\cite{Z3}. Unfortunately, this mapping is an \emph{approximation}, since
such complex string manipulations are generally \emph{inexpressible} in any 
string theories supported by existing string solvers.
%
%\philipp{too general, just say that they are not expressible in the SMT-LIB theory?}. 
To leverage this, CEGAR (counter-example guided 
abstraction and refinement) is used in~\cite{LMK19}, while ensuring that an
\emph{under-approximation} is preserved. 
%CEGAR (counter-example guided abstraction and refinement). 
This results in a rather severe price in both precision and performance: the
refinement process may not terminate even for extremely simple programs
(e.g.\ the above examples).
%\anthony{Zhilin: please check this}\zhilin{The description is consistent with the experimental results.}

Therefore, the current state-of-affairs is unsatisfactory because even 
the introduction of very simple RegEx expressions in programs (e.g.\ the above 
examples) results in path constraints that can \emph{not} be solved by existing 
symbolic executions in combination with string solvers. In this paper, we would
like to firstly advocate that string solvers should \emph{natively} support
important features of \regexp{}
%\philipp{should be reformulated. current solvers support RegEx \emph{on some level}, namely excluding look-arounds, capture groups, etc.}) 
in their SMT theories. 
Existing work
(e.g.\ the reduction to Z3 provided by~\cite{LMK19}) shows that this is a 
monumental theoretical and programming task, 
not to mention the loss in precision and the performance penalty. Secondly, we 
present \emph{the first} string
theory and string solver that natively provide such a support.

\paragraph*{Contributions.}
In this paper, we provide \emph{the first} string theory and string
solver that natively support \regexp{}. Not only can our theory/solver
easily express and solve Example~\ref{exmp-name-swap} and
Example~\ref{ex:normalize} --- which hitherto no existing string
solvers and string analysis can handle --- our experiments using a
library of 98,117~real-world regular expressions indicate that our
solver substantially outperforms the existing method~\cite{LMK19} in
terms of the number of solved problems and runtime.
%We provide 
%more details of our contributions
%\philipp{agreed. we should also say what parameter 30x-50x refers to} 
We provide more details of our contributions below.

Our string theory provides for the first time a native support of the match and the 
replace functions, which use JavaScript\footnote{JavaScript was chosen 
because it is relevant to string solving~\cite{BEK,Berkeley-JavaScript}, due to vulnerabilities in JavaScripts 
caused by string manipulations. Our method can be easily adapted to \regexp{}
semantics in other languages.} RegEx in the input arguments. Here is a quick
summary of our string constraint language (see Section~\ref{sec:logic} for
more details):
\[
\begin{array}{l c l}
\smallskip
\varphi & \eqdef  & x = y \mid z = x \concat y \mid y  = \extract_{i, e}(x) \mid
y  = \replace_{\pat, \rep}(x) \mid 
\\
& & y = \replaceall_{\pat, \rep}(x)   \mid
 x \in e \mid \varphi \wedge \varphi \mid \varphi \vee \varphi \mid \neg \varphi \
\label{eq:SL-intro}
\end{array}
\]
where $e, \pat$ are \regexps, $i \in \mathbb{N}$, $x,y,z$ are variables, and $\rep$ 
is called the
replacement string and might refer to strings matched in capturing groups,
as in Example~\ref{exmp-name-swap}. Apart from the standard concatenation
operator $\concat$, we support $\extract$, which extracts the string matched by
the $i$th capturing group in the \regexp{} $e$ (note that match can be simulated
by several calls to $\extract$). We also support $\replace$ 
(resp.~$\replaceall$), which replaces the first occurrence (resp.~all
%<<<<<<< HEAD
occurrences) of substrings in $x$ matched by $\pat$ by $\rep$. Our solver/theory
also covers the most important features of \regexp{} (including greedy/lazy
matching, capturing groups, among others) that make up 74.97\% of the \regexp{}
expressions of~\cite{LMK19} across 415,487 NPM packages. 

A crucial step in the development of our string solver is a formalization of
the semantics of the $\extract$, $\replace$, and $\replaceall$ functions in
an automata-theoretic model that is amenable to analysis (among others, closure
properties; see below).
To this end, we introduce a new 
transducer model called \emph{Prioritized Streaming String Transducers (PSSTs)},
which is inspired by %extends and combines 
two automata/transducer models: prioritized finite-state automata~\cite{BM17} 
and streaming string transducers~\cite{AC10,AD11}. PSSTs allow us to precisely
capture the non-standard semantics of \regexp{} operators (e.g. greedy/lazy Kleene star) by priorities and 
deal with capturing groups by string variables. 
We show that $\extract$, $\replace$, and $\replaceall$ can all be expressed as 
PSSTs. More importantly, we have performed an extensive experiment
validating our formalization against JavaScript semantics. 
%\anthony{Say
%details here. Maybe good to say which ECMAScript}.\zhilin{How to define the Regex semantics by PSSTs deserves a place in the intro.}

Next, by means of a sound sequent calculus, our string solver (implemented in 
the standard DPLL(T) setting of SMT solvers~\cite{NieuwenhuisetalJACM2006}) 
will exploit crucial closure and algorithmic properties satisfied by PSSTs.
%Our calculus is sound for the
In particular, the solver attempts to
(1) \emph{propagate} regular constraints (i.e.\ the constraints $x \in e$) in the formula around by means of
the string functions $\cdot$, $\replace$, $\replaceall$, and $\extract$, and 
(2) either detect conflicting regular constraints, or find a satisfiable assignment.
%%%%%%%%%%%%%%%%%%%%%%%%%%%
\OMIT{
Here, a \regexp{} constraint on a set $S$ of variables simply refers to a 
boolean
combination of constraints of the form $x \in \Lang(e)$, for a variable 
$x \in S$ and \regexp{} $e$.
}
%%%%%%%%%%%%%%%%%%%%%%%%%%%
A single step of the regular-constraint propagation computes either the 
$post$-image or the $pre$-image of the above functions. In particular, 
%In order to perform a single step of this constraint propagation, we perform
it is crucial that each step of our constraint propagation preserves
regularity of the constraints.
Since the $post$-image does not always preserve regularity,
we only propagate by taking $post$-image when regularity is preserved.
On the other hand, one of our crucial results is that taking $pre$-image 
always preserves regularity:
%%%%%%%%%%%%%%%%%%%%%%%%%%%%%%%%%%%%%%
\OMIT{
To this end, we introduce a new 
transducer model called \emph{Prioritized Streaming String Transducer (PSST)},
which is inspired by %extends and combines 
two automata/transducer models: prioritized finite-state automata~\cite{BM17} 
and streaming string transducers~\cite{AC10,AD11}. PSSTs allow us to precisely
capture the non-standard semantics of \regexp{} operators by priorities and 
deal with capturing groups by string variables. 
We show that extract and replace functions can be expressed in terms of PSSTs.
}
%%%%%%%%%%%%%%%%%%%%%%%%%%%%%%%%%%%%%%
regular constraints are \emph{effectively closed under
taking $pre$-image of functions captured in PSSTs}.
Finally, despite the fact that our above string theory is undecidable
(which follows from~\cite{LB16}), we show that our string solving algorithm is
guaranteed to terminate (and therefore is also complete) under the assumption
that the input formula syntactically satisfies the so-called 
\emph{straight-line restriction}.

We implement our decision procedure 
%in a new solver \ostrich\  
on top of the open-source solver~OSTRICH~\cite{CHL+19}, and carry out
extensive experiments to evaluate the performance. For the benchmarks,
we generate two collections of JavaScript programs (with 98,117
programs in each collection), from a library of real-world regular
expressions~\cite{DMC+19}, by using two simple JavaScript program
templates containing match and replace functions, respectively.  Then
we generate all the four (resp.\ three) path constraints for each
match (resp.\ replace) JavaScript program and put them into one
SMT-LIB script. {\ostrich} is able to answer all four (resp.\ three)
queries in 97\% (resp.\ 91.5\%) of the match (resp.\ replace) scripts,
with the average time 1.57s (resp.\ 6.62s) per file.  Running
\expose{}~\cite{LMK19} with the same time budget on the same
benchmarks, we show that \ostrich{} offers a 8x--18x speedup in
comparison to \expose{}, while being able to cover substantially more
paths (9.6\% more for match, 49.9\% more for replace), making \ostrich{} the first string solver that is able to
handle \regexps{} precisely and efficiently.
%\philipp{this should be reformulated; earlier we claim
%that we have the only string solver supporting RegEx, so it does not make
%sense to say we have the fastest now.}\zhilin{the footnote is dangerous in the sense that we are claiming that we are the first solver supporting Regex and then we are saying that Expose can act as a string solver supporting Regex.}

\paragraph{Organization.} In Section~\ref{sec:mot}, more details of Example~\ref{exmp-name-swap} are worked out to illustrate our approach. The string constraint language supporting {\regexps} is presented in Section~\ref{sec:logic}. The semantics of the {\regexp}-dependent string functions are formally defined via PSSTs in Section~\ref{sec:semantics}. The sequent calculus for solving the string constraints is introduced in Section~\ref{sect:calculus}. The implementation of the string solver and experiments are described in Section~\ref{sect:impl}. The related work is given in Section~\ref{sec-related}. Finally, Section~\ref{sec-conc} concludes this paper.

\section{A Detailed Example}\label{sec:mot}

In this section, we provide a detailed example to illustrate our string
solving method. 
Consider the JavaScript program in Figure~\ref{fig-run-exmp}; this example is
similar to Example \ref{exmp-name-swap} from the Introduction. The function ``authorNameDBLPtoACM'' in 
Figure~\ref{fig-run-exmp} transforms %the name format of 
an author list in the DBLP BibTeX style to the one in the ACM BibTeX style. For instance,  if a paper is authored by Alice M. Brown and John Smith, then the author list in the DBLP BibTeX style is ``Alice M. Brown and John Smith'', while it is ``Brown, Alice M. and Smith, John'' in the ACM BibTeX style. 

\begin{figure*}[tb]
%\vspace{-2mm}
\begin{center}
\small
\begin{minted}{javascript}
function authorNameDBLPtoACM(authorList)
{
  var autListReg 
    = /^[A-Z](\w*|.)(\s[A-Z](\w*|.))*(\sand\s[A-Z](\w*|.)(\s[A-Z](\w*|.))*)*$/;
  if (autListReg.test(authorList)) {
    var nameReg = /([A-Z](?:\w*|.)(?:\s[A-Z](?:\w*|.))*)(\s[A-Z](?:\w*|.))/g;
    return authorList.replace(nameReg, "$2, $1");
  }
  else return authorList;
}
\end{minted}
\end{center}
%\vspace{-4mm}
\caption{Change the author list from the DBLP format to the ACM format\label{fig-run-exmp}}
%\vspace{-2mm}
\end{figure*}

%%%%%%%%%%%%%%%%%%%%%%%%%%%%%%%%%%%%%%%%%%%%%
%%%%%%%%%%%%%%%%%%%%%%%%%%%%%%%%%%%%%%%%%%%%%
%\OMIT{
%\begin{figure*}[htbp]
%\begin{center}
%{
%%\small
%\begin{minted}[linenos]{javascript}
%    function changeNameFormat(authorList)
%    {
%      var autListReg = /^([\w-\.\s]+,[\w-\.\s]+\sand\s)*[\w-\.\s]+,[\w-\.\s]+$/;
%      if (autListReg.test(authorList)) {
%        var nameReg = /([\w-\.]+(?:\s+[\w-\.]+)*),\s*([\w-\.]+(?:\s+[\w-\.]+)*)((\s+and\s+)|$)/g;
%        return authorList.replace(nameReg, "$1$3");
%      }
%      else return authorList;
%    }
%\end{minted}
%}
%\end{center}
%\caption{Change the name format of an author list: A motivating example}
%\label{fig-run-exmp}
%\end{figure*}
%}
%%%%%%%%%%%%%%%%%%%%%%%%%%%%%%%%%%%%%%%%%%%%%
%%%%%%%%%%%%%%%%%%%%%%%%%%%%%%%%%%%%%%%%%%%%%
 
The input of the function ``authorNameDBLPtoACM'' is {\sf authorList}, which is expected to follow the pattern specified by the regular expression {\sf autListReg}. Intuitively, {\sf autListReg} stipulates that {\sf authorList} %can be obtained by joining 
%is the concatenation of  
joins the strings of full names %of the following structure a full-name string is 
as a concatenation of a given name, middle names, and a family name, separated by the blank symbol (denoted by $\backslash$s). Each of the given, middle, family names is a concatenation of a capital alphabetic letter (denoted by [A-Z]) followed by a sequence of letters (denoted by $\backslash$w) or a dot symbol (denoted by $.$). Between names, the word ``and'' is used as the separator.
%of alphabetic letters or  (denoted by $\backslash$w), bar (denoted by $-$), dot (denoted by $.$), or the blank symbol (denoted by $\backslash$s), with the comma in between. 
The symbols \^{} and $\$$ denote the beginning and the end of a string input respectively.

The DBLP name format of each author is specified by the regular expression {\sf nameReg}  in Figure~\ref{fig-run-exmp}, which describes the format of a full name.
% of a given name, middle names, and a family name, separated by the blank symbol.
%a family name (a sequence of alphabetic letters, $-$ or $.$), followed by the comma, then a given name, finally the word ``and'' or $\$$ denoting the end of the input, where $\backslash s$ represents the blank symbol and $\backslash w$ represents alphabetic letters, digits, or the underline symbol $\_$. 
\begin{itemize}
\item There are two capturing groups in {\sf nameReg}, one for recording the concatenation of the given name and middle names, and the other for recording the family name. 
Note that the symbols ?: in (?:$\backslash$s[A-Z](?:$\backslash$w*|.)) denote the non-capturing groups, i.e. matching the subexpression, but not remembering the match.

\item The \emph{greedy} semantics of the Kleene star * is utilized here to guarantee that the subexpression (?:$\backslash$s[A-Z](?:$\backslash$w*|$\backslash$.))* matches all the middle names (since there may exist multiple middle names) and thus ${\sf nameReg}$ matches the full name. For instance, the first match of ${\sf nameReg}$ in  ``Alice M. Brown and John Smith'' is ``Alice M. Brown'', instead of ``Alice M.''. In comparison, if the semantics of * is assumed to be non-greedy, then (?:$\backslash$s[A-Z](?:$\backslash$w*|$\backslash$.))* can be matched to the empty string, thus ${\sf nameReg}$ is matched to ``Alice M.'', which is \emph{not} what we want. Therefore, the greedy semantics of * is essential for the correctness of ``authorNameDBLPtoACM''.
%=======
%\item The \emph{greedy} semantics of the Kleene star * is utilized here to guarantee that the subexpression (?:$\backslash$s[A-Z](?:$\backslash$w*|.))* matches all the middle names (since there may exist multiple middle names) and thus ${\sf nameReg}$ matches the full name. For instance, the first match of ${\sf nameReg}$ in  ``Alice M. Brown and John Smith'' is ``Alice M. Brown'', instead of ``Alice M.''. In comparison, if the semantics of * is assumed to be non-greedy, then (?:$\backslash$s[A-Z](?:$\backslash$w*|.))* can be matched to the empty string, thus ${\sf nameReg}$ is matched to ``Alice M.'', which is what we would like to avoid. Therefore, the greedy semantics of * is essential for the correctness of ``authorNameDBLPtoACM''.
%>>>>>>> c103b21aed01d0787e4815264e14638bb78265a9
%Similarly, the regular expression {\sf reLast} describes the name format of the last author, except that it replaces the word ``and'' by the symbol $\$$, denoting the end of the author list. 
%
\item The global flag ``g'' is used in {\sf nameReg} so that the name format of each author is transformed. 
\end{itemize}
The name format transformation is via the {\sf replace} function, i.e. {\sf authorList.replace(nameReg, ``$\$$2, $\$$1'')}, where $\$1$ and $\$2$ refer to the match of the first and second capturing group respectively. 

A natural post-condition of {\sf authorNameDBLPtoACM} is that there exists at least one occurrence of the comma symbol between every two occurrences of ``$\sf and$''. 
This post-condition has to be established by the function on \emph{every} execution path.  As an example, consider the path shown in Fig.~\ref{fig-run-exmp-path},
in which the branches taken in the program are represented as
\texttt{assume} statements. The negated post-condition is enforced by
the regular expression in the last \texttt{assume}. For this path, the
post-condition can be proved by showing that the program in
Fig.~\ref{fig-run-exmp-path} is infeasible: there does not exist an
initial value {\sf authorList} so that no assumption fails and the
program executes to the end.

\begin{figure}[tb]
%\vspace{-2mm}
\begin{center}
\small
\begin{minted}[linenos]{javascript}
  var autListReg = 
      /^[A-Z](\w*|.)(\s[A-Z](\w*|.))*(\sand\s[A-Z](\w*|.)(\s[A-Z](\w*|.))*)*$/;
  assume(autListReg.test(authorList));
  var nameReg = /([A-Z](?:\w*|.)(?:\s[A-Z](?:\w*|.))*)(\s[A-Z](?:\w*|.))/g;
  var result = authorList.replace(nameReg, "$2, $1");
  assume(/\sand[^,]*\sand/.test(result));
\end{minted}
\end{center}
%\vspace{-2mm}
\caption{Symbolic execution of a path of the JavaScript program in Fig.~\ref{fig-run-exmp}}
\label{fig-run-exmp-path}
%\vspace{-4mm}
\end{figure}

To enable symbolic execution of the JavaScript programs like in Fig.~\ref{fig-run-exmp-path}, one needs to model both the greedy semantics of the Kleene star and store the matches of capturing groups. For this purpose, we introduce prioritized streaming string transducers (PSST, cf.\ Section~\ref{sec:semantics}) by which {\sf replace(nameReg, ``$\$$2, $\$$1'')} is represented as a PSST $\cT$, where the \emph{priorities} are used to model the greedy semantics of $*$ and the \emph{string variables} are used to record the matches of the capturing groups as well as the return value. Then the  symbolic execution of the program in Fig.~\ref{fig-run-exmp-path} can be equivalently turned into the satisfiability of the following string constraint,
\begin{equation} \label{eq:motiv}
{\sf authorList} \in {\sf autListReg}\wedge {\sf result} = \cT({\sf authorList}) \wedge {\sf result} \in {\sf postConReg},
\end{equation}
where {\sf postConReg} =
\mintinline{javascript}{/^.*\sand[^,]*\sand.*$/}, and {\sf autListReg}
is as in Fig.~\ref{fig-run-exmp}.

Our solver is able to show that \eqref{eq:motiv} is unsatisfiable. On
the calculus level (introduced in more details in
Section~\ref{sect:calculus}), the main inference step applied for this
purpose is the computation of the \emph{pre-image} of {\sf postConReg}
under the function
$\cT$; in other words, we compute the language of all strings that are
mapped to incorrect strings (containing two ``\textsf{and}''s without
a comma in between) by
$\cT$. This inference step relies on the fact that the pre-images of
regular languages under PSSTs are regular (see
Lemma~\ref{lem:psst_preimage}). Denoting the pre-image of {\sf
  postConReg}  by
$\cB$, formula~\eqref{eq:motiv} is therefore equivalent to
\begin{equation}
  \label{eq:motiv2}
  {\sf authorList} \in \cB \wedge
{\sf authorList} \in {\sf autListReg}\wedge {\sf result} = \cT({\sf authorList}) \wedge {\sf result} \in {\sf postConReg}.
\end{equation}

To show that this formula (and thus \eqref{eq:motiv}) is
unsatisfiable, it is now enough to prove that the languages defined by
$\cB$ and {\sf autListReg} are disjoint.

%\input{example2}

%%%%%%%%%%%%%%%%%%%%%%%%%%%% 

\section{A String Constraint Language Natively Supporting \regexp}
\label{sec:logic}
In this section, we define a string constraint language natively supporting {\regexp}. Throughout the paper, $\Int^+$ denotes the set of positive integers, and  $\nat$ denotes the set of natural numbers. Furthermore, for $n\in \Int^+$, let $[n]:=\{1, \ldots, n\}$. 
We use $\Sigma$ to denote a finite set of letters, called \emph{alphabet}. A \emph{string} over $\Sigma$ is a finite sequence of letters from $\Sigma$. We use $\Sigma^*$ to denote the set of strings over $\Sigma$, $\varepsilon$ to denote the empty string, and $\Sigma^\varepsilon$ to denote $\Sigma \cup \{\varepsilon\}$. 
A string $w'$ is called a \emph{prefix} (resp. \emph{suffix}) of $w$ if $w = w'w''$ (resp. $w = w'' w'$) for some string $w''$.

%%%%%%%%%%%%%%%%%%%%%%%%%%%%%%%%%%%%%%%%%
%\hide{
%	\tl{this part can be moved to intro?}
%	Regular expressions are a well-known concept in formal language. %and  have the same expressibility as finite state automata. 
%	Many programming languages provide build-in regular expressions %capabilities either built-in 
%	or otherwise via libraries. Programmers widely use regular expressions in software development, especially in the development of web applications. However, it should be emphasized that regular expressions used in programming languages are considerably different from those in formal language theory, mainly on the following aspects: greedy/non-greedy semantics of the quantifiers ($*$ and its variant $+$), non-commutativity of the alternation operator, capturing groups, and backreferences. In the sequel, we take all these aspects into account and define the class of real-world regular expressions considered in this paper. 
%}
%%%%%%%%%%%%%%%%%%%%%%%%%%%%%%%%%%%%%%%%%  

%\subsection{Syntax of regular expressions}
We start with the syntax of {\regexp} %tailored to its usage in modern programming language. We shall start with the syntax of regular expressions, 
which is essentially that used in JavaScript. (We do not include backreferences though.)
\begin{definition}[Regular expressions, {\regexp}]	
	%\begin{multline*}
	\[
	\begin{split}
		e & \eqdef  \emptyset \mid \varepsilon \mid a \mid  (e) \mid %\$n \mid 
		[e + e] \mid [e \concat e] \mid [e^?] \mid [e^{??}] \mid  \\
		&          [e^*]  \mid [e^{*?}] \mid [e^+] \mid  [e^{+?}] \mid [e^{\{m_1,m_2\}}] \mid [e^{\{m_1,m_2\}?}] 
	\end{split}
	\]
	%\end{multline*}
	%	
	where $a \in \Sigma$,  $n \in \Int^+$, $m_1,m_2 \in \Nat$ with $m_1 \le m_2$. 
	%	Since $+$ is associative and commutative, we also write $(e_1 + e_2) + e_3$ as $e_1 + e_2 + e_3$ for brevity.  
\end{definition}
%We abbreviate $[e \concat [e^*]]$ as $[e^+]$ and $[e \concat [e^{*?}]]$ as $[e^{+?}]$. 
%
For $\Gamma = \{a_1, \ldots, a_k\}\subseteq \Sigma$, we write $\Gamma$ for  $[[\cdots [a_1 + a_2] + \cdots] + a_k]$ and thus $[\Gamma^\ast] \equiv [[[\cdots [a_1 + a_2] + \cdots] + a_k]^\ast]$. Similarly for $[\Gamma^{\ast?}]$, $[\Gamma^+]$, and $[\Gamma^{+?}]$. We write $|e|$ for the length of $e$, i.e., the number of symbols occurring in $e$.
Note that square brackets $[]$ are used for the operator precedence and the parentheses $()$ are used for \emph{capturing groups}. 
%
%Parenthesis pairs are indexed according to the occurrence sequence of their left parentheses, and it is required that every back reference $\$ n$ occurs after the $n$-th pair of parentheses. For instance, $[[([[a+b]^*]) \concat c] \concat \$1]$ is in {\regexp}, where $\$1$ refers to the matching of the subexpression $[[a+b]^*]$. Intuitively, it denotes the set of strings of the form $u c u$, where $u$ is a string of $a$ and $b$. 
%

The operator $[e^*]$ is the \emph{greedy} Kleene star, meaning that $e$ should be matched as many times as possible. In contrast, the operator $[e^{*?}]$ is the \emph{lazy} Kleene star, meaning $e$ should be matched  as few times as possible. The Kleene plus operators $[e^+]$ and $[e^{+?}]$ are similar to $[e^*]$ and $[e^{*?}]$ but $e$ should be matched at least once. Moreover, as expected,  the repetition operators $[e^{\{m_1,m_2\}}]$ require the number of times that $e$ is matched is between $m_1$ and $m_2$ and $[e^{\{m_1,m_2\}?}]$ is the lazy variant. Likewise, the optional operator has greedy and lazy variants $[e^?]$ and $[e^{??}]$, respectively. 

For two {\regexp} $e$ and $e'$, we say that $e'$ is a \emph{subexpression} of $e$,
if one of the following conditions holds: 1) $e'=e$, 2) $e = [e_1 \cdot e_2]$ or $[e_1 + e_2]$, and $e'$ is a subexpression of $e_1$ or $e_2$, 3) $e = [e_1^?], [e_1^{??}], [e_1^{\ast}]$, $[e_1^{+}]$, $[e_1^{\ast?}]$, $[e_1^{+?}]$, $[e_1^{\{m_1, m_2\}}]$, $[e_1^{\{m_1, m_2\}?}]$ or $( e_1)$, and $e'$ is a subexpression of $e_1$. We use $S(e)$ to denote the set of subexpressions of $e$. %

% 
%We use $\cgexp$ to denote the fragment of {\regexp} excluding backreferences $\$ n$ (where {\sf reg} represents regular languages), and $\refexp$ to denote the set of regular expressions generated by a concatenation of letters and backreferences, formally %regular expressions 
%defined by $e \eqdef \varepsilon \mid a \mid \$n \mid [e \concat e]$.  
%%\tl{define the semantics here?}

We shall formalize the semantics of \regexp, in particular, for a given regular expression and an input string, how the string is matched against the regular expression, in Section~\ref{sect:regextopsst}.

In the rest of this section, we define the string constraint language $\strline$. 

The syntax of $\strline$ is defined by the following rules.
\[
\begin{array}{l c l}
\smallskip
\varphi & \eqdef  & x = y \mid z = x \concat y \mid y  = \extract_{i, e}(x) \mid
y  = \replace_{\pat, \rep}(x) \mid 
\\
& & y = \replaceall_{\pat, \rep}(x)   \mid
 x \in e \mid  \varphi \wedge \varphi \mid \varphi \vee \varphi \mid \neg \varphi \
\label{eq:SL}
\end{array}
\]
where
\begin{itemize}
	\item $\concat$ is the string concatenation operation which concatenates two strings,
\item  $e \in${\regexp} and $\pat \in${\regexp},
\item for the $\extract$ function, $i \in \Nat$,
	\item  for the $\replace$ and $\replaceall$ operation, $\rep \in \refexp$, where $\refexp$ is defined as a concatenation of letters from $\Sigma$, the references $\$i$ ($i \in \Nat$), as well as $\refbefore$ and $\refafter$. (Intuitively, $\$0$ denotes the matching of $\pat$, $\$i$ with $i > 0$ denotes the matching of the $i$-th capturing group, $\refbefore$ and $\refafter$ denote the prefix before resp. suffix after the matching of $\pat$.)
%
%	\item for assertions, $e \in \regexp$.\philipp{I believe we should remove the word ``assert'', it is easier/more readable to consider constraints like $x \in e \wedge y \in e' \wedge ...$}
\end{itemize}

The $\extract_{i, e}(x)$ function extracts the match of the $i$-th capturing group in the successful match of $e$ to $x$ for $x \in \Lang(e)$ (otherwise, the return value of the function is undefined). Note that $\extract_{i, e}(x)$ returns $x$ if $i=0$. Moreover, if the $i$-th capturing group of $e$ is \emph{not} matched, even if $x \in \Lang(e)$, then $\extract_{i, e}(x)$ returns a special symbol $\nullchar$, denoting the fact that its value is undefined. For instance, when $[[a^+] + ([a^*])]$ is matched to the string $aa$, $[a^+]$, instead of $([a^*])$, will be matched, since $[a^+]$ precedes $([a^*])$. Therefore, $\extract_{1, [[a^+] + ([a^*])]}(aa) = \nullchar$. 

%by the {\PSST} $\cT_e$ constructed from $e$ , with all the string variables, except the string variable corresponding to the $i$-th capturing group, removed.
%
%\zhilin{to be consistent with the definition of semantics of regex-string matching}\philipp{I'm for removing the special case
%  $\nullchar$, just use the empty string instead.}\zhilin{for me, $\nullchar$ is important to guarantee the consistency with Javascript semantics. }

\begin{remark}
The match function in programming languages, e.g. $\sf str.match(reg)$ function in JavaScript, finds the first match of $\sf reg$ in $\sf str$,  assuming that $\sf reg$ does not contain the global flag. We can use $\extract$ to express the first match of $\sf reg$ in $\sf str$ by adding $[\Sigma^{*?}]$ and $[\Sigma^*]$ before and after $\sf reg$ respectively. More generally, the value of the $i$-th capturing group in the first match of a $\regexp$ $\sf reg$ in $\sf str$ can be specified as $\extract_{i+1, {\sf reg'}}({\sf str})$, where ${\sf reg'} = [[[\Sigma^{*?}] \concat ({\sf reg})] \concat [\Sigma^*]]$. The other string functions involving regular expressions, e.g. {\sf exec} and {\sf test}, without global flags, are similar to {\sf match}, thus can be encoded by $\extract$ as well.
\end{remark}

The function $\replaceall_{\pat, \rep}(x)$ is parameterized by the \emph{pattern} $\pat \in \regexp$ and the \emph{replacement string} $\rep \in \refexp$.
For an input string $x$, it identifies all matches of $\pat$ in $x$ and replaces them with strings specified by $\rep$. More specifically, $\replaceall_{\pat, \rep}(x)$ finds the first match of $\pat$ in $x$ and replaces the match with $\rep$, let $x'$ be the suffix of $x$ after the first match of $\pat$,  then it finds the first match of $\pat$ in $x'$ and replace the match with $\rep$, and so on.
%The set $\refexp$ of replacement strings is defined using the following syntax, where $i \in \Nat$.
%\[
%    \rep = a \mid \$i \mid \refbefore \mid \refafter
%\]
%That is $\rep$ is a string of characters that may also contain \emph{references}.
%A reference $\$i$ where $i > 0$ is instantiated by the string matched by the $i$th capture group.
%For instance, let $w = 2.5, 3.4$, $\pat = [[([\Gamma^+])\concat .?] \concat ([\Gamma^*])]$ and $\rep = \$1$, then $\replaceall_{\pat, \rep}(w) = 2, 3$.
%
A reference $\$i$ where $i > 0$ is instantiated by the matching of the $i$-th capturing group.
There are three special references\footnote{
    The corresponding syntax for $\$0$, $\refbefore$ and $\refafter$ in JavaScript are $\$\&$, $\$`$, and $\$'$.
} $\$0$, $\refbefore$, and $\refafter$.
These are instantiated by the matched text, the text occurring before the match, and the text occurring after the match respectively.
In particular, if the input word is $u v w$ where $v$ has been matched and will be replaced, then $\$0$ takes the value $v$, $\refbefore$ takes the value $u$, and $\refafter$ takes the value $w$.
When there are multiple matches in a $\replaceall$, the values of $\refbefore$ and $\refafter$ are always with respect to the original input string $x$.

The $\replace_{\pat, \rep}(x)$ function is similar to $\replaceall_{\pat, \rep}(x)$, except that it replaces only the first (leftmost) match of $\pat$.

A $\strline$ formula $\varphi$ is said to be \emph{straight-line}, if 1) it contains neither negation nor disjunction, 2) the equations in $\varphi$ can be ordered into a sequence, say $x_1 = t_1, \ldots, x_n = t_n$, such that $x_1,\ldots, x_n$ are mutually distinct, moreover, for each $i \in [n]$, $x_i$ does \emph{not} occur in $t_1, \ldots, t_{i-1}$. Let $\strlinesl$ denote the set of straight-line $\strline$ formulas.

As a crucial step for solving the string constraints in $\strline$, we shall define the formal semantics of the $\extract$, $\replace$, and $\replaceall$ functions in the next section.

\section{Semantics of string functions via {\PSST}}\label{sec:semantics}

Our goal in this section is to define the formal semantics of the string functions involving {\regexp} used in $\strline$, that is, $\extract$, $\replace$ and $\replaceall$. To this end, we need to first define the semantics of {\regexp}-string matching. One of the key novelties here is to utilize an extension of finite-state automata with transition priorities and string variables, called prioritized streaming string transducers (abbreviated as \PSST). It turns out that {\PSST} provides a convenient means to capture the non-standard semantics of {\regexp} operators and to store the matches of capturing groups in {\regexp}, which paves the way to define the semantics of string functions (and the string constraint language). 

\subsection{Prioritized streaming string transducers (\PSST)}

{\PSST}s can be seen as an extension of finite-state automata with transition priorities and string variables. We first recall the definition of classic finite-state automata.

\begin{definition}[Finite-state Automata] \label{def:nfa}
	A \emph{(nondeterministic) finite-state automaton}
	(\FA{}) over a finite alphabet $\ialphabet$ is a tuple $\Aut =
	(\ialphabet, \controls, q_0, \finals, \transrel)$ where 
	$\controls$ is a finite set of 
	states, $q_0\in \controls$ is
	the initial state, $\finals\subseteq \controls$ is a set of final states, and 
	$\transrel\subseteq \controls \times 
	\ialphabet^\varepsilon \times  \controls$ is the
	transition relation. 
\end{definition}

For an input string $w$, a \emph{run} of $\Aut$ on $w$
%(with $a_0 = \EndLeft$ and $a_{n+1} = \EndRight$)
is a sequence $q_0 a_1 q_1 \ldots a_n q_n$ such that $w = a_1 \cdots a_n$ and $(q_{j-1}, a_{j}, q_{j}) \in
\transrel$ for every $j \in [n]$.
The run is said to be \defn{accepting} if $q_n \in \finals$.
A string $w$ is \defn{accepted} by $\Aut$ if there is an accepting run of
$\Aut$ on $w$. 
%In particular, the empty string $\varepsilon$ is accepted by $\Aut$ if $q_0 \in F$. 
The set of strings accepted by $\Aut$, i.e., the language \defn{recognized} by $\Aut$, is denoted by $\Lang(\Aut)$.
%Since we deal with computational complexity in the sequel, we define
The \defn{size} $|\Aut|$ of $\Aut$ is the cardinality of $\transrel$, the set of transitions.

For a finite set $Q$, let $\overline{Q} = \bigcup_{n\in \Nat}\{ (q_1, \ldots, q_n) \mid \forall i \in [n], q_i \in Q \wedge \forall i,j \in[n], i \neq j \rightarrow q_i \neq q_j \}$. Intuitively, $\overline{Q}$ is the set of sequences of non-repetitive elements from $Q$. In particular, the empty sequence $() \in \overline{Q}$. Note that the length of each sequence from $\overline{Q}$ is bounded by  $| Q |$. For a sequence $P = (q_1, \ldots, q_n) \in \overline{Q}$ and  $q \in Q$, we write $q \in P$ if  $q = q_i$ for some $i \in [n]$. Moreover, for $P_1 = (q_1, \ldots, q_m) \in \overline{Q}$ and $P_2 = (q'_1, \ldots, q'_n) \in \overline{Q}$, we say $P_1 \cap P_2 = \emptyset$ if $\{q_1, \ldots, q_m\} \cap \{q'_1, \ldots, q'_n\} = \emptyset$.

%%%%%%%%%%%%%%%%%%%%%%%%%%%%%%%%%%%%%%%%%%%%
%%%%%%%%%%%%%%%%%%%%%%%%%%%%%%%%%%%%%%%%%%%%
%\OMIT{
%\begin{definition}[Prioritized Streaming String Transducers]
%A \emph{prioritized streaming string transducer} (PSST) is a tuple $\psst = (Q, \Sigma, X, \delta, \tau, E, q_0, F)$, where $Q$ is a
%finite set of states, $\Sigma$ is the input and output alphabet such that $\nullchar \not \in \Sigma$, $X$ is a finite set of variables, $\delta \in Q \times \Sigma \rightarrow \overline{Q}$, $\tau \in Q \rightarrow \overline{Q} \times \overline{Q}$, $E$ is a partial function from $Q \times \Sigma^\varepsilon \times
%  Q$ to $X \rightarrow \{\nullchar\} \cup (X \cup \Sigma)^{\ast}$, i.e. the set of assignments,
%   $q_0 \in Q$ is the initial state, and $F$ is a partial function
%  from $Q$ to $(X \cup \Sigma)^{\ast}$.
%\end{definition}
%}
%%%%%%%%%%%%%%%%%%%%%%%%%%%%%%%%%%%%%%%%%%%%
%%%%%%%%%%%%%%%%%%%%%%%%%%%%%%%%%%%%%%%%%%%%

\begin{definition}[Prioritized Streaming String Transducers]
A \emph{prioritized streaming string transducer} (\PSST) is a tuple $\psst = (Q, \Sigma, X, \delta, \tau, E, q_0, F)$, where 
\begin{itemize}
\item $Q$ is a finite set of states, 
\item $\Sigma$ is the input and output alphabet, 
\item $X$ is a finite set of string variables, 
\item $\delta \in Q \times \Sigma \rightarrow \overline{Q}$ defines the non-$\varepsilon$ transitions as well as their priorities (from highest to lowest),
\item $\tau \in Q \rightarrow \overline{Q} \times \overline{Q}$ such that for every $q \in Q$, if $\tau(q) = (P_1; P_2)$, then $P_1 \cap P_2 = \emptyset$, (Intuitively, $\tau(q)=(P_1; P_2)$ specifies the $\varepsilon$-transitions at $q$, with the intuition that the $\varepsilon$-transitions to the states in $P_1$ (resp. $P_2$) have higher (resp. lower) priorities than the non-$\varepsilon$-transitions out of $q$.)
\item $E$ associates with each transition a string-variable assignment function, i.e., $E$ is partial function from $Q \times \Sigma^\varepsilon \times
  Q$ to $X \rightarrow (X \cup \Sigma)^{\ast}$ such that its domain is the set of tuples $(q, a, q')$ satisfying that either $a \in \Sigma$ and $q' \in \delta(q, a)$ or $a = \varepsilon$ and $q' \in \tau(q)$,
\item  $q_0 \in Q$ is the initial state, and
\item  $F$ is the output function, which is a partial function from $Q$ to $(X \cup \Sigma)^{\ast}$.
\end{itemize}
\end{definition}
For $\tau(q)=(P_1; P_2)$, we will use $\pi_1(\tau(q))$ and $\pi_2(\tau(q))$ to denote $P_1$ and $P_2$ respectively.  
%With slight abuse of notation, we write $q\in (P_1; P_2)$ for $q\in P_1\cup P_2$. 
%Intuitively, $\tau(q)=(P_1; P_2)$ specifies the $\varepsilon$-transitions at $q$, with the intuition that the $\varepsilon$-transitions to the states in $P_1$ (resp. $P_2$) have higher (resp. lower) priorities than the non-$\varepsilon$-transitions out of $q$.
The size of $\psst$, denoted by $|\psst|$, is defined as $\sum \limits_{(q, a, q') \in \dom(E)} \sum \limits_{x \in X} |E((q, a, q'))(x)|$, where $|E((q, a, q'))(x)|$ is the length of $E(q, a, q')(x)$, i.e., the number of symbols from $X \cup \Sigma$ in it. A PSST $\psst $ is said to be \emph{copyless} if for each transition $(q, a, q')$ in $\psst$ and each $x \in X$, $x$ occurs in $(E(q, a, q')(x'))_{x' \in X}$ at most once. A PSST $\psst$ is said to be \emph{copyful} if it is not copyless. For instance, if $X = \{x_1, x_2\}$ and $E(q, a, q')(x_1) = x_1$ and $E(q, a, q')(x_2) = x_1a$ for some transition $(q, a, q')$, then $x_1$ occurs twice in $(E(q, a, q')(x'))_{x' \in X}$, thus $\psst$ is copyful. 

A run of $\psst$ on a string $w$ is a sequence $q_0 a_1 s_1 q_1 \ldots a_m s_m q_m$ such that
\begin{itemize}
%\item $q_m \in F$,
%
\item for each $i \in [m]$, 
\begin{itemize}
\item either $a_i \in \Sigma$, $q_i \in \delta (q_{i-1}, a_i)$, and $s_i = E (q_{i - 1}, a_i, q_i)$, 
\item or $a_i = \varepsilon$, $q_i \in \tau(q_{i-1})$ and $s_i = E (q_{i - 1}, \varepsilon, q_i)$,
\end{itemize}

%\item for every subsequence $q_i a_{i+1} s_{i+1} q_{i+1} \ldots a_{j} s_j q_j$ such that  $i < j$ and $a_{i+1} = \cdots = a_j = \varepsilon$, it holds that $q_i, \ldots, q_j$ are mutually distinct. (Intuitively, loops of $\varepsilon$-transitions are forbidden.) 
\item for every subsequence $q_i a_{i+1} s_{i+1} q_{i+1} \ldots a_{j} s_j q_j$ such that  $i < j$ and $a_{i+1} = \cdots = a_j = \varepsilon$,  it holds that each $\varepsilon$-transition occurs at most once in it, namely, for every $k, l: i \le k < l < j$, $(q_k, q_{k+1}) \neq (q_l, q_{l+1})$.
\end{itemize}
Note that it is possible that $\delta(q, a) = ()$, that is, there is no $a$-transition out of $q$. 
From the assumption that each $\varepsilon$-transition occurs at most once in a sequence of $\varepsilon$-transitions, we deduce that given a string $w$, the length of a run of $\psst$ on $w$, i.e. the number of transitions in it, is $O(|w||\psst|)$. 

%A run of $\psst$ is the sequence $q_0 a_1 s_1 q_1 \ldots a_m s_m q_m$, where $F (q_m)$ is defined and for each $i \in [m], q_i \in \delta (q_{i-1}, a_i)$ and $s_i = E (q_{i - 1}, a_i, q_i)$. 
For any pair of runs $R = q_0 a_1 s_1 \ldots a_m s_m q_m$ and $R' = q_0 a'_1
  s_1' \ldots a'_n s_n' q_n'$ such that $a_1 \ldots a_m = a'_1 \ldots a'_n$, we say that $R$ is of a higher priority over $R'$ if 
\begin{itemize}
	\item either $R'$ is a prefix of $R$ (in this case, the transitions of $R$ after $R'$ are all $\varepsilon$-transitions), 
	\item or there is an index $j$ satisfying one of the following constraints:
	\begin{itemize}
		\item $q_0 a_1 q_1 \ldots q_{j-1} a_j = q_0 a'_1 q'_1 \ldots q'_{j-1} a'_j$, $q_j \neq q'_j$, $a_j \in \Sigma$, and we have that $\delta (q_{j - 1}, a_j) =(\ldots, q_j, \ldots, q_j', \ldots)$,
		\item $q_0 a_1 q_1 \ldots q_{j-1} a_j = q_0 a'_1 q'_1 \ldots q'_{j-1} a'_j$, $q_j \neq q'_j$, $a_j  = \varepsilon$,  and one of the following holds: (i) $\pi_1(\tau(q_{j - 1})) = (\ldots, q_j, \ldots, q_j', \ldots)$, (ii) $\pi_2(\tau(q_{j - 1})) = (\ldots, q_j, \ldots, q_j', \ldots)$, or (iii) $q_j \in \pi_1(\tau(q_{j - 1}))$ and $q'_j \in \pi_2(\tau(q_{j-1}))$, 
		\item $q_0 a_1 q_1 \ldots q_{j-1}  = q_0 a'_1 q'_1 \ldots q'_{j-1} $, $a_j  = \varepsilon$, $a'_j  \in \Sigma$, $q_j \in \pi_1(\tau(q_{j - 1}))$, and $q'_j \in \delta(q_{j-1}, a'_j)$, 
		\item $q_0 a_1 q_1 \ldots q_{j-1}  = q_0 a'_1 q'_1 \ldots q'_{j-1} $, $a_j  \in \Sigma$, $a'_j  = \varepsilon$, $q_j \in \delta(q_{j - 1}, a_j)$, and $q'_j \in \pi_2(\tau(q_{j-1}))$.
	\end{itemize}
\end{itemize}

  % $p \neq p'$ and, for the smallest index $j$ with $q_j \neq q_j'$,
 % $\delta (q_{j - 1}, a_j) = \ldots q_j \ldots q_j' \ldots$
  
An \emph{accepting} run of $\psst$ on $w$ is a run of $\psst$ on $w$, say $R = q_0 a_1 s_1 \ldots a_m s_m q_m$, such that 1) $F(q_m)$ is defined, 2)  $R$ is of the highest priority among those runs satisfying 1). The output of $\psst$ on $w$, denoted by $\psst(w)$, is defined as $\eta_m(F(q_m))$, where $\eta_0(x) = \varepsilon$ for each $x \in X$, and $\eta_{i}(x) = \eta_{i-1}(s_{i}(x))$ for every $1 \le i \le m$ and $x \in X$. Note that here we abuse the notation $\eta_m(F(q_m))$ and $\eta_{i-1}(s_{i}(x))$ by taking a function $\eta$ from $X$ to $\Sigma^*$ as a function from $(X \cup \Sigma)^*$ to $\Sigma^*$, which maps each $x \in X$ to $\eta(x)$ and each $a \in \Sigma$ to $a$. If there is no accepting run of $\psst$ on $w$, then $\psst(w) = \bot$, that is, the output of $\psst$ on $w$ is undefined. The string relation defined by $\psst$, denoted by $\cR_\psst$,  is 
$\{(w, \psst(w)) \mid w \in \Sigma^\ast, \psst(w)  \neq \bot\}$.
%Note that in the definition of $\cR_\psst$ above, the inputs of $\psst$ whose outputs are in $(\Sigma \cup \{\nullchar\})^* \setminus (\Sigma^* \cup \{\nullchar\})$ are ignored.

\begin{example}
The {\PSST} $\cT=(Q, \Sigma, X, \delta, \tau, E,  q_{0}, F)$ to extract the match of the first capturing group for the regular expression \mintinline{javascript}{(\d+)(\d*)} 
%in Example~\ref{exm-plre} 
%
is illustrated in Fig.~\ref{fig-psst-exmp}, where $x_1$ and $x_2$ store the matches of the two capturing groups. More specifically, in $\cT$ we have $\Sigma = \{0,\cdots,9\}$, $X= \{x_1,x_2\}$, $F(q_{4}) = x_1$ denotes the final output, and $\delta, \tau, E$ are illustrated %by the edges 
in Fig.~\ref{fig-psst-exmp}, where the dashed edges denote the $\varepsilon$-transitions of lower priorities than the non-$\varepsilon$-transitions and the symbol $\ell$ denotes the currently scanned input letter. For instance, for the state $q_2$, $\delta(q_2, \ell) = (q_2)$ for $\ell \in \{0,\ldots, 9\}$, $\tau(q_2) = (();(q_3))$, $E(q_2, \ell, q_2)(x_1) = x_1 \ell$, $E(q_2, \ell, q_2)(x_2) = x_2$,  $E(q_2, \varepsilon, q_3)(x_1) = x_1$, and $E(q_2, \varepsilon, q_3)(x_2) = \varepsilon$. Note that the identity assignments, e.g. $E(q_2, \varepsilon, q_3)(x_1) = x_1$, are omitted in Fig.~\ref{fig-psst-exmp} for readability.  For the input string $w$=``2050'', the accepting run of $\cT$ on $w$ %, namely, a run of $\cT$ on $w$ ending at $q_4$ and of the highest priority, 
is 
\[
q_0 \xrightarrow[x_1:=\varepsilon]{\varepsilon} q_1 \xrightarrow[x_1:=x_12]{2} q_2  \xrightarrow[x_1:=x_10]{0} q_2  \xrightarrow[x_1:=x_15]{5} q_2  \xrightarrow[x_1:=x_10]{0} q_2  \xrightarrow[x_2:=\varepsilon]{\varepsilon} q_3  \xrightarrow{\varepsilon} q_4,
\]
where the value of $x_1$ and $x_2$ when reaching the state $q_4$ are ``2050'' and $\varepsilon$ respectively. 
%From $\delta(q_4, \backslash s) = q_5q_{6}$, we know that $q_5$ is prior to $q_6$. 
%Therefore, whenever $\cT_{\sf nameReg}$ reads $\backslash$s at the state $q_3$,  it will choose to go the state $q_5$ greedily, unless this choice would lead to the nonacceptance (in this case, $q_6$ will be chosen). 

\begin{figure*}[ht]
\centering
%\rule{\linewidth}{0cm}
\includegraphics[width=0.8\textwidth]{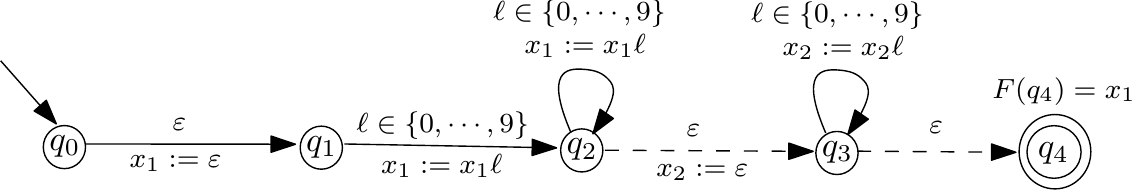}
\caption{The PSST $\cT$: Extract the matching of the first capturing group in ($\backslash$d+)($\backslash$d*)}
\label{fig-psst-exmp}
%\vspace{-4mm}
\end{figure*}
\end{example}

%  $\tmop{Out} (r) =
%  s_{\varepsilon} \circ s_1 \circ s_2 \ldots s_n \circ F (q_n)$ where
%  $s_{\varepsilon}$ is the empty substitution which maps all variables to
%  $\varepsilon$.

% Note that in the definition of \NSST, there is no \emph{copyless} restriction.

%%%%%%%%%%%%%%%%%%%%%%%%%%%%%%%%%%%%%%%%%%

\subsection{Semantics of \regexp-String Matching} \label{sect:regextopsst}

We now define the formal semantics of {\regexp}. Traditionally they are interpreted as a regular language which can be defined inductively. 
In our case, where {\regexp} are mainly used in string functions, %arisen from analysis of programming language such as JavaScript, %owing to the introduction of greedy/lazy semantics,  
what matters is %not only the language denoted by the regular expression, but also 
the intermediate result when parsing a string against the given {\regexp}. %This is especially the case when the capturing group is involved. 
As a result, we shall present an operational (as opposed to traditional denotational) account of the \regexp-string matching by constructing {\PSST}s out of regular expressions.

%, which  is defined by recursively constructing {\PSST}s out of regular expressions.
Note that in \cite{BDM14,BM17}, a construction from {\regexp} to prioritized finite transducers (PFT) was given. The construction therein is a variant of the classical Thompson construction from regular expressions to nondeterministic finite automata \cite{Thompson68}. In particular, the size of the constructed PFT is linear in the size of the given {\regexp}. One may be tempted to think that the construction in \cite{BDM14,BM17} can be easily adapted to construct {\PSST}s out of regular expressions. 
Nevertheless, the construction in \cite{BDM14,BM17} does \emph{not} work for so called \emph{problematic regular expressions}, i.e.,  those regular expressions that contain the subexpressions $e^*$ or $e^{*?}$ with $\varepsilon \in \Lang(e)$. Moreover, the construction therein did not consider the repetition operators $[e_1^{\{m_1,m_2\}}]$ or $[e_1^{\{m_1,m_2\}?}]$. 
%Finally, the construction therein was not validated against the actual semantics of regex-string matching in programming languages. 
%
Our construction, which is considerably different from that in \cite{BDM14,BM17}, works for arbitrary regular expressions. In particular, the size of the constructed {\PSST} can be \emph{exponential} in the size of the given regular expression in the worst case. Moreover, we validate by extensive experiments that our construction is consistent with the actual \regexp-string matching in JavaScript. %The size of the constructed {\PSST} can be \emph{exponential} in the size of the given regular expression in the worst case.

For technical convenience, we assume that $F$ in a {\PSST} is a set of final states, instead of an output function, in the sequel. 
The main idea of the construction is to split the set of final states, $F$, into two disjoint subsets $F_1$ and $F_2$, with the intention that $F_1$ and $F_2$ are responsible for accepting the empty string resp. non-empty strings. 
Therefore, the {\PSST}s constructed below are of the form $(Q, \Sigma, X, \delta, \tau, E, q_0, (F_1, F_2))$. 
The necessity of this splitting will be illustrated in Example~\ref{exmp-psst-partition}.

Furthermore, to deal with the situation that some capturing group may not be matched to any string and its value is undefined, we introduce a special symbol $\nullchar$ %for denoting this situation 
and assume that the initial values of all the string variables are $\nullchar$. %To avoid the tediousness, 
For simplicity, in the definition of a {\PSST}, if $\delta(q, a, q') = ()$ or $\tau(q, \varepsilon, q') = ((); ())$,  they will not be stated explicitly. Moreover, we will omit all the assignments $E(q, a, q')(x)$ such that $E(q, a, q')(x) = x$.
%By default, all the PSSTs from now on satisfy these assumptions. 

\OMIT{To be polished: 
To ensure this, making copies of {\PSST} might be needed during the construction, as is done in concatenation of two {\PSST}s (Definition \ref{def-psstconcat}), which leads to the worst case exponential blow-up.
However, the separation of $F$ into two sets is necessary for correct modeling of {\regexp} semantics. This is justified by the observation that in real world languages like JavaScript, the behaviour of certain {\regexp} operator, like greedy Kleene star $e^*$, is influenced by whether the subexpression $e$ matches an empty string. 
For example, the expression $(a^{*?})^*$ either matches $\varepsilon$ holistically and in this case the subexpression $a^{*?}$ is not entered and matched at all, or it matches the concatenation of a number of \emph{non-empty} strings which are accepted by $a^{*?}$. Note that this is an instance of problematic regular expression in the setting of \cite{BDM14, BM17}. By splitting the set of final states, we keep track of the emptiness of subexpressions during the match, and thus give a more fine-grained semantics to those operators and successfully handle this class of expression.
}

For {\PSST}s of the form $(Q, \Sigma, X, \delta, \tau, E, q_0, (F_1, F_2))$, we introduce a notation to be used in the construction, namely, the concatenation of two {\PSST}s. 

\begin{definition}[Concatenation of two PSSTs]\label{def-psstconcat}
For $i \in \{1,2\}$, let $\cT_i$ be a PSST such that $\cT_i = (Q_i, \Sigma, X_i, \delta_i, \tau_i, E_i, q_{i,0}, (F_{i,1}, F_{i,2}))$. Then the \emph{concatenation} of $\cT_1$ and $\cT_2$, denoted by $\cT_1 \concat \cT_2$, is defined as follows (see Fig.~\ref{fig-psstconcat}): 
Let  
$\cT'_{2} = (Q'_{2}, \Sigma, X_{2}, \delta'_{2}, \tau'_{2}, E_{2}', q'_{2,0}, (F'_{2, 1}, F'_{2,2}))$ be a fresh copy of $\cT_{2}$, but with the string variables of $\cT_{2}$ kept unchanged. 
Then 
\[\cT = ( Q_{1} \cup Q_{2} \cup Q'_{2}, \Sigma, X_1 \cup X_2, \delta, \tau, q_{1,0}, (F_{2,1}, F_{2,2} \cup F'_{2,1} \cup F'_{2,2}))\] 
where 
	\begin{itemize}
	\item $\delta$ comprises the transitions in $\delta_1$, $\delta_2$, and $\delta'_2$,
%	\begin{itemize}
%	\item for every $a \in \Sigma$, $\delta_e(q_{e,0}, a) = ()$,
	%
%	 \item for every $i \in \{1,2\}$, $q \in Q_{e_i}$ and $a \in \Sigma$, $\delta_e(q, a) = \delta_{e_i}(q, a)$,
%
%	\item for every $q' \in Q'_{e_2}$ and $a \in \Sigma$, $\delta_e(q', a) = \delta'_{e_2}(q',a)$, 
%	 \end{itemize}
			%
	\item $\tau$ comprises the transitions in $\tau_1$, $\tau_2$, $\tau'_2$, and the following transitions,
	\begin{itemize}
%	\item for every $q \in Q_{1} \setminus (F_{1,1} \cup F_{1,2})$, $\tau(q) = \tau_{1}(q)$, 
%
	\item for every $f_{1,1} \in F_{1,1}$, $\tau(f_{1,1}) = ((q_{2,0}); ())$, 
	\item for every $f_{1,2} \in F_{1,2}$, $\tau(f_{1,2}) = ((q'_{2,0}), ())$,
	\end{itemize}
	\item $E$ inherits all the assignments in $E_1$, $E_2$, and $E'_2$, and includes the following assignments:  for every $f_{1,1} \in F_{1,1}$, $f_{1,2} \in F_{1,2}$, and $x' \in X_2$, $E(f_{1,1}, \varepsilon, q_{2,0})(x') = E(f_{1,2}, \varepsilon, q'_{2,0})(x')= \nullchar$. (Intuitively, the values of all the variables in $X_2$ are reset when entering $\cT_2$ and $\cT'_2$.)
  \end{itemize}
% Fig.~\ref{fig-reg2pfa-2} depicts the construction. 
		\begin{figure}[tb]
%			\vspace{-2mm}
			\centering
			%\rule{\linewidth}{0cm}
			\includegraphics[width = 0.5\textwidth]{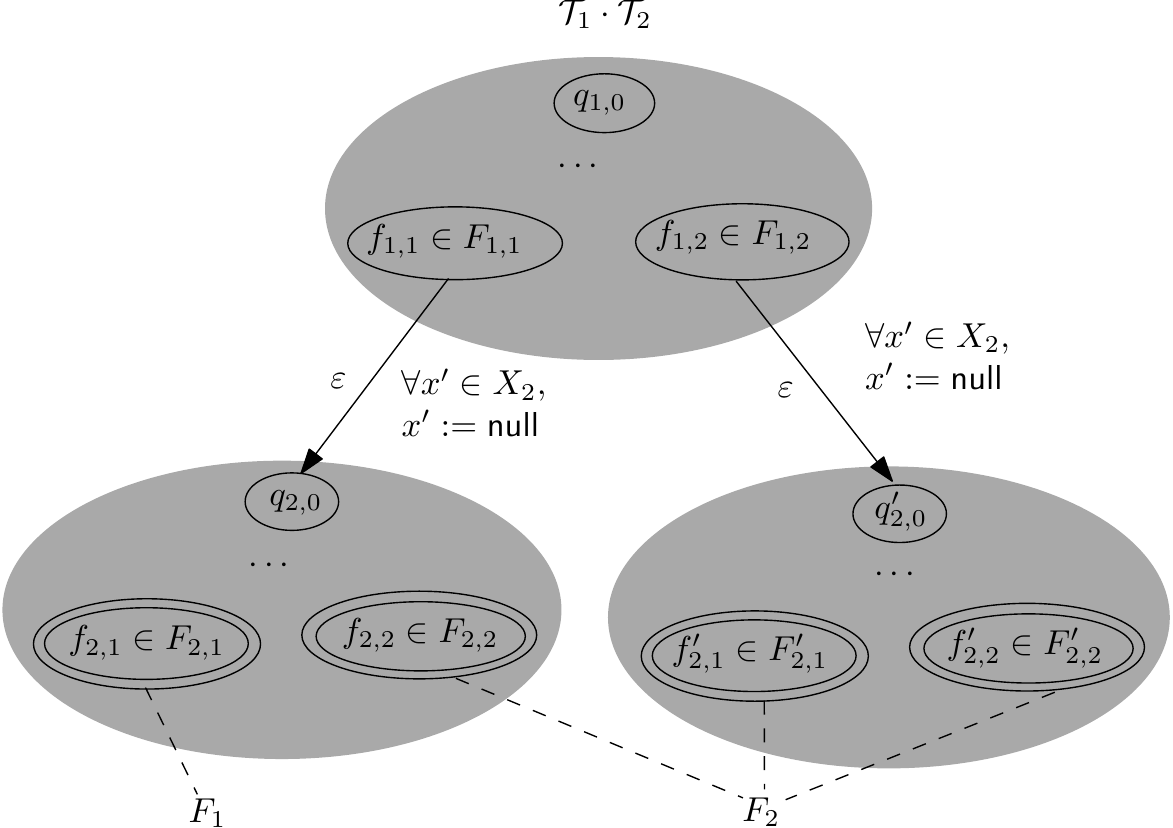}
			\caption{$\cT_1\concat \cT_2$: Concatenation of $\cT_1$ and $\cT_2$}
			\label{fig-psstconcat}
%			\vspace{-4mm}
		\end{figure}  
\end{definition}
Note that in the above definition, it is possible that $X_1 \cap X_2 \neq \emptyset$. We remark that if $F_{1,1} = \emptyset$ or $F_{2,1} = \emptyset$, then \emph{one copy} of $\cT_2$, instead of two copies, is sufficient for the concatenation.

We shall recursively construct a {\PSST} $\cT_e$ for each {\pcre} $e$, such that the initial state has no incoming transitions and each of its final states has no outgoing transitions. Moreover, all the transitions out of the initial state are $\varepsilon$-transitions. 
%
%For convenience, we assign a unique index for each subexpression in a {\pcre}. 
%Let us use $\idx_e$ to denote the index function associated with a {\pcre} $e$. 
We assume that in $\cT_e$, a string variable $x_{e'}$ is introduced for each subexpression $e'$ of $e$. 

The construction is technical and below we only select to present some representative cases. The other cases are
\ifproceeding given in the long version of this paper \cite{popl22-full}. 
\else relegated to Appendix~\ref{app:reg2psst}. 
\fi

\paragraph{Case $e = (e_1)$} $\cT_e$ is adapted from $\cT_{e_1} = (Q_{e_1}, \Sigma, X_{e_1}, \delta_{e_1}, \tau_{e_1}, E_{e_1},  q_{e_1,0}, (F_{e_1,1}, F_{e_1,2}))$ by adding the string variable $x_e$ and the assignments for $x_e$, that is, $X_e = X_{e_1} \cup \{x_e\}$ and for each transition $(q, a, q')$ in $\cT_{e_1} $ with $a \in \Sigma^\varepsilon$, we have $E_e(q, a, q')(x_e) = E_{e_1}(q, a, q')(x_{e_1})$.

\paragraph{Case $e = [e_1 + e_2]$ (see Fig.~\ref{fig-reg2pfa-1})} For $i \in \{1, 2\}$, let it be the case that we have
$\cT_{e_i} = (Q_{e_i}, \Sigma, X_{e_i}, \delta_{e_i}, \tau_{e_i}, E_{e_i},  q_{e_i,0}, (F_{e_i,1}, F_{e_i,2}))$. Moreover, assume $X_{e_1} \cap X_{e_2} = \emptyset$.
%$\cT_{e_2} = (Q_{e_2}, \Sigma, X_{e_2}, \delta_{e_2}, \tau_{e_2}, E_{e_2}, q_{e_2,0}, (F_{e_2, 1}, F_{e_2,2}))$.
% 
Then 
\[\cT_e = (Q_{e_1} \cup Q_{e_2} \cup \{q_{e,0}\}, \Sigma, X_{e_1} \cup X_{e_2} \cup \{x_e\}, 
		\delta_e, \tau_e, E_e, q_{e,0}, (F_{e_1,1} \cup F_{e_2,1}, F_{e_1,2} \cup F_{e_2,2}))\] where  
\begin{itemize}
%			\item $q_{e,0}  \not \in Q_{e_1} \cup Q_{e_2}$, 
%			 \item  $X_e = X_{e_1} \cup X_{e_2} \cup \{x_e\}$,
			\item $\delta_e$ comprises the transitions in $\delta_{e_1}$ and $\delta_{e_2}$,
%			defined as follows:
%			\begin{itemize}
%			\item $\delta_e(q, a) = \delta_{e_i}(q, a)$ for every $i \in \{1,2\}$, $q \in Q_{e_i}$ and $a \in \Sigma$, 
			%
%			\item $\delta_e(q_{e,0}, a)  = ()$ for every $a \in \Sigma$, 
%			\end{itemize}
			%
			\item $\tau_e$ comprises the transitions in $\tau_{e_1}$ and $\tau_{e_2}$, as well as the transition $\tau_e(q_{e,0}) = ((q_{e_1,0}); (q_{e_2,0}))$,
%			defined as follows: 
%			\begin{itemize}
%			\item $\tau_e(q) = \tau_{e_i}(q)$ for every $i \in \{1,2\}$ and $q \in Q_{e_i}$, 
%			\item $\tau_e(q_{e,0}) = ((q_{e_1,0},q_{e_2,0}); ())$,
%			\end{itemize}
			\item $E_e$ inherits $E_{e_1}$, $E_{e_2}$, plus the assignments $E_e(q_{e,0}, \varepsilon, q_{e_1,0})(x_{e}) = E_e(q_{e,0}, \varepsilon, q_{e_2,0})(x_{e}) =\varepsilon$, as well as $E_e(q,a, q')(x_{e}) = x_e a$ for every transition $(q, a, q')$ in $\cT_{e_1}$ and $\cT_{e_2}$ (where $a \in \Sigma^\varepsilon$).
			%  is defined as follows: 
%			\begin{itemize}
%			\item for each $i\in\{1,2\}$, transition $(q, a, q')$ in $\cT_{e_i}$ (where $a \in \Sigma^\varepsilon$), $x \in X_{e_i}$ and $x' \in X_{e} \setminus X_{e_i}$, $E_e(q,a,q')(x) = E_{e_i}(q,a,q')(x)$ and $E_e(q,a,q')(x') =x'$, 
%			\item $E_e(q_{e,0}, \varepsilon, q_{e_1,0})(x_{e}) = E_e(q_{e,0}, \varepsilon, q_{e_2,0})(x_{e}) =\varepsilon$, and $E_e(q_{e,0}, \varepsilon, q_{e_1,0})(x) = E_e(q_{e,0}, \varepsilon, q_{e_2,0})(x) =x$ for every $x \in X_{e_1} \cup X_{e_2}$.
%			\end{itemize}
\end{itemize}
%Fig.~\ref{fig-reg2pfa-1} depicts the construction.  	
		\begin{figure}[ht]
%		\vspace{-2mm}
			\centering
			%\rule{\linewidth}{0cm}
			\includegraphics[width = 0.6\textwidth]{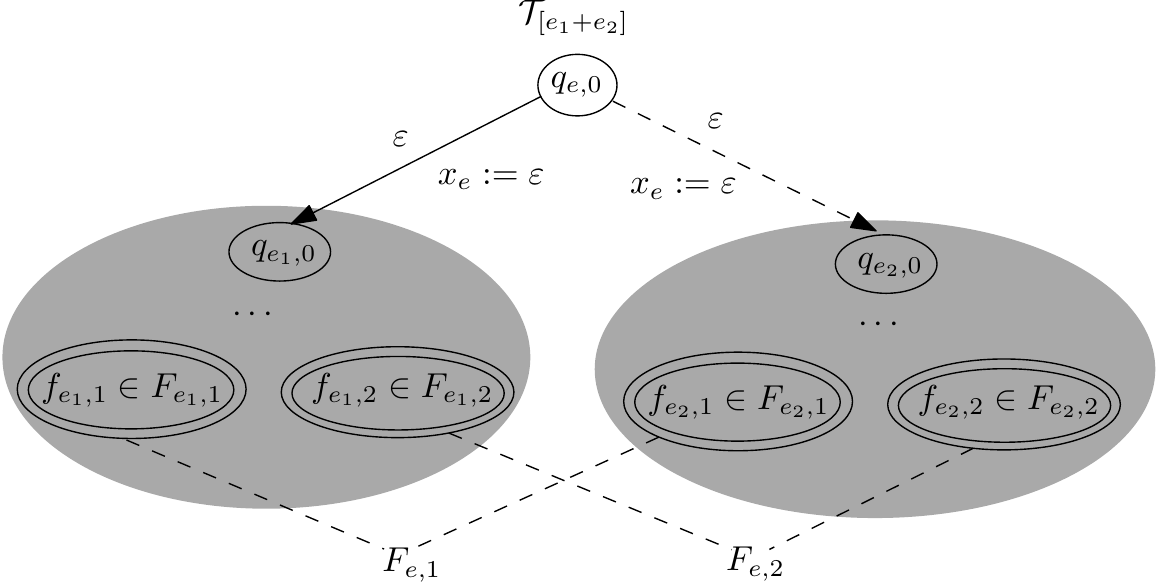}
			\caption{The PSST $\cT_{[e_1+e_2]}$}
			\label{fig-reg2pfa-1}
%			\vspace{-4mm}
		\end{figure}

\paragraph{Case $e = [e_1 \concat e_2]$} 
For $i \in \{1,2\}$, let  
$\cT_{e_i} = (Q_{e_i}, \Sigma, X_{e_i}, \delta_{e_i}, \tau_{e_i}, E_{e_i}, q_{e_i,0}, (F_{e_i,1}, F_{e_i,2}))$. Moreover, let us assume that $X_{e_1}\cap X_{e_2}=\emptyset$.
Then $\cT_e$ is obtained from $\cT_{e_1} \concat \cT_{e_2}$ (the concatenation of $\cT_{e_1}$ and $\cT_{e_2}$, see Fig.~\ref{fig-psstconcat}) by adding a string variable $x_e$, a fresh state $q_{e,0}$ as the initial state, the $\varepsilon$-transition $\tau_e(q_{e,0}) = ((q_{e_1,0});())$, and the assignments $E_e(q_{e,0}, \varepsilon, q_{e_1,0})(x_e) = \varepsilon$, $E_e(p, a, q)(x_e) = x_e a$ for every transition $(p, a, q)$ in $\cT_{e_1}$, $\cT_{e_2}$, and $\cT'_{e_2}$ (where $a \in \Sigma^\varepsilon$).

\paragraph{Case $e = [e_1^?]$ (see Fig.~\ref{fig-reg2pfa-6})} Let $\cT_{e_1} = (Q_{e_1}, \Sigma, X_{e_1}, \delta_{e_1}, \tau_{e_1}, E_{e_1}, q_{e_1,0}, (F_{e_1,1}, F_{e_1,2}))$. Then 
\[\cT_e = (Q_{e_1} \cup \{q_{e,0}, f_{\varepsilon}\}, \Sigma, X_{e_1} \cup \{x_e\}, 
		\delta_e, \tau_e, E_{e}, q_{e,0}, (\{f_{\varepsilon}\}, F_{e_1,2}))\]
where  
		\begin{itemize}
%			\item $q_{e,0}, f_\varepsilon  \not \in Q_{e_1}$,
			\item $\delta_e$ is exactly $\delta_{e_1}$,
%			\begin{itemize}
%			\item $\delta_e(q, a) = \delta_{e_1}(q, a)$ for every $q \in Q_{e_1}$ and $a \in \Sigma$, 
%			\item $\delta_e(q_{e,0}, a)  = ()$ and $\delta_e(f_{\varepsilon}, a) = ()$ for every $a \in \Sigma$, 
%			\end{itemize}
			%
			\item $\tau_e$ comprises the transitions in $\tau_{e_1}$, as well as the transition $\tau_e(q_{e,0}) = ((q_{e_1,0}, f_{\varepsilon}); ())$,
%			\begin{itemize}
%			\item $\tau_e(q) = \tau_{e_1}(q)$ for every $q \in Q_{e_1}$, 
%			\item $\tau_e(q_{e,0}) = ((q_{e_1,0}, f_{\varepsilon}); ())$,
%			\end{itemize}
			\item $E_e$ inherits $E_{e_1}$ and includes the assignments $E_e(q_{e,0},\varepsilon,q_{e_1, 0})(x_e) = E_e(q_{e,0},\varepsilon,f_{\varepsilon})(x_e) =\varepsilon$, as well as $E_e(q, a, q')(x_e) =x_e a$ for every transition $(q, a, q')$ in $\cT_{e_1}$ (where $a \in \Sigma^\varepsilon$).
%			defined as follows: 
%			\begin{itemize}
%			\item for each transition $(q, a, q')$ in $\cT_{e_1}$ (where $a \in \Sigma^\varepsilon$), $E_e(q,a,q')(x_e) =x_ea$, and for every $x \in X_{e_1}$, $E_e(q,a,q')(x) =E_{e_1}(q, a,q')(x)$, 
			%
%			\item $E_e(q_{e,0},\varepsilon,q_{e_1, 0})(x_e) = E_e(q_{e,0},\varepsilon,f_{\varepsilon})(x_e) =\varepsilon$,  and for each $x \in X_{e_1}$, $E_e(q_{e,0},\varepsilon,q_{e_1, 0})(x)  = E_e(q_{e,0},\varepsilon, f_{\varepsilon})(x) =x$. 
%			\end{itemize}
		\end{itemize}
Note that $F_{e_1,1}$ is not included into $F_{e,1}$ here.
%
%Fig.~\ref{fig-reg2pfa-6} depicts the construction.
		\begin{figure}[tb]
%		\vspace{-2mm}
			\centering
			%\rule{\linewidth}{0cm}
			\includegraphics[width = 0.8\textwidth]{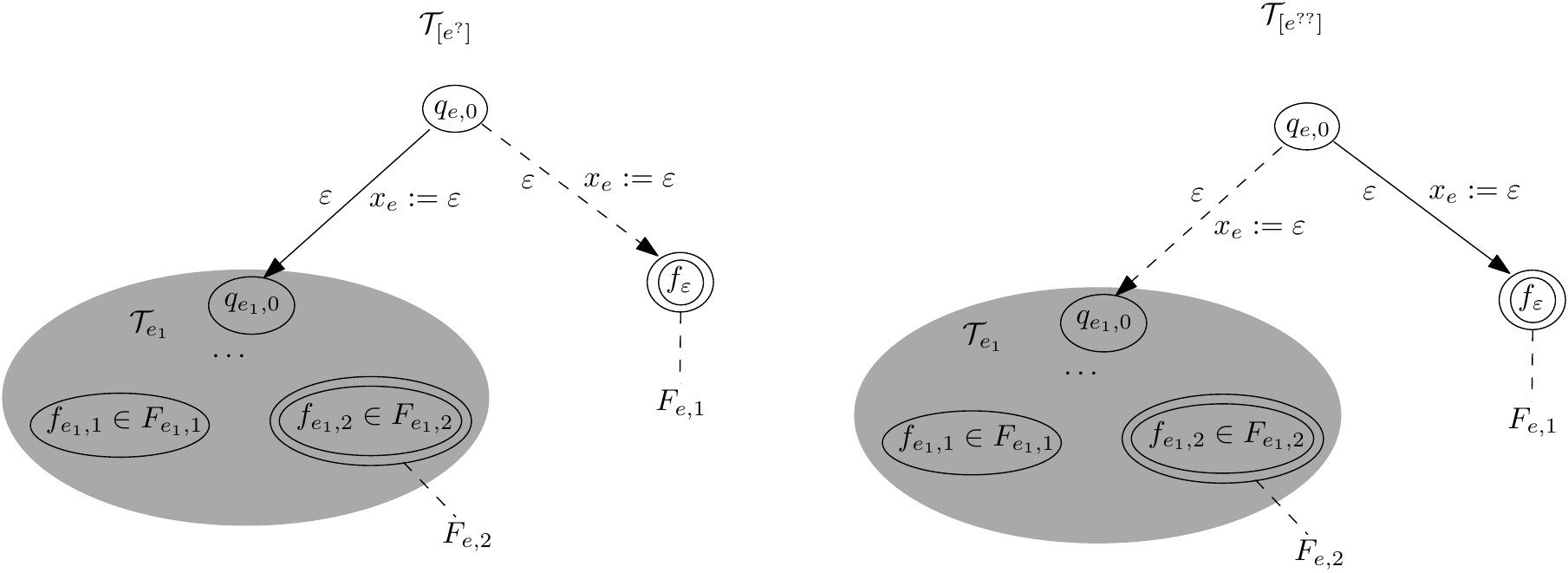}
			\caption{The PSST $\cT_{[e_1^?]}$ and $\cT_{[e_1^{??}]}$}
			\label{fig-reg2pfa-6}
%			\vspace{-4mm}
		\end{figure}

%%%%%%%%%%%%%%%%%%%%%%%%%%%%%%%%%%%%%%%%%%%%%%%%%%%%%%%%%%%%%%%%%%%%%%%%%%%%%%%%%%%%%%%%%%%%%%%%%%%%%
\paragraph{Case $e = [e_1^{??}]$ (see Fig.~\ref{fig-reg2pfa-6})}  In this case, $\cT_{[e_1^{??}]}$ is  almost the same as $\cT_{[e_1^{?}]}$. The only difference is that the priorities of the two $\varepsilon$-transitions out of $q_{e,0}$ are swapped, namely, $\tau_e(q_{e,0}) = ((f_{\varepsilon}, q_{e_1,0}); ())$ here.
%%%%%%%%%%%%%%%%%%%%%%%%%%%%%%%%%%%%%%%
\OMIT
{
Let $\cT_{e_1} = (Q_{e_1},
		\Sigma, X_{e_1}, \delta_{e_1}, \tau_{e_1}, E_{e_1} , q_{e_1,0}, (F_{e_1,1}, F_{e_1,2}))$. 
Then 
\[\cA_e = (Q_{e_1} \cup \{q_{e,0}, q_{\varepsilon}\}, \Sigma, X, 
		\delta_e, \tau_e, E, q_{e,0}, (\{q_{\varepsilon}\}, F_{e_1,2}))\] 
where 
		\begin{itemize}
			\item $q_{e,0}  \not \in Q_{e_1}$
			\item $\delta_e(q, a) = \delta_{e_1}(q, a)$ for every $q \in Q_{e_1}$ and $a \in \Sigma$, $\delta_e(q_{e,0}, a)  = ()$ and $\delta_e(q_{\varepsilon}, a) = ()$ for every $a \in \Sigma$, 
			\item $\tau_e(q) = \tau_{e_1}(q)$ for every $q \in Q_{e_1}$, $\tau_e(q_{e,0}) = ((q_{\varepsilon}, q_{e_1,0}); ())$,
			
			\item for each transition $(q, a, q')$ from $\delta_{e_1}$, $E(q,a,q')(x) = E_1(q,a,q')(x)$ and $E(q_{e,0},\varepsilon,q_{\varepsilon})(x) =\varepsilon$
		\end{itemize}
}
%%%%%%%%%%%%%%%%%%%%%%%%%%%%%%%%%%%%%%%
\OMIT{
Fig.~\ref{fig-reg2pfa-7} depicts the construction. 
		\begin{figure}[ht]
			\centering
			%\rule{\linewidth}{0cm}
			\includegraphics[width = 0.5\textwidth]{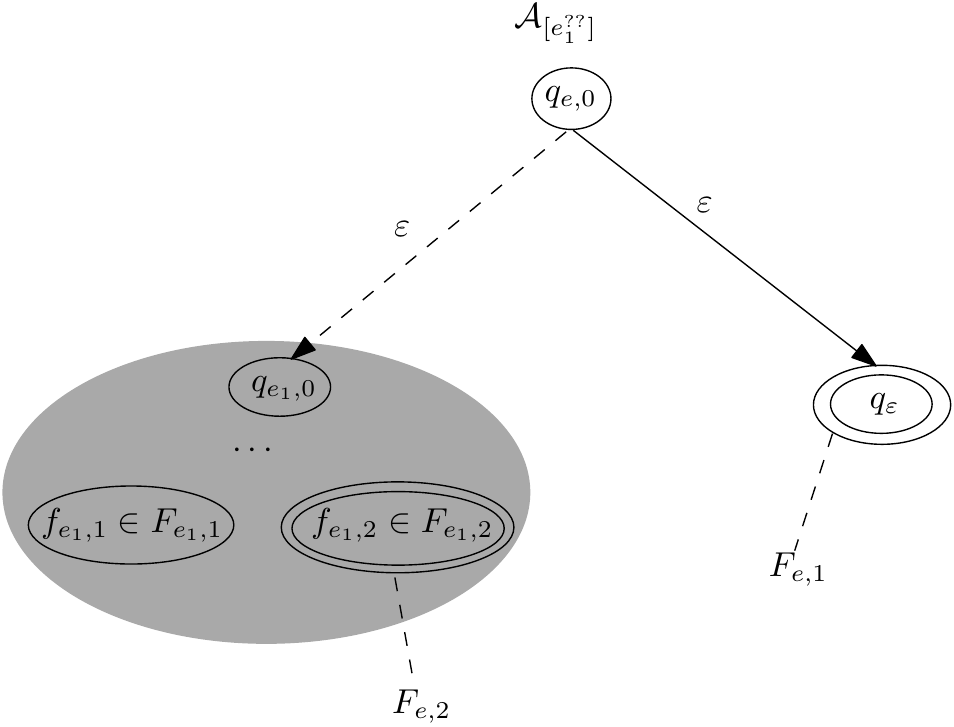}
			\caption{The PFA $\cA_{[e_1^{??}]}$}
			\label{fig-reg2pfa-7}
		\end{figure}
}	

%%%%%%%%%%%%%%%%%%%%%%%%%%%%%%%%%%%%%%%%%%%%%%%%%%%%%%%%%%%%%%%%%%%%%%%%%%%%%%%%%%%%%%%%%%%%%%%%%%%%%

\paragraph{Case $e = [e_1^{\ast}]$ (see Fig.~\ref{fig-reg2pfa-3})} 
Let $\cT_{e_1} = (Q_{e_1}, \Sigma, X_{e_1}, \delta_{e_1}, \tau_{e_1}, E_{e_1}, q_{e_1,0}, (F_{e_1,1}, F_{e_1,2}))$. Then
\[ \cT_e = (Q_{e_1} \cup \{q_{e,0}, f_{e,1}, f_{e,2}\}, \Sigma, X_{e}, \delta_e, E_{e}, \tau_e, q_{e,0}, (\{f_{e,1}\}, \{f_{e,2}\}))\] where 
		\begin{itemize}
%			\item $q_{e,0}, f_{e,1}, f_{e,2} \not \in Q_{e_1}$,
			
			\item $\delta_e$ is exactly $\delta_{e_1}$,
%			 defined as follows:
%			\begin{itemize}
%			\item for every $q \in Q_{e_1}$ and $a \in \Sigma$, $\delta_e(q, a) = \delta_{e_1}(q, a)$, 
%			\item for every $a \in \Sigma$, $\delta_e(q_{e,0}, a) = \delta_e(f_{e,0}, a) = \delta_e(f_{e,1}, a) = ()$,
%			\end{itemize}
			%moreover, $\delta(q_0, a) = \delta(f_0, a)  = ()$,
%
			\item $\tau_e$ comprises the transitions in $\tau_{e_1}$,  as well as the transitions $\tau_e(q_{e,0}) = ((q_{e_1,0}, f_{e,1}); ())$,  $\tau_e(f_{e_1,1}) = ((q_{e_1,0});())$ for every $f_{e_1,1} \in F_{e_1,1}$, and $\tau_e(f_{e_1,2}) = ((q_{e_1,0}, f_{e,2});())$ for every $f_{e_1,2} \in F_{e_1,2}$, 
%			defined as follows: 
%			\begin{itemize}
%			\item for every $q \in Q_{e_1} \setminus (F_{e_1,1} \cup F_{e_1,2})$,  $\tau_e(q) = \tau_{e_1}(q)$, moreover, $\tau_e(q_{e,0}) = ((q_{e_1,0},f_{e,0}); ())$,  $\tau_e(f_{e_1,1}) = ((q_{e_1,0});())$ for every $f_{e_1,1} \in F_{e_1,1}$, $\tau_e(f_{e_1,2}) = ((q_{e_1,0}, f_{e,1});())$ for every $f_{e_1,2} \in F_{e_1,2}$, $\tau_e(f_{e,0}) =\tau_e(f_{e,1}) = (();())$. (Intuitively, the $\varepsilon$-transitions from $q_{e,0}$ to $q_{e_1,0}$ and $f_{e,0}$, from each $f_{e_1,1} \in F_{e_1,1}$ to $q_{e_1,0}$, and from $f_{e_1,2} \in F_{e_1,2}$ to $q_{e_1,0}$ and $f_{e,1}$ respectively are added, moreover, the $\varepsilon$-transition from $q_{e,0}$ to $f_{e,0}$ and from $f_{e_1,2} \in F_{e_1,2}$ to $f_{e,1}$ are of the lowest priority.)
%			\end{itemize}
			\item $E_e$ inherits $E_{e_1}$ plus the assignments $E_e(q_{e,0},\varepsilon,f_{e,1})(x_e) = E_e(q_{e,0},\varepsilon,q_{e_1,0})(x_e) = \varepsilon$, $E_e(f_{e_1,1},\varepsilon, q_{e_1,0})(x) = E_e(f_{e_1,2},\varepsilon, q_{e_1,0})(x)= \nullchar$ for every $f_{e_1,1} \in F_{e_1,1}$, $f_{e_1,2} \in F_{e_1,2}$, and $x \in X_{e_1}$, as well as $E_e(q, a, q')(x_e) = x_e a$ for every transition $(q, a, q')$ in $\cT_{e_1}$ with $a \in \Sigma^\varepsilon$. (Intuitively, the values of all the string variables in $X_{e_1}$ are reset when starting a new iteration of $e_1$.)
%			 is defined as follows: 
%			\begin{itemize}
%			\item for each transition $(q, a, q')$ in $\cT_{e_1}$, $E_e(q,a,q')(x_e) = x_e a$, and for every $x \in X_{e_1}$, $E_e(q,a,q')(x) = E_{e_1}(q,a,q')(x)$, 
%			\item $E_e(q_{e,0},\varepsilon,f_{e,0})(x_e) = E_e(q_{e,0},\varepsilon,q_{e_1,0})(x_e) = \varepsilon$, and for every $x \in X_{e_1}$, $E_e(q_{e,0},\varepsilon,f_{e,0})(x) = E_e(q_{e,0},\varepsilon,q_{e_1,0})(x) = x$, 
%			\item for every $f_{e_1,1} \in F_{e_1,1}$, $E_e(f_{e_1,1},\varepsilon,q_{e_1,0})(x_{e}) = \varepsilon$, and for every $x \in X_{e_1}$, $E_e(f_{e_1,1},\varepsilon,q_{e_1,0})(x) =x$, 
%			\item for every $f_{e_1,2} \in F_{e_1,2}$, $E_e(f_{e_1,1},\varepsilon,q_{e_1,0})(x_{e}) = \varepsilon$, $E_e(f_{e_1,2},\varepsilon,f_{e,1})(x_{e}) = x_e$, and for every $x \in X_{e_1}$, $E_e(f_{e_1,2},\varepsilon,q_{e_1,0})(x) = E(f_{e_1,2},\varepsilon, f_{e,1})(x) = x$.
%			\end{itemize}
		\end{itemize}

%Fig.~\ref{fig-reg2pfa-3} depicts the construction. 
\begin{figure}[tb]
	\centering
	%\rule{\linewidth}{0cm}
	\includegraphics[width = 0.8\textwidth]{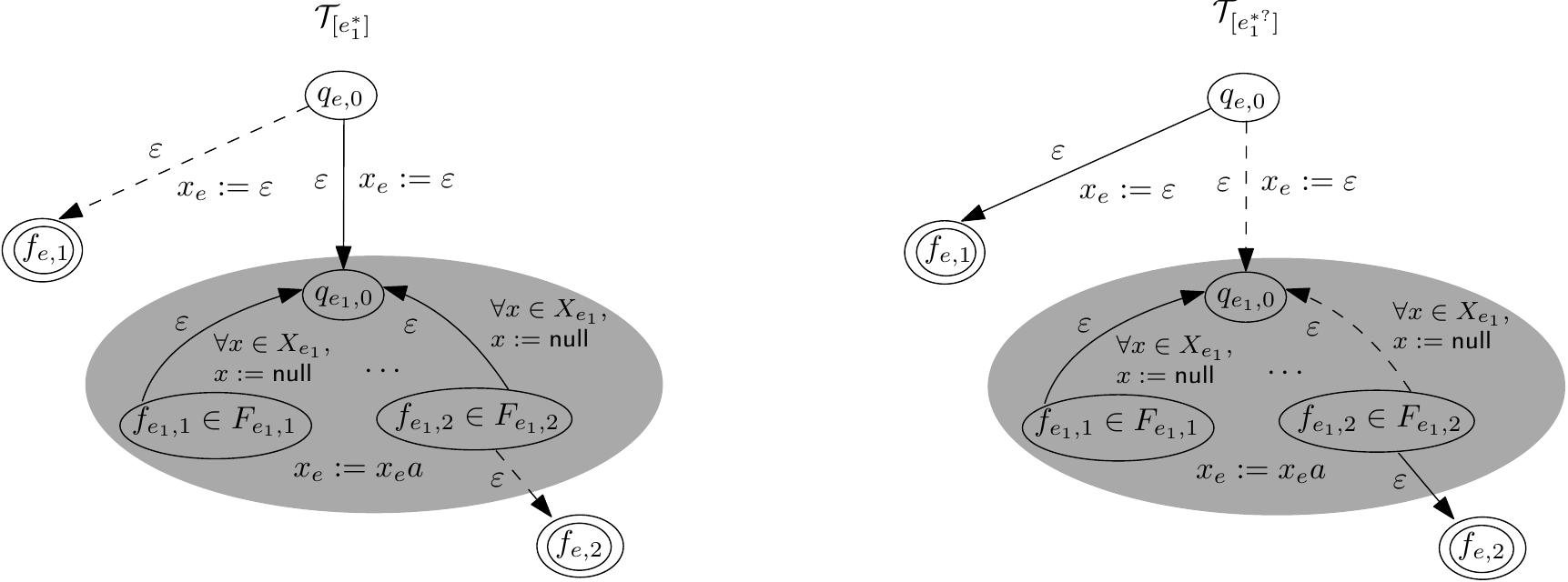}
	\caption{The PSST $\cT_{[e_1^\ast]}$ and $\cT_{[e_1^{\ast?}]}$}
	\label{fig-reg2pfa-3}
%	\vspace{-4mm}
\end{figure}

%%%%%%%%%%%%%%%%%%%%%%%%%%%%%%%%%%%%%%%%%%%%%%%%%%%%%%%%%%%%%%%%%%%%%%%%%%%%%%%%%%%%%%%%%%%%%%%%%%%%%
\paragraph{Case $e = [e_1^{\ast?}]$ (see Fig.~\ref{fig-reg2pfa-3})} The construction is almost the same as $e = [e_1^{\ast}]$. The only difference is that the priorities of the $\varepsilon$-transitions out of $q_{e,0}$ resp. $f_{e_1,2} \in F_{e_1,2}$ are swapped.
\OMIT{
Let $\cA_{e_1} = (Q_{e_1}, \Sigma, X_1, \delta_{e_1}, \tau_{e_1}, E_1, q_{e_1,0}, (F_{e_1,1}, F_{e_1,2}))$. 
Then 
\[\cA_e = (Q_{e_1} \cup \{q_{e,0}, f_{e,0}, f_{e,1}\}, \Sigma, X_1, \delta_e, \tau_e, E, q_{e,0}, (\{f_{e,0}\}, \{f_{e,1}\}))\]  
where 
		\begin{itemize}
			\item $q_{e,0}, f_{e,0} \not \in Q_{e_1}$,
			
			\item for every $q \in Q_{e_1}$ and $a \in \Sigma$, $\delta_e(q, a) = \delta_{e_1}(q, a)$, 
			%moreover, $\delta(q_0, a) = \delta(f_0, a)  = ()$,
			
			\item for every $q \in Q_{e_1} \setminus (F_{e_1,1} \cup F_{e_1,2})$,  $\tau_e(q) = \tau_{e_1}(q)$, moreover, $\tau_e(q_{e,0}) = ((f_{e,0}, q_{e_1,0}); ())$,  $\tau_e(q) = ((q_{e_1,0});())$ for every $q \in F_{e_1,1}$, $\tau_e(q) = ((f_{e,1}, q_{e_1,0});())$ for every $q \in F_{e_1,2}$, $\tau_e(f_{e,0}) =\tau_e(f_{e,1}) = (();())$. (Intuitively, the $\varepsilon$-transitions from $q_{e,0}$ to $f_{e,0}$ and $q_{e_1,0}$ , from each $q \in F_{e_1,1}$ to  $q_{e_1,0}$, and from each $q \in F_{e_1,2}$ to $f_{e,1}$ and $q_{e_1,0}$ respectively are added, moreover, the $\varepsilon$-transition from $q_{e,0}$ to $q_{e_1,0}$ and from $q \in F_{e_1,2}$ to $q_{e_1,0}$ are of the lowest priority.)
			
			\item for each transition $(q, a, q')$ from $\delta_{e_1}$, $E(q,a,q')(x) = E_1(q,a,q')(x)$, and $E(q_{e,0},\varepsilon,q_{f_e,0})(x) =\varepsilon$, $E(q_{e,0},\varepsilon,q_{e_1,0})(x) =\varepsilon$, $E(f_{e_1,2},\varepsilon,f_{e,1})(x) =x$.
		\end{itemize}
}

\paragraph{Case $e = [e_1^{+}]$}  We first construct $\cT_{e_1}$ and $\cT^-_{[e^\ast_1]}$, where $\cT^-_{[e^\ast_1]}$ is obtained from $\cT_{[e^\ast_1]}$ by dropping the string variable $x_{[e^\ast_1]}$. Therefore, $\cT_{e_1}$ and $\cT^-_{[e^\ast_1]}$ have the same set of string variables, $X_{e_1}$. Then we construct $\cT_{e}$ by adding into $\cT_{e_1} \concat \cT^-_{[e^\ast_1]}$ a fresh state $q_{e,0}$ as the initial state, and the transitions $\tau_e(q_{e,0}) = ((q_{e_1,0});())$, as well as the assignments $E_e(q_{e,0}, \varepsilon, q_{e_1,0})(x_e) = \varepsilon$, $E_e(q, a, q')(x_e) = x_e a$ for every transition $(q, a, q')$ in $\cT_{e_1} \concat \cT^-_{[e^\ast_1]}$. 
%(Note that in $\cT_{e_1} \concat \cT^-_{[e^\ast_1]}$, the values of all variables in $X_{e_1}$ are reset when entering $ \cT^-_{[e^\ast_1]}$ and $(\cT^-_{[e^\ast_1]})'$.)

%%%%%%%%%%%%%%%%%%%%%%%%%%%%%%%%%%%%%%%%%%%%%%%%%%%%%%%%%%%%%%%%%%%%%%%%%%%%%%%%%%%%%%%%%%%%%%%%%%%%%
\paragraph{Case $e = [e_1^{\{m_1,m_2\}}]$ for $1 \le m_1 < m_2$ (see Fig.~\ref{fig-reg2pfa-4})} We first construct $\cT^{\{m_1\}}_{e_1}$ as the concatenation of $m_1$ copies of $\cT_{e_1}$ (Recall Definition~\ref{def-psstconcat} for the concatenation of PSSTs). Note that $\cT^{\{m_1\}}_{e_1}$ is different from $\cT_{e_1^{m_1}}$, the PSST constructed from $e_1^{m_1}$, the concatenation of the expression $e_1$ for $m_1$ times. In particular, the set of string variables in $\cT^{\{m_1\}}_{e_1}$ is $X_{e_1}$, which is different from that of $\cT_{e_1^{m_1}}$. 

Then we construct the PSST $\cT^{\{1,m_2-m_1\}}_{e_1}$ (see Fig.~\ref{fig-reg2pfa-4}), which consists of $m_2-m_1$ copies of $\cT_{e_1}$, denoted by $(\cT^{(i)}_{e_1})_{i \in [m_2-m_1]}$, as well as the $\varepsilon$-transition from $q^{(1)}_{e_1,0}$ to a fresh state $f^\prime_0$ (of the lowest priority), and the $\varepsilon$-transitions from each $f^{(i)}_{e_1,2} \in F^{(i)}_{e_1,2}$ with $1\le i < m_2-m_1$ to $q^{(i+1)}_{e_1,0}$ (of the highest priority) and a fresh state $f^\prime_1$ (of the lowest priority). The final states of $\cT^{\{1,m_2-m_1\}}_{e_1}$ are $(\{f_0'\},\{f_1'\})$. (Intuitively, each $\cT^{(i)}_{e_1}$ accepts only nonempty strings, thus $f^{(i)}_{e_1,1} \in F^{(i)}_{e_1,1}$ contains no outgoing transitions in $\cT^{\{1,m_2-m_1\}}_{e_1}$. ) Note that the set of string variables in $\cT^{\{1,m_2-m_1\}}_{e_1}$ is still $X_{e_1}$.
\begin{figure}[tb]
%	\vspace{-2mm}
	\centering
	%\rule{\linewidth}{0cm}
	\includegraphics[width = 0.8\textwidth]{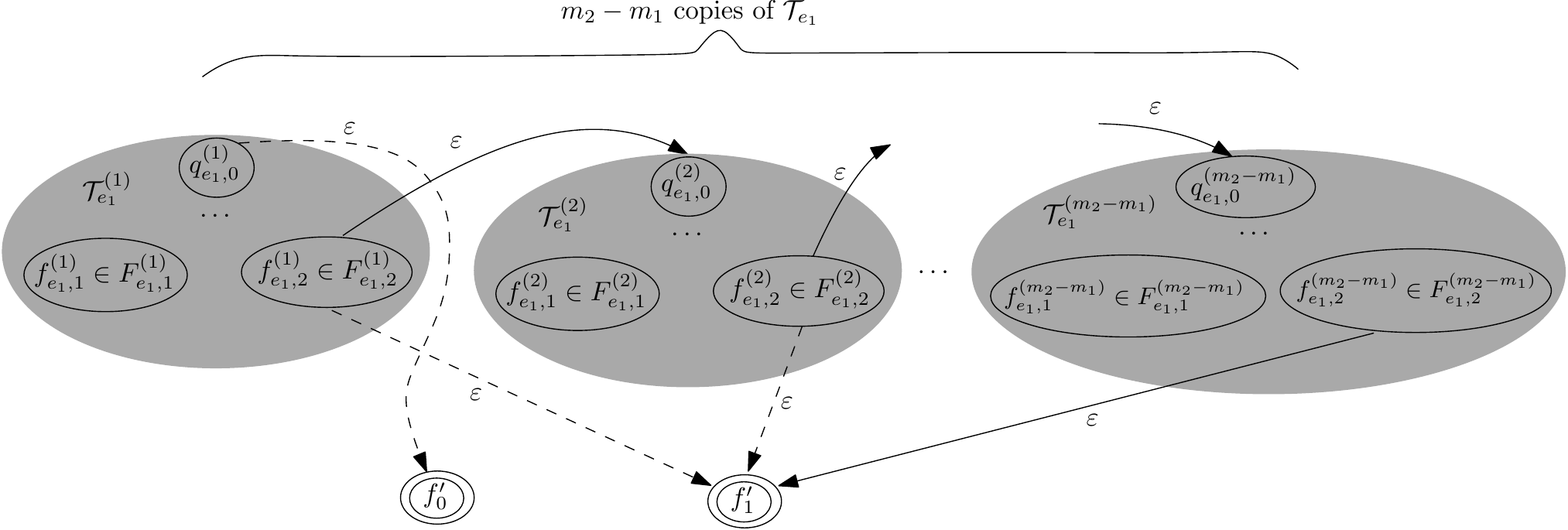}
	\caption{The PSST $\cT^{\{1,m_2-m_1\}}_{e_1}$}
	\label{fig-reg2pfa-4}
%	\vspace{-4mm}
\end{figure}  

Finally, we construct $\cT_e$ from $\cT^{\{m_1\}}_{e_1} \concat \cT^{\{1,m_2-m_1\}}_{e_1}$, the concatenation of $\cT^{\{m_1\}}_{e_1}$ and $\cT^{\{1,m_2-m_1\}}_{e_1}$, by adding a fresh state $q_{e,0}$, a string variable $x_e$, the $\varepsilon$-transition $\tau_e(q_{e,0}) = ((q_{e_1,0});())$ (assuming that $q_{e_1,0}$ is the initial state of $\cT^{\{m_1\}}_{e_1}$),  and also the assignments $E_e(q_{e_0}, \varepsilon, q_{e_1,0})(x_e) = \varepsilon$, as well as $E_e(q, a, q')(x_e) = x_e a$ for each transition $(q, a, q')$ in  $\cT^{\{m_1\}}_{e_1} \concat \cT^{\{1,m_2-m_1\}}_{e_1}$.

\begin{example}\label{exmp-pfa}
Consider {\pcre} $e = [a^+]$. 
%	The PFA corresponding to the RWRE $e = [[([\Gamma^+])\concat .?] \concat ([\Gamma^*])]$ 
	%in Example~\ref{exmp-regex-match-tree}
	%
We first construct $\cT_{a}$ and $\cT^-_{a^*}$ (recall that $\cT^-_{a^*}$ is obtained from $\cT_{a^*}$ by removing the string variable $x_{[a^*]}$, see Fig.~\ref{fig-pfa}).  Then we construct $\cT_{e}$ from $\cT_{a} \concat \cT^-_{a^*}$ by adding the initial state $q_{[a^+],0}$, the string variable $x_{[a^+]}$, as well as the assignments for $x_{[a^+]}$ (see Fig.~\ref{fig-pfa}). Note here only one copy of $\cT^-_{a^*}$ is used in $\cT_{a} \concat \cT^-_{a^*}$, since $\varepsilon$ is not accepted by $\cT_{a}$.
 %
%illustrated in Fig.~\ref{fig-pfa}, where the dashed (resp. thicker solid) lines represent the $\varepsilon$-transitions of lower (resp. higher) priorities than non-$\varepsilon$ transitions (if there is any), and the doubly circled states are final states. For instance, $\delta(q_1, \ell) = (q_1)$ for every $\ell \in \{0, \dots, 9\}$, $\delta(q_1, .) = ()$, $\tau(q_1) = ((); (q_2))$. Since the $\varepsilon$-transition has lower priority than the $\ell$-transition at the state $q_1$, whenever the currently scanned letter is $\ell \in \{0,\cdots,9\}$ at $q_1$,  the PFA will choose to go to $q_1$ greedily, until there is no more $\ell  \in \{0,\cdots,9\}$. (In this case, it has to choose the $\epsilon$-transition and goes to $q_2$.)
	%
	\begin{figure}[tb]
%		\vspace{-2mm}
		\centering
		%\rule{\linewidth}{0cm}
		\includegraphics[width=0.9\textwidth]{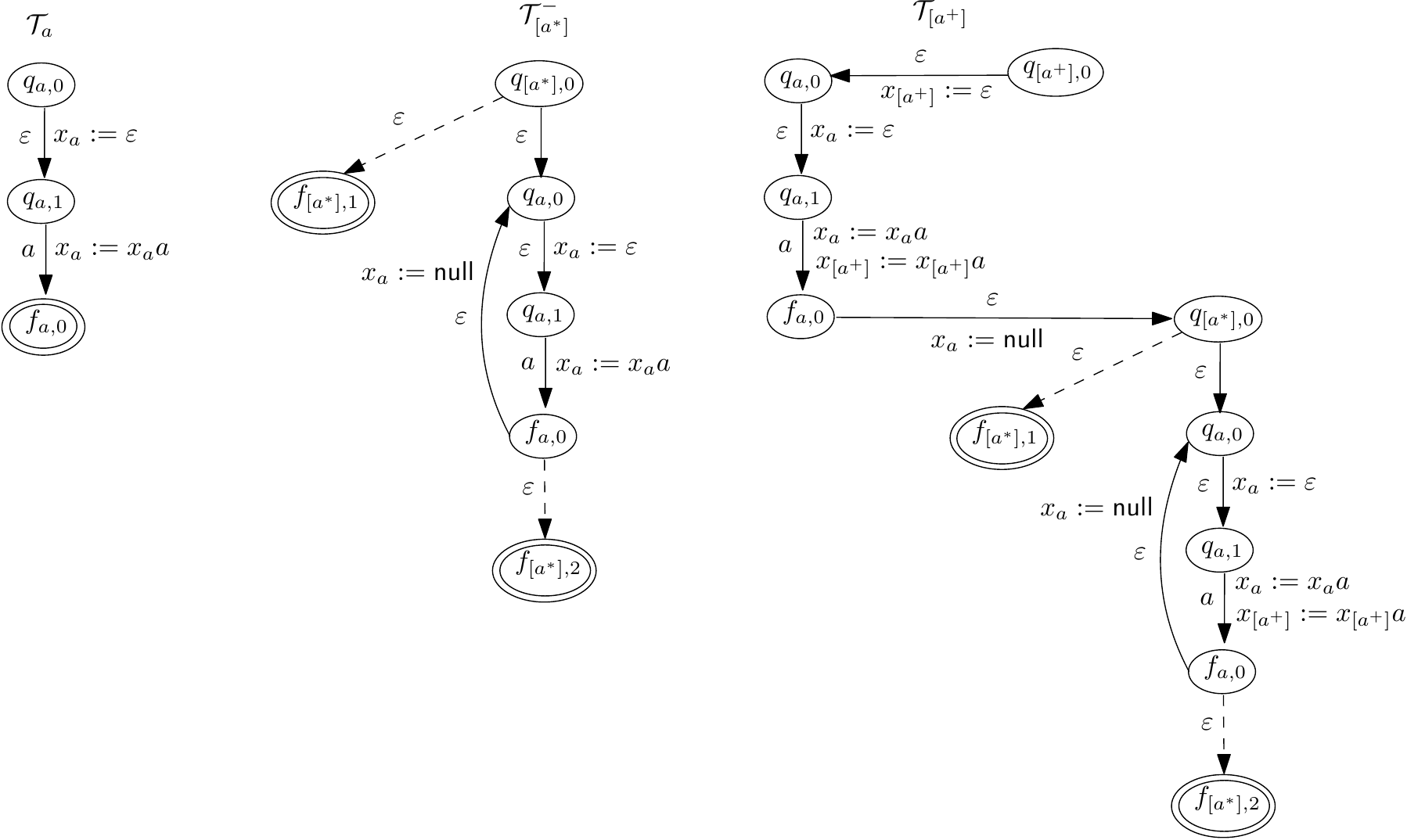}
		\caption{The PSST $\cT_e$ for $e = [a^+]$}
		\label{fig-pfa}
%		\vspace{-4mm}
	\end{figure}
\end{example}

%\zhilin{the necessity of final-state-set splitting is illustrated below.}
The following example illustrates the necessity of splitting final states into two disjoint subsets.

\begin{example}\label{exmp-psst-partition}
Consider {\pcre} $e = [([a^{*?}])^{*}]$. If we execute ``$aaa$''.match(/($a$*?)*/)  in node.js, then 
the result is the array $[$``$aaa$'', ``$a$''$]$, which means ($a$*?)* is matched to ``$aaa$'' and ($a$*?) is matched to $a$. If we did not split the set of final states into two disjoint subsets, we would have obtained a {\PSST} $\cT^\prime_e$ as illustrated in Fig.~\ref{fig-psst-partition}, to simulate the matching of $e$ against words. The accepting run of $\cT^\prime_e$ on $w  = aaa$ is 
%
%\zhilei{The run on $aaa$ should be complete, with $x_e = aaa$ and $x_{([a^{*?}])} = \varepsilon$}
$$
\begin{array}{l}
q_{[([a^{*?}])^{*}]} \xrightarrow{\varepsilon} 
q_{([a^{*?}]), 0} \xrightarrow{\varepsilon} 
f_{([a^{*?}])} \xrightarrow{\varepsilon} 
q_{([a^{*?}]), 0} \xrightarrow{\varepsilon} 
q_{a, 0} \xrightarrow{\varepsilon}
q_{a, 1} \xrightarrow{a} 
f_a \xrightarrow{\varepsilon} 
f_{([a^{*?}])} \xrightarrow{\varepsilon} 
\\
\hspace*{8mm} 
q_{([a^{*?}]), 0} \xrightarrow{\varepsilon} 
q_{a, 0} \xrightarrow{\varepsilon} 
q_{a, 1} \xrightarrow{a} 
f_a \xrightarrow{\varepsilon} 
f_{([a^{*?}])} \xrightarrow{\varepsilon} 
q_{([a^{*?}]), 0} \xrightarrow{\varepsilon} 
q_{a, 0} \xrightarrow{\varepsilon} 
q_{a, 1} \xrightarrow{a} 
f_a \xrightarrow{\varepsilon} 
f_{([a^{*?}])} \xrightarrow{\varepsilon} 
\\
\hspace*{12mm} 
q_{([a^{*?}]), 0} \xrightarrow{\varepsilon} 
f_{([a^{*?}])} \xrightarrow{\varepsilon} 
f_{[([a^{*?}])^*]},
\end{array}
$$
where $x_{e} = aaa$ and $ x_{([a^{*?}])} = \varepsilon$, namely, $e $ is matched to ``$aaa$''  and  $([a^{*?}])$ is matched to $\varepsilon$. Therefore, the semantics of $e$ defined by $\cT^\prime_e$  is \emph{inconsistent} with semantics of /($a$*?)*/ in node.js. Intuitively, the semantics of /($a$*?)*/ in node.js requires that either it is matched to $\varepsilon$ in whole and the subexpression $a$*? is \emph{not} matched at all, or it is matched to a concatenation of \emph{non-empty} strings each of which matches $a$*?. This semantics can be captured by (adapted) {\PSST}s where the set of final states is split into two disjoint subsets.
	\begin{figure}[tb]
%		\vspace{-2mm}
		\centering
		%\rule{\linewidth}{0cm}
		\includegraphics[width=0.9\textwidth]{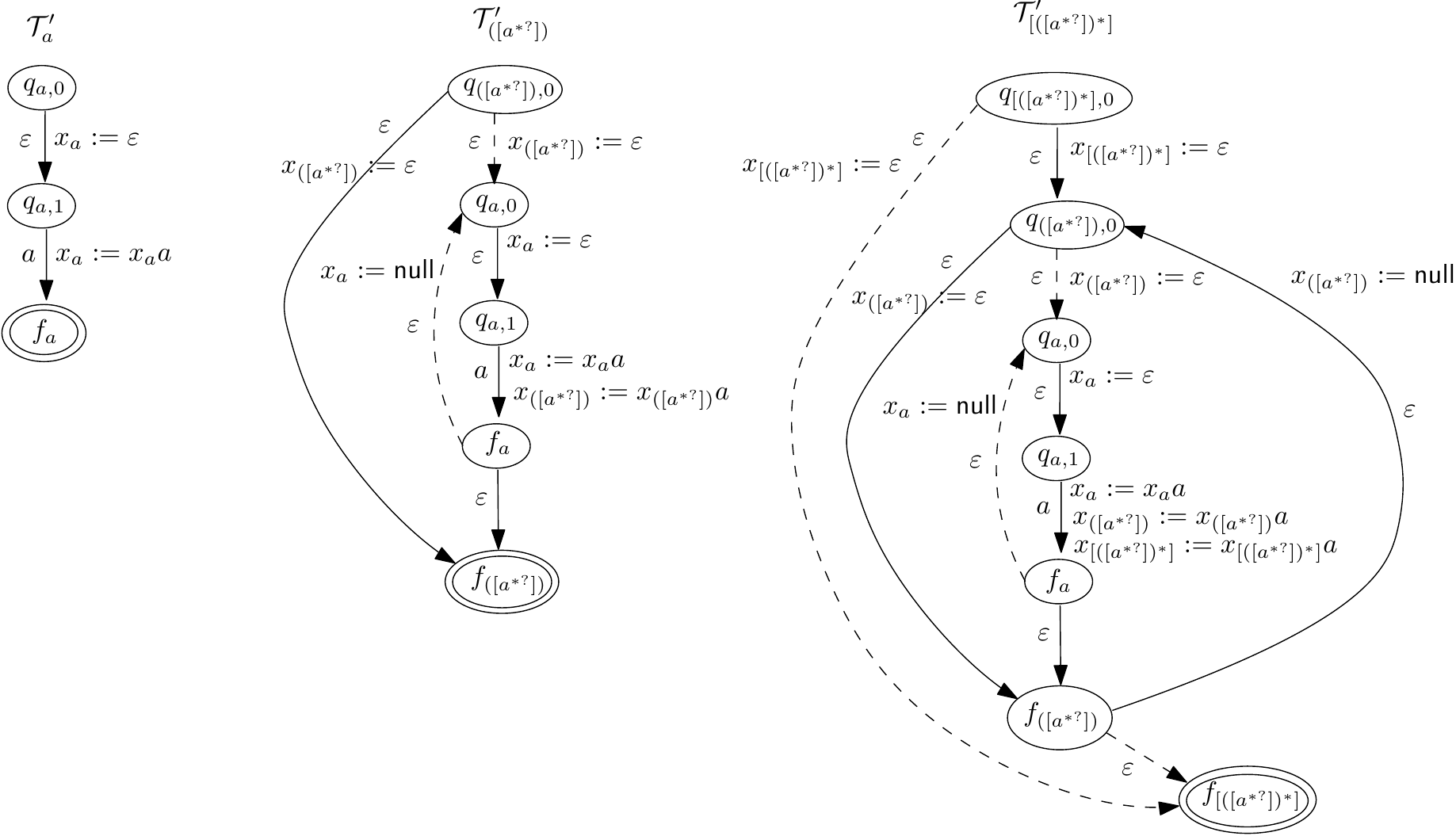}
		\caption{The PSST $\cT^\prime_e$ for $e = [([a^{*?}])^{*}]$ with a single set of final states}
		\label{fig-psst-partition}
%		\vspace{-4mm}
	\end{figure}
\end{example}

%For example, the expression $(a^{*?})^*$ either matches $\varepsilon$ holistically and in this case the subexpression $a^{*?}$ is not entered and matched at all, or it matches the concatenation of a number of \emph{non-empty} strings which are accepted by $a^{*?}$. Note that this is an instance of problematic regular expression in the setting of \cite{BDM14, BM17}. By splitting the set of final states, we keep track of the emptiness of subexpressions during the match, and thus give a more fine-grained semantics to those operators and successfully handle this class of expression.

%%%%%%%%%%%%%%%%%%%%%%%%%%%%%%%%
%\color{red}
%% a discussion of lookahead etc. 
%\begin{remark}
%Notice that $\extract$ and $\replace$ use regexes with non-classical semantics whereas %. It is previously not known wPre-image recognizability was not known; 
%previous results considered classical regular expressions without references (a non-example is Example \ref{exmp-name-swap}). We exploit PSSTs to capture the new semantics and prove pre-image recognizability. Handling other features such as lookahead/backreferences is non-trivial, partly due to the undecidability of the resulting SL fragment (e.g., with backreferences). 
%One may investigate alternating PSSTs or product constructions for lookahead, and extend PSSTs to inspect variable contents for backreferences, which are left as future work.
%\end{remark}
%\color{black}

%%%%%%%%%%%%%%%%%%%%%%%%%%%%%%%%%%%%%%%%%%
%\subsection{Validation experiments for the formal semantics} \label{sect：valid}

\paragraph{Validation experiments for the formal semantics} \label{sect：valid}
We have defined %the formal semantics of the regular-expression 
{\regexp}-string matching by constructing {\PSST}s. %out of regular expressions. 
In the sequel, we conduct experiments to validate the formal semantics against the actual JavaScript {\regexp}-string matching.

Let $\opset$ denote the set of {\regexp} operators: alternation $+$, concatenation $\concat$, optional $?$, lazy optional $??$, Kleene star $*$, lazy Kleene star $*?$, Kleene plus $+$, lazy Kleene plus $+?$, repetition $\{m_1,m_2\}$, and lazy repetition $\{m_1,m_2\}?$. Moreover, let $\opset^{2}$ (resp. $\opset^{3}$) denote the set of pairs (resp. triples) of operators from $\opset$. 
Aiming at a good coverage of different syntactical ingredients of {\regexp}, we generate regular expressions for every element of $\opset^{\le 3} = \opset \cup \opset^2 \cup \opset^3$.
As arguments of these operators, we consider the following character sets: $\mathbb{S} = \{$a, $\ldots$, z$\}$, $\mathbb{C}=\{$A, $\ldots$, Z$\}$, $\mathbb{D} = \{0,\ldots,9\}$, and $\mathbb{O}$, the set of ASCII symbols not belonging to $\mathbb{S} \cup \mathbb{C} \cup \mathbb{D}$.
Intuitively, these character sets correspond to JavaScript character classes [a-z], [A-Z], [0-9], and [{\textasciicircum}a-zA-Z0-9] (where {\textasciicircum} denotes complement).
Moreover, for the regular expression generated for each element of $\opset^{\le 3}$, we set the subexpression corresponding to its first component as the capturing group. 
For instance, for the pair $(*?, *)$, we generate the {\regexp} $[([\mathbb{S}^{*?}])^{*}]$. In the end, we generate $10+10*10+10*10*10 = 1110$ {\regexp}es. 

For each generated {\regexp} $e$, we construct a PSST $\psst_e$, whose output corresponds to the matching of the first capturing group in $e$.  Moreover, we generate from $\psst_e$ an input string $w$ as well as the corresponding output $w'$. We require that the length of $w$ is no less than some threshold (e.g., $10$), in order to avoid the empty string and facilitate a  meaningful comparison with the actual semantics of JavaScript regular-expression matching. 
Let {\sf reg} be the JavaScript regular expression corresponding to $e$. Then we execute the following JavaScript program $\cP_{e,w}$,
\begin{center}
{
\small
\begin{minted}{javascript}
      var x = w; console.log(x.match(reg)[1]);
\end{minted}
}
\end{center}
and confirm that its output is equal to $w'$, thus validating that the formal semantics of  \regexp-string matching defined by PSSTs is consistent with the actual semantics of JavaScript {\sf match} function. For instance, for the {\regexp} expression $[([\mathbb{S}^{*?}])^{*}]$, we generate from the $\psst_e$ the input string $w= aaaaaaaaaa$, together with the output $a$. Then we generate  the JavaScript program from ${\sf reg}$ and $w$, execute it, and obtain the same output $a$.

In all the generated {\regexp}s, we confirm the consistency of the formal semantics of  \regexp-string matching defined by PSSTs with the actual JavaScript semantics, namely, for each {\regexp}  $e$, the output of the PSST $\psst_e$ on $w$ is equal to the output of the JavaScript program $\cP_{e,w}$.

\subsection{Modeling string functions by PSSTs}

The $\extract$, $\replace$ and $\replaceall$ functions can be accurately modeled using PSSTs.
That is, we can reduce satisfiability of our string logic to satisfiability of a logic containing only concatenation, PSST transductions, and membership of regular languages.

\begin{lemma}\label{lem-str-fun-to-psst}
    The satisfiability of $\strline$ reduces to the satisfiability of boolean combinations of formulas of the form $z=x \concat y$, $y=\cT(x)$, and $x \in \cA$, where $\cT$ is a PSST and $\cA$ is an FA.
\end{lemma}

First, observe that regular constraints (aka membership queries) $x \in e$ can be reduced to FA membership queries $x \in \cA$ using standard techniques.
Features such a greediness and capture groups do not affect whether a word matches a {\regexp}, they only affect \emph{how} a string matches it.
Thus, for  regular constraints, these features can be ignored and a standard translation from regular expressions to finite automata can be used.

The $\extract_{i,e}$ function can be defined by a PSST $\cT_{i,e}$ obtained from the PSST $\cT_e$ (see Section~\ref{sect:regextopsst}) by removing all string variables, except $x_{e'}$, where $e'$ is the subexpression of $e$ corresponding to the $i$th capturing group, and setting the output expression of the final states as $x_{e'}$.

We give a sketch of the encoding of $\replaceall$ here.
Full formal details are given in \ifproceeding the long version of the paper \cite{popl22-full}. \else Appendix~\ref{appendix:sec-extract-replace-to-psst}. \fi
The encoding of $\replace$ is almost identical to that of $\replaceall$.

A call $\replaceall_{\pat, \rep}(x)$ replaces every match of $\pat$ by a value determined by the replacement string $\rep$.
Recall, $\rep$ may contain references $\$i$, $\refbefore$, or $\refafter$.
The first step in our reduction to PSSTs is to eliminate the special references $\$0$, $\refbefore$, and $\refafter$.
In essence, this simplification uses PSST transductions to insert the contextual information needed by $\refbefore$ and $\refafter$ alongside each substring that will be replaced.
Then, the call to $\replaceall$ can be rewritten to include this information in the match, and use standard references ($\$i$) in the replacement string.
The reference $\$0$ can be eliminated by wrapping each pattern with an explicit capturing group.

We show informally how to construct the PSST for $\replaceall_{\pat, \rep}$ where all the references in $\rep$ are of the form $\$i$ with $i > 0$.
The full reduction is 
\ifproceeding given in the long version of the paper \cite{popl22-full}.
\else given in the appendix. \fi

Let $\rep = w_1 \$i_1 w_2 \cdots w_k \$i_k w_{k+1}$. For each $j \in [k]$, we introduce a \emph{fresh} string variable $y_{j}$. 
Let us use $\rep[(y_1, \cdots, y_k)/(\$i_1,\cdots, \$i_k)]$ to denote the sequence $w_1 y_1 w_2 \cdots w_k y_k w_{k+1}$.
For instance, if $\rep = a \$1 a \$2 a \$1 a$, then 
$\rep[(y_1, y_2, y_3)/(1,2, 1)] = a y_1 a y_2 a y_3 a$. 
%is $a w_1 a w_2 a w_3 a$ when $w_j$ is the value stored in $x_{e'_{i_j}}$.
%
Moreover, let $e'_{i_1},\ldots, e'_{i_k}$ be the subexpressions of $\pat$ corresponding to the $i_1$th, $\ldots$, $i_k$th capturing groups.
Note here we use mutually distinct fresh variables $y_1, \cdots, y_k$ for $\$i_1, \cdots, \$i_k$, even if $i_j$ and $i_{j'}$ may be equal for $j \neq j'$. 
We make this choice for the purpose of satisfying the copyless property~\cite{AC10} of PSSTs, which leads to improved complexity results in some cases (discussed in the sequel).
%If we tried to use a single variable – e.g. $\rep[x/\$i_1,x/\$i_2]$ – then the resulting transition in the encoding below would not be copyless.
If we tried to use the same variable for the different occurrences of the same reference – then the resulting transition in the encoding below would not be copyless.
Moreover, the construction below guarantees that the values of different variables for the multiple occurrences of the same reference are actually the same.

%Although, in this notation, it is possible to substitute different occurrences of a reference with different word values, our encoding will always substitute the same value.
%That is, if $i_j = i_{j'}$ then $x_j$ will contain the same value as $x_{j'}$.
%We use two variables in this case to satisfy the copyless property~\cite{AC10}, which leads to improved complexity results in some cases (discussed in the sequel).
%If we tried to use a single variable – e.g. $\rep[x/\$i_1,x/\$i_2]$ – then the resulting transition in the encoding below would not be copyless.

%By abuse of notation, we will write
%$\rep[x_{e'_{i_1}}/\$i_1,\cdots,x_{e'_{i_k}}/\$i_k]$
%to denote the replacement of the references, in order of appearance, by the contents of the variables $x_{e'_{i_j}}$.
%E.g., if
%$\rep = a \$1 a \$2 a \$1 a$
%then
%$\rep[x_{e'_{i_1}}/\$i_1, x_{e'_{i_2}}/\$i_2,x_{e'_{i_3}}/\$i_3]$
%is $a w_1 a w_2 a w_3 a$ when $w_j$ is the value stored in $x_{e'_{i_j}}$.

 Suppose $\cT_\pat = (Q_{\pat}, \Sigma, X_{\pat}, \delta_{\pat}, \tau_{\pat}, E_\pat, q_{\pat,0}, (F_{\pat,1}, F_{\pat,2}))$.
 %, $\rep = w_1 \$i_1 w_2 \cdots w_k \$i_k w_{k+1}$, and $e'_{i_1},\ldots, e'_{i_k}$ are the subexpressions of $\pat$ corresponding to the $i_1$th, $\ldots$, $i_k$th capturing groups.
Then $\cT_{\replaceall_{\pat,\rep}}$ is obtained from $\cT_\pat$ by adding the fresh string variables $y_1, \cdots, y_k$ and a fresh state $q'_0$ such that (see Fig.~\ref{fig-psst-replaceall})
\begin{itemize}
    \item $\cT_{\replaceall_{\pat,\rep}}$ goes from $q'_0$ to $q_{\pat,0}$ via an $\varepsilon$-transition of higher priority than the non-$\varepsilon$-transitions, in order to search the first match of $\pat$ starting from the current position,
    \item when $\cT_{\replaceall_{\pat,\rep}}$ stays at $q'_0$, it keeps appending the current letter to the end of $x_0$, which stores the output of $\cT_{\replaceall_{\pat,\rep}}$,
    \item starting from $q_{\pat, 0}$, $\cT_{\replaceall_{\pat,\rep}}$ simulates $\cT_\pat$ and stores the matches of capturing groups of $\pat$ into the string variables (in particular,
    the matches of the $i_1$th, $\ldots$, $i_k$th capturing groups into the string variables $x_{e'_{i_1}}, \cdots, x_{e'_{i_k}}$ respectively), moreover, for each $j \in [k]$, $y_j$ is updated in the same way as $x_{e'_{i_j}}$ (in particular, for each transition $(q, a, q')$ in $\cT_\pat$ such that $E_\pat(q, a, q')(x_{e'_{i_j}}) = x_{e'_{i_j}} a$,  we have $E_\pat(q, a, q')(y_j) = y_j a$),

    \item when the first match of $\pat$ is found, $\cT_{\replaceall_{\pat,\rep}}$ goes from $f_{\pat,1} \in F_{\pat, 1}$ or $f_{\pat, 2} \in F_{\pat, 2}$ to $q'_0$ via an $\varepsilon$-transition, it then appends the replacement string, which is $\rep[(y_1, \cdots, y_k)/(\$i_1,\cdots, \$i_k)]$, to the end of $x_0$, resets the values of all the string variables, except $x_0$, to $\nullchar$, and keeps searching for the next match of $\pat$.
\end{itemize}

It may be observed that the PSST will be copyless.
That is, the value of a variable is not copied to two or more variables during a transition.
In all but the last case, variables are only copied to themselves, via assignments of the form $x_{e'} := x_{e'} a$, $x_{e'} := x_{e'}$, $x_{e'} := \varepsilon$, or $x_{e'} := \nullchar$.
In the final case, when a replacement is made, the assignments are
%$x_0 = x_0 w_1 x_{e'_{i_1}} w_2 \cdots w_k x_{e'_{i_k}} w_{k+1}$
$x_0 := x_0 w_1y_1 w_2 \cdots w_k y_k w_{k+1}$
and
$x^\prime := \nullchar$ for all the variables $x^\prime \in X_\pat \cup \{y_1, \cdots, y_k\}$.
Again, only one copy of the value of each variable is retained.

Copyful PSSTs are only needed when removing $\refbefore$ and $\refafter$ from the replacement strings.
To see this, consider the prefix preceding the first replacement in a string.
If $\refbefore$ appears in the replacement string, this prefix will be copied an unbounded number of times (once for each matched and replaced substring).
Conversely, references of the form $\$i$ are ``local'' to a single match.
By having a separate variable for each occurence of $\$i$ in the replacement string, we can avoid having to make copies of the values of the variables.

\begin{figure}[ht]
%	\vspace{-2mm}
    \centering
    \includegraphics[scale=0.7]{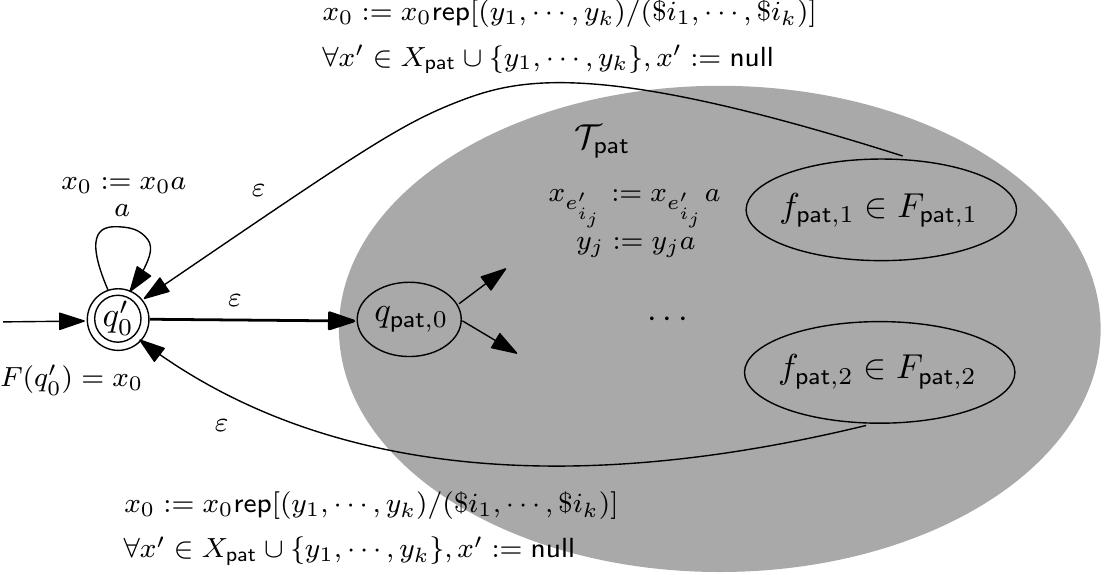}
    \caption{The PSST $\cT_{\replaceall_{\pat,\rep}}$}
    \label{fig-psst-replaceall}
%    \vspace{-2mm}
\end{figure}

%%%%%%%%%%%%%%%%%%%%%%%%%%%%

\section{A Propagation-Based Calculus for String Constraints}
\label{sect:calculus}

\begin{table}
  \small
  \caption{Rules of the one-sided sequent calculus. A term
    $e^c$ denotes the complement of a regular expression~$e$,  i.e.,
    ~$\lang{e^c} = \Sigma^* \setminus \lang{e}$.}
  \label{tab:calculus}

  \begin{gather*}
    \infer[$\wedge$]
    {\seqq{\varphi \wedge \psi}}
    {\seqq{\varphi, \psi}}
    \quad
    \infer[$\neg\vee$]
    {\seqq{ \neg(\varphi \vee \psi)}}
    {\seqq{ \neg\varphi, \neg\psi}}
    \quad
    \inferii[$\vee$]
    {\seqq{\varphi \vee \psi}}
    {\seqq{\varphi}}{\seqq{\psi}}
    \quad
    \inferii[$\neg\wedge$]
    {\seqq{\neg(\varphi \wedge \psi)}}
    {\seqq{\neg\varphi}}{\seqq{\neg\psi}}
    \quad
    \infer[$\neg\neg$]
    {\seqq{ \neg\neg\varphi}}
    {\seqq{\varphi}}
    \\[1ex]
    \infer[$\not\in$]
    {\seqq{x \not\in e}}
    {\seqq{x \in e^c}}
    \quad
    \inferC[\ruleName{$\not=$}]{\text{where~} y \text{~is fresh}}
    {\seqq{x \not= f(x_1, \ldots, x_n)}}
    {\seqq{x \not= y, y = f(x_1, \ldots, x_n)}}
    \quad
    \inferii[Cut]
    {\seqq{}}
    {\seqq{x \in e}}{\seqq{x \in e^c}}
    \\[2ex]
    \infer[$=$-Prop]
    {\seqq{x \in e, x = y}}
    {\seqq{x \in e, x = y, y \in e}}
    \qquad
    \inferC[$\not=$-Subsume]{\text{if~} \lang{e_1} \cap \lang{e_2} = \emptyset}
    {\seqq{x \in e_1, x \not= y, y \in e_2}}
    {\seqq{x \in e_1, y \in e_2}}
    \\[2ex]
    \inferC[$=$-Prop-Elim]{\text{if~} |\lang{e}| = 1}
    {\seqq{x \in e, x = y}}
    {\seqq{x \in e, y \in e}}
    \qquad
    \inferC[$\not=$-Prop-Elim]{\text{if~} |\lang{e}| = 1}
    {\seqq{x \in e, x \not= y}}
    {\seqq{x \in e, y \in e^c}}
    \\[3ex]
    \begin{array}{@{}c@{~}l@{}}
      \infer[Close]
      {\seqq{x \in e_1, \ldots, x \in e_n}}
      {}
      &
        \text{if~} \lang{e_1} \cap \cdots \cap\lang{e_n} = \emptyset
      \\[2ex]
      \infer[Subsume]
      {\seqq{x \in e, x \in e_1, \ldots, x \in e_n}}
      {\seqq{x \in e_1, \ldots, x \in e_n}}
      &
        \text{if~} \lang{e_1} \cap \cdots \cap\lang{e_n} \subseteq \lang{e}
      \\[2.5ex]
      \infer[Intersect]
      {\seqq{x \in e_1, \ldots, x \in e_n}}
      {\seqq{x \in e}}
      &
        \text{if~}
        \begin{array}{l}
          n > 1 \text{~and~}
          \\
          \lang{e_1} \cap \cdots \cap\lang{e_n} = \lang{e}
        \end{array}
      \\[4ex]
      \infer[Fwd-Prop]
      {\seqq{x = f(x_1, \ldots, x_n), x_1 \in e_1, \ldots, x_n \in e_n}}
      {\seqq{x \in e, x = f(x_1, \ldots, x_n), x_1 \in e_1, \ldots, x_n \in e_n}}
      &
        \text{if~} \lang{e} = f(\lang{e_1}, \ldots, \lang{e_n})
      \\[3ex]
      \infer[Fwd-Prop-Elim]
      {\seqq{x = f(x_1, \ldots, x_n), x_1 \in e_1, \ldots, x_n \in e_n}}
      {\seqq{x \in e, x_1 \in e_1, \ldots, x_n \in e_n}}
      &
        \text{if~}
        \begin{array}{l}
          \lang{e} = f(\lang{e_1}, \ldots, \lang{e_n})
          \\
          \text{and~} 
          |\lang{e}|= 1
          \end{array}
      \\[3ex]
      \infer[Bwd-Prop]
      {\seqq{x \in e, x = f(x_1, \ldots, x_n)}}
      {\big\{\seqq{x \in e, x = f(x_1, \ldots, x_n),
      x_1 \in e_1^i, \ldots, x_n \in e_n^i}\big\}_{i=1}^k}
      &
        \text{if~}
        \begin{array}{l}
          f^{-1}(\lang{e}) = \\
          \bigcup_{i=1}^k \big( \lang{e_1^i} \times \cdots \times \lang{e_n^i} \big)
          \end{array}
    \end{array}
  \end{gather*}
%  \vspace{-4mm}
\end{table}

We now introduce our calculus for solving string constraints in
$\strline$ (see Table~\ref{tab:calculus}), state its correctness, and observe that it gives rise to a
decision procedure for the fragment $\strlinesl$ of straightline
formulas. The calculus is based on the principle of propagating
regular language constraints by computing images and pre-images of
string functions. We deliberately keep the calculus minimalist and
focus on the main proof rules; for an implementation, the calculus has
to be complemented with a suitable strategy for applying the rules, as
well as standard SMT optimizations such as non-chronological
back-tracking and conflict-driven learning. An implementation also has
to choose a suitable effective representation of \regexp{}
membership constraints, for instance using finite-state automata.\footnote{
    Recall features such as greediness do not need to be modeled for simple membership queries as they do not change the accepted language.
}
In particular, we use the fact that---for membership---\regexp{} can be complemented.
We denote the complement of $e$ in a membership constraint by $e^c$.
Our calculus is parameterized in the set of considered string
functions; in this paper, we work with the set
$\{\concat, \extract, \replace, \replaceall\}$ consisting of
concatenation, extraction, and replacement, but this set can be
extended by other functions for which images and/or pre-images can be
computed (see Section~\ref{sec:rules}).

\subsection{Sequents and Examples}

The calculus operates on \emph{one-sided sequents,} and can be
interpreted as a sequent calculus in the sense of
Gentzen~\cite{Gentzen35} in which all formulas are located in the
antecedent (to the left of the turnstile~$\vdash$). A one-sided sequent is a
finite set $\Gamma \subseteq \strline$ of string constraints. For sake
of presentation, we write sequents as lists of formulas separated by
comma, and $\seqq{\varphi_1, \ldots, \varphi_n}$ for the union
$\Gamma \cup \{\varphi_1, \ldots, \varphi_n\}$. We say that a
sequent~$\seqq{}$ is \emph{unsatisfiable} if $\bigwedge \Gamma$ is
unsatisfiable. Our calculus is refutational and has the purpose of
either showing that some initial sequent~$\seqq{}$ is unsatisfiable,
or that it is satisfiable by constructing a solution for it. A
solution is a
sequent~$\seq{x_1 \in w_1, x_2 \in w_2, \ldots, x_n \in w_n}$ that
defines the values of string variables using \regexps{} that only consist
of single words.

\begin{figure}
  \begin{prooftree}
    \AxiomC{}
    \LeftLabel{\ruleName{Close}}
    \UnaryInfC{$x \in a^+\Sigma^*,
      x = y \concat z, y \in a^+, z \in \Sigma^*,
      x \in b^+(a^c)^*, x = \replaceall_{a,
        b}(x)$}
    \LeftLabel{\ruleName{Fwd-Prop}}
    \UnaryInfC{$x \in a^+\Sigma^*,
      x = y \concat z, y \in a^+, z \in \Sigma^*, x = \replaceall_{a,
        b}(x)$}
    \LeftLabel{\ruleName{Fwd-Prop}}
    \UnaryInfC{$x = y \concat z, y \in a^+, z \in \Sigma^*, x = \replaceall_{a,
        b}(x)$}
    \LeftLabel{\ruleName{$\wedge^*$}}
    \UnaryInfC{$x = y \concat z \wedge y \in a^+ \wedge z \in \Sigma^*
      \wedge x = \replaceall_{a,
    b}(x)$}
  \end{prooftree}
%  \vspace{-2mm}
  \caption{Proof of unsatisfiability for \eqref{eq:calcEx1} in
    Example~\ref{ex:calc1}}
  \label{fig:calcEx1}
%  \vspace{-2mm}
\end{figure}

\begin{figure}
  \begin{prooftree}
    \AxiomC{$x \in a, z \in a, y \in \epsilon, r \in b$}
    \LeftLabel{\ruleName{Subsume$\mbox{}^*$}}
    \UnaryInfC{$x \in a, z \in a, y \in \epsilon, r \in b, \ldots$}
    \LeftLabel{\ruleName{FPE}}
    \UnaryInfC{$z \in a, y \in \epsilon, x \in a, r = \replaceall_{a, b}(x), \ldots$}
    \LeftLabel{\ruleName{FPE}}
    \UnaryInfC{\rule{0em}{1.6ex}$z \in a, y \in \epsilon, x = y \cdot z, \ldots$}
    \AxiomC{$\vdots$}
    \UnaryInfC{$z \in a^c, \ldots$}
    \LeftLabel{\ruleName{Cut}}
    \BinaryInfC{$y \in \epsilon, z \in a^+, x = y \concat z, 
      x \in a^+, \ldots$}
    \AxiomC{$\vdots$}
    \UnaryInfC{$y \in a^+, z \in a^*, \ldots$}
    \LeftLabel{\ruleName{Bwd-Prop}}
    \BinaryInfC{$x = y \concat z, x \in a^+, r = \replaceall_{a, b}(x)$}
    \LeftLabel{\ruleName{$\wedge^*$}}
    \UnaryInfC{$x = y \concat z \wedge x \in a^+ \wedge r = \replaceall_{a, b}(x)$}
  \end{prooftree}
%  \vspace{-2mm}
  \caption{Proof of satisfiability for \eqref{eq:calcEx2} in
    Example~\ref{ex:calc2}. \ruleName{FPE} stands for \ruleName{Fwd-Prop-Elim}}
  \label{fig:calcEx2}
%  \vspace{-4mm}
\end{figure}

\begin{example}
  \label{ex:calc1}
  We first illustrate the calculus by showing unsatisfiability of the
  constraint\footnote{Note here for convenience, in the regular constraints $x \in e$, we write $e$ as in classical regular expressions and do not strictly follow the syntax of $\strline$, since in this case, only the language defined by $e$ matters. }:
  %\philipp{how should we write regex $\Sigma^*$?}\zhilin{added a footnote}
  %\zhilin{Two questions here: Is it better that the notation here is consistent with that in STR, should we use the motivating example to illustrate the calculus ?}
  \begin{equation}
    \label{eq:calcEx1}
    x = y \concat z \wedge y \in a^+ \wedge z \in \Sigma^*
    \wedge x = \replaceall_{a, b}(x)
  \end{equation}
  To this end, we construct a proof tree that has \eqref{eq:calcEx1}
  as its root, by applying proof rules until all proof goals have been
  closed (Fig.~\ref{fig:calcEx1}). The proof is growing upward, and
  is built by first eliminating the conjunctions~$\wedge$, resulting
  in a list of formulas. Next, we apply the rule~\ruleName{Fwd-Prop} for
  \emph{forward-propagation} of a regular expression constraint. Given
  that $y \in a^+, z \in \Sigma^*$, from the
  equation~$x = y \concat z$ we can conclude that $x \in
  a^+\Sigma^*$. From $x \in a^+\Sigma^*$ and
  $x = \replaceall_{a, b}(x)$, we can next conclude that
  $x \in b^+(a^c)^*$, i.e., $x$ starts with $b$ and
  cannot contain the letter~$a$. Finally,
  the proof can be closed because the languages~$a^+\Sigma^*$ and
  $b^+(a^c)^*$ are disjoint.
\end{example}

\begin{example}
  \label{ex:calc2}
  We next consider the case of a satisfiable formula in $\strlinesl$:
  \begin{equation}
    \label{eq:calcEx2}
    x = y \concat z \wedge x \in a^+ \wedge r = \replaceall_{a, b}(x)
  \end{equation}
  Fig.~\ref{fig:calcEx2} shows how a solution can be constructed for
  this formula. The strategy is to first derive constraints for the
  variables~$y, z$ whose value is not determined by any
  equation. Given that $x \in a^+$, from the equation $x = y \cdot z$
  we can derive that either $y \in \epsilon, z \in a^+$ or
  $y \in a^+, z \in a^*$, using rule~\ruleName{Bwd-Prop}. We focus on
  the left branch~$y \in \epsilon, z \in a^+$. Since propagation is
  not able to derive further information for $y, z$, and no
  contradiction was detected, at this point we can conclude
  satisfiability of \eqref{eq:calcEx2}. To construct a solution, we
  pick an arbitrary value for $z$ satisfying the
  constraint~$z \in a^+$, and use \ruleName{Cut} to add the
  formula~$z \in a$ to the branch. Again following the left branch, we
  can then use \ruleName{Fwd-Prop-Elim} to evaluate $x = y\cdot z$ and
  add the formula~$x \in a$, and after that $r \in b$ due to
  $r = \replaceall_{a, b}(x)$. Finally, \ruleName{Subsume} is used to
  remove redundant \regexp{} constraints from the proof goal. The
  resulting sequent (top-most sequent on the left-most branch) is a
  witness for satisfiability of \eqref{eq:calcEx2}.
\end{example}

\subsection{Proofs and Proof Rules}
\label{sec:rules}

More formally, proof rules are relations between a finite list of
sequents (the premises), and a single sequent (the conclusion). Proofs
are finite trees growing upward, in which each node is labeled with a
sequent, and each non-leaf node is related to the node(s) directly
above it through an instance of a proof rule. A proof branch is a path
from the proof root to a leaf. A branch is closed if a closure rule (a
rule without premises) has been applied to its leaf, and open
otherwise. A proof is closed if all of its branches are closed.

The proof rules of the calculus are shown in
Table~\ref{tab:calculus}. The first row shows standard proof rules to
handle Boolean operators; see, e.g.,
\cite{DBLP:books/daglib/0022394}. Rule~\ruleName{$\not\in$} turns
negated membership predicates into positive ones through
complementation, and rule~\ruleName{$\not=$} negative function
applications into positive ones. As a result, only disequalities
between string variables remain. The rule~\ruleName{Cut} can be used to
introduce case splits, and is mainly needed to extract solutions once
propagation has converged (as shown in Example~\ref{ex:calc2}).

The next four rules handle equations between string
variables. Rule~\ruleName{=-Prop} propagates \regexp{} constraints from
the left-hand side to the right-hand side of an equation;
\ruleName{=-Prop-Elim} in addition removes the equation in the case
where the propagated constraint has a unique solution. The
rule~\ruleName{$\not$=-Prop-Elim} similarly turns a singleton \regexp{}
for the left-hand side of a disequality into a \regexp{} constraint on the
right-hand side.  As a convention, we allow application of
\ruleName{=-Prop}, \ruleName{=-Prop-Elim}, and
\ruleName{$\not$=-Prop-Elim} in both directions, left-to-right and
right-to-left of equalities/disequalities. Finally, \ruleName{$\not=$-Subsume} eliminates
disequalities that are implied by the \regexp{} constraints of a proof
goal.

The rule~\ruleName{Close} closes proof branches that contain
contradictory \regexp{} constraints, and is the only closure rule needed
in our calculus. \ruleName{Subsume} removes \regexp{} constraints
that are implied by other constraints in a sequent, and
\ruleName{Intersect} replaces multiple \regexps{} with a single
constraint.

The last three rules handle applications of functions
$f \in \{\concat, \extract, \replace, \replaceall\}$ through
propagation. Rule~\ruleName{Fwd-Prop} defines forward propagation, and
adds a \regexp{} constraint~$x \in e$ for the value of a function by
propagating constraints about the arguments. The \regexp{}~$e$ encodes
the image of the argument \regexps{} under $f$:
\begin{definition}[Image]
  For an $n$-ary string
  function~$f : \Sigma^* \times \cdots \times \Sigma^* \to \Sigma^*$
  and languages $L_1, \ldots, L_n \subseteq \Sigma^*$, we define the
  \emph{image} of $L_1, \ldots, L_n$ under $f$ as
  $f(L_1, \ldots, L_n) = \{f(w_1, \ldots, w_n) \in \Sigma^* \mid w_1
  \in L_1, \ldots, w_n \in L_n \}$.
\end{definition}

Forward propagation is often useful to prune proof branches. It is
easy to see, however, that the images of regular languages under the
functions considered in this paper are not always regular; for
instance, $\replace_{\pat,\$0\$0}$ can map regular languages to
context-sensitive languages. In such cases, the side condition of
\ruleName{Fwd-Prop} cannot be satisfied by any \regexp{}~$e$, and the rule
is not applicable.

Rule~\ruleName{Fwd-Prop-Elim} handles the special case of forward
propagation producing a singleton language. In this case, the function
application is not needed for further reasoning and can be
eliminated. This rule is mainly used during the extraction of
solutions (as shown in Example~\ref{ex:calc2}).

Rule~\ruleName{Bwd-Prop} defines the dual case of backward
propagation, and derives \regexp{} constraints for function arguments
from a constraint about the function value. The argument constraints
encode the \emph{pre-image} of the propagated language:
\begin{definition}[Pre-image]
  For an $n$-ary string
  function~$f : \Sigma^* \times \cdots \times \Sigma^* \to \Sigma^*$
  and a language $L \subseteq \Sigma^*$, we define the
  \emph{pre-image} of $L$ under $f$ as the relation
  $f^{-1}(L) = \{(w_1, \ldots, w_n) \in (\Sigma^*)^n \mid 
  f(w_1, \ldots, w_n) \in L\}$.
\end{definition}
A \textbf{key result} of the paper is that pre-images of regular
languages under the functions considered in the paper can always be
represented in the
form~$\bigcup_{i=1}^k ( \lang{e_1^i} \times \cdots \times \lang{e_n^i}
)$, i.e., they are \emph{recognizable languages}~\cite{CCG06}. This implies that
\ruleName{Bwd-Prop} is applicable whenever a \regexp{} constraint for the
result of a function application exists, and prepares the ground for
the decidability result in the next section.  For concatenation,
recognizability was shown in \cite{Abdulla14,CHL+19}. This paper
contributes the corresponding result for all functions defined by
PSSTs:
\begin{lemma}[Pre-image of regular languages under PSSTs]
  \label{lem:psst_preimage}
  Given a PSST $\psst = (Q_T, \Sigma$, $X$, $\delta_T$, $\tau_T$, $E_T$,  $q_{0, T}$, $F_T$) and an \FA{} $\Aut
  = (Q_A, \Sigma, \delta_A, q_{0, A}, F_A)$, we can compute an \FA{} $\cB = (Q_B,
  \Sigma, \delta_B, q_{0, B}, F_B)$ in exponential time  such that $\Lang(\cB) = \cR^{-1}_{\cT}(\Lang(\Aut))$.
\end{lemma}
The proof of Lemma~\ref{lem:psst_preimage} is given in
\ifproceeding the long version of the paper \cite{popl22-full}.
\else Appendix~\ref{app-pre-image}. \fi 
Moreover, we have already shown in Lemma~\ref{lem-str-fun-to-psst} that $\extract$, $\replace$, and $\replaceall$ can be reduced to PSSTs.
We can finally observe that the calculus is sound:
\begin{lemma}[Soundness]
  The sequent calculus defined by Table~\ref{tab:calculus} is sound:
  (i) the root of a closed proof is an unsatisfiable sequent; and (ii)
  if a proof has an open branch that ends with a
  solution~$\seq{x_1 \in w_1, x_2 \in w_2, \ldots, x_n \in w_n}$, then
  the
  assignment~$\{x_1 \mapsto w_1, x_2 \mapsto w_2, \ldots, x_n \mapsto
  w_n\}$ is a satisfying assignment of the root sequent.
\end{lemma}

\begin{proof}
  By showing that each of the proof rules in Table~\ref{tab:calculus}
  is an equivalence transformation: the conclusion of a proof rule is
  equivalent to the disjunction of the premises.
\end{proof}

%%% Local Variables:
%%% mode: latex
%%% TeX-master: "main"
%%% TeX-command-extra-options: "-shell-escape"
%%% End:

%%%%%%%%%%%%%%%%%%%%%%%%%%%%

\subsection{Decision Procedure for $\strlinesl$} \label{sec:decision}

One of the main results of this paper is the decidability of the
$\strlinesl$ fragment of straightline formulas including
concatenation, extract, replace, and replaceAll:
\begin{theorem}\label{thm-main}
  Satisfiability of $\strlinesl$ formulas is decidable.
\end{theorem}

\begin{proof}
  We define a terminating strategy to apply the rules in
  Table~\ref{tab:calculus} to formulas in the $\strlinesl$
  fragment. The resulting proofs will either be closed, proving
  unsatisfiability, or have at least one satisfiable goal containing a
  solution:
\begin{itemize}
\item \emph{Phase 1:} apply the Boolean rules (first row of
  Table~\ref{tab:calculus}) to eliminate Boolean operators.
\item \emph{Phase 2:} apply rule~\ruleName{Bwd-Prop} to all regex
  constraints and all function applications on all proof
  branches. Whenever contradictory regex constraints occur in a proof
  goal, use \ruleName{Close} to close the branch. Also apply
  \ruleName{=-Prop} to systematically propagate constraints across
  equations. This phase terminates because $\strlinesl$
  formulas are acyclic.

  If all branches are closed as a result of Phase~2, the considered
  formula is unsatisfiable; otherwise, we can conclude satisfiability,
  and Phase~3 will extract a solution.
\item \emph{Phase 3:} select an open branch of the proof. On this
  branch, determine the set~$I$ of input variables, which are the
  string variables that do not occur as left-hand side of equations or
  function applications. For every $x \in I$, use rule~\ruleName{Cut}
  to introduce an assignment~$x \in w$ that is consistent with the
  regex constraints on $x$. Then systematically apply
  \ruleName{Subsume}, \ruleName{=-Prop}, \ruleName{Fwd-Prop-Elim} to
  evaluate remaining formulas and produce a
  solution.
\end{itemize}
%\vspace{-3ex}
\end{proof}

\paragraph*{Complexity analysis.} Because the pre-image computation for each PSST incurs an exponential blow-up in the size of the input automaton $\Aut$, the aforementioned decision procedure has a non-elementary complexity in the worst-case. In fact, this is optimal and a matching lower-bound is given in 
\ifproceeding the long version of the paper \cite{popl22-full}. \else Appendix~\ref{sec:tower-hard}. \fi

When $\refbefore$ and $\refafter$ are not used, the PSSTs in the reduction are copyless, and the exponential blow-up in the size of the input \FA{} $\Aut$ can be avoided.
%\mat{change}
That is, the pre-image automaton $\cB$ such that $\Lang(\cB) = \cR^{-1}_{\cT}(\Lang(\Aut))$ is exponential only in the size of $\cT$ and not the size of $\Aut$.
%\mat{endchange}
Hence, the exponentials do not stack on top of each other during the backwards analysis and the non-elementary blow-up is not necessary.
%\mat{change}
Since the PSST $\cT$ may be exponential in the size of the underlying regular expression, we may compute automata that are up to double exponential in size. The states of these automata can be stored in exponential space and the transition relation can be computed on the fly, giving an exponential space algorithm. More details are \ifproceeding given in the long version of this paper \cite{popl22-full}.
\else relegated to Appendix~\ref{app:expspace}.
\fi
%\mat{endchange}

Moreover, since the number of PSSTs is usually small in the path constraints of string-manipulating programs, the performance of the decision procedure is actually good on the benchmarks we tested, with the average running time per query a few seconds (see Section~\ref{sect:impl}).

\section{Implementation and Experiments}
\label{sect:impl}

%For our experiments, we have implemented an SMT solver,
%\ostrich, %\footnote{Name anonymized for doubly-blind review, and will  be provided in the final version.} 
%based on the calculus for
%$\strline$.  \ostrich\ extends the open-source solver
%OSTRICH~\cite{CHL+19}, and is able to decide satisfiability of
%$\strlinesl$ formulas. \ostrich\ also inherits support for most of the
%other operations of the SMT-LIB theory of Unicode
%strings\footnote{\url{http://smtlib.cs.uiowa.edu/theories-UnicodeStrings.shtml}}
%from OSTRICH.

We extend the open-source solver
OSTRICH~\cite{CHL+19} to support for $\strline$ %implemented %an SMT solver,
%\ostrich,  
based on the calculus. In particular, it can decide the satisfiability of
$\strlinesl$ formulas. 
The extension can handle most of the
other operations of the SMT-LIB theory of Unicode
strings.\footnote{\url{http://smtlib.cs.uiowa.edu/theories-UnicodeStrings.shtml}}
%from OSTRICH.

%\cite{CHL+19}, %\url{https://github.com/uuverifiers/ostrich}}
%which provides a modular and easy-to-use framework for extending all
%sorts of string operations. 
%As shown in Section \ref{sec:decision},
%PSSTs satisfy the conditions required by the backward reasoning
%approach of OSTRICH, which enables us to integrate our logic with
%standard string theory. 

\subsection{Implementation}

Our solver extends classical regular expressions in SMT-LIB
with indexed {\sf re.capture} and {\sf re.reference} operators, which
denote capturing groups and references to them. We also add {\sf re.*?}, {\sf re.+?}, {\sf re.opt?} and {\sf re.loop?} as the lazy counterparts of
Kleene star, plus operator, optional operator and loop operator.

Three new string operators are introduced to make use of these extended regular
expressions: {\sf str.replace\_cg}, {\sf str.replace\_cg\_all}, and
{\sf str.extract}. The operators {\sf str.replace\_cg} and {\sf
  str.replace\_cg\_all} are the counterparts of the standard {\sf
  str.replace\_re} and {\sf replace\_re\_all} operators, and allow
capturing groups in the match pattern and references in the
replacement pattern. E.g., the following constraint swaps the
first name and the last name, as in Example~\ref{exmp-name-swap}:
%lower-case characters~$x$ with a sub-sequent upper-case character~$Y$:

{\small
\begin{minted}{lisp}
(= w (str.replace_cg_all v
      (re.++ ((_ re.capture 1)
                 (re.+ (re.union (re.range "A" "Z") (re.range "a" "z"))))
             (str.to.re " ")
             ((_ re.capture 2)
                 (re.+ (re.union (re.range "A" "Z") (re.range "a" "z")))))
      (re.++ (_ re.reference 2) (_ re.reference 1))))
\end{minted}
}

The replacement string is written as a regular expression only
containing the operators {\sf re.++}, {\sf str.to\_re}, and {\sf
  re.reference}. The use of string variables in the replacement
parameter is not allowed, since the resulting transformation could
not be mapped to a \PSST.

The indexed operator {\sf str.extract} implements $\extract_{i, e}$ in
$\strline$. For instance,

{\small
\begin{minted}{lisp}
  ((_ str.extract 1)
      (re.++ (re.*? re.allchar)((_ re.capture 1) (re.+ (re.range "a" "z")) re.all)) 
      x)
\end{minted}
}

\noindent
extracts the left-most, longest sub-string of lower-case
characters from a string~$x$.

Our implementation is able to handle \textit{anchors} as well,
although for reasons of presentation we did not introduce them as part
of our formalism. Anchors match certain
positions of a string without consuming any input characters. In most
programming languages, it is common to use \verb!^!
and
\verb!$! in regular expressions to signify the start and end of a
string, respectively. We add \textsf{re.begin-anchor} and
\textsf{re.end-anchor} for them. Our implementation correctly models
the semantics of anchors and is able to solve constraints containing
these operators.

\ostrich\ implements the procedure in Theorem~\ref{thm-main}, and
focuses on SL formulas. The three string operators mentioned above
will be converted into an equivalent PSST (see \ifproceeding the full
version of the paper \cite{popl22-full}).  \else
Appendix~\ref{appendix:sec-extract-replace-to-psst}). \fi {\ostrich}
then iteratively applies the propagation rules from
Section~\ref{sect:calculus} to derive further \regexp\ constraints,
and eventually either detect a contradiction, or converge and find a
fixed-point. For straight-line formulas, the existence of a
fixed-point implies satisfiability, and a solution can then be
constructed as described in Section~\ref{sect:calculus}.  In addition,
similar to other SMT solvers, {\ostrich} applies simplification rules
(e.g., {\sf Fwd-Prop-Elim, =-Prop, Subsume, Close}, etc in
Table~\ref{tab:calculus}) to formulas before invoking the SL
procedure.  This enables \ostrich\ to solve some formulas outside of the SL
fragment, but is not a complete procedure for non-SL formulas.

\begin{figure}[tb]
  \scriptsize

  \begin{minipage}{0.55\linewidth}
\begin{minted}{lisp}
(declare-fun x () String)
(define-fun  y () String (str.replace_cg_all x <re1> <repl>))
(push 1)
(assert (str.in.re x (re.++ re.all <re1> re.all)))
(assert (str.in.re y (re.++ re.all <re2> re.all)))
(check-sat) (get-model)
(pop 1) (push 1)
(assert (str.in.re x (re.++ re.all <re1> re.all)))
(assert (not (str.in.re y 
          (re.++ re.all <re2> re.all))))
(check-sat) (get-model)
(pop 1) (push 1)
(assert (not (str.in.re x (re.++ re.all <re1> re.all))))
(check-sat) (get-model)
(pop 1)
\end{minted}
  \end{minipage}\hfill
  \raisebox{-19.2ex}{\rule{0.4pt}{40ex}}\hfill
  \begin{minipage}{0.43\linewidth}
\begin{minted}{js}
function fun(x) {
   if(/<re1>/.test(x)) {
      var y = x.replace(/<re1>/g, <repl>);
      if(/<re2>/.test(y))
        console.log("1");
      else
        console.log("2");
   }
   else
       console.log("3");
}

var S$ = require("S$");
var x = S$.symbol("x", "");
fun(x);
\end{minted}
  \end{minipage}
  
  \caption{Harnesses with replace-all: SMT-LIB for \ostrich\ (left),
    and JavaScript for \expose{} (right).}
  \label{fig:harness}
\end{figure}

%%%%%%%%%%%%%%%%%%%%%%%%%%%%%%%%%%%%%%%%%%%%%%%%%%%%%%%%%%%%%%%%%%%%%%%%%%%%%%%%%

\subsection{Experimental evaluation}

Our experiments have the purpose of answering the following main questions:

%\smallskip
%%%%%%%%%%%%%%%%%%%%%%%%%%%%
\OMIT
{
\textbf{R1:} Are the  {\regexp}s defined in this paper
suitable to encode regular expressions in programming languages,
for instance ECMAScript regular expressions~\cite{ECMAScript11}?\zhilin{Should R1 be removed, since we already have semantics validation experiments ?}\philipp{yes, I think we should remove it; everybody agrees?}
}
%%%%%%%%%%%%%%%%%%%%%%%%%%%%
\noindent
\textbf{R1:} How does \ostrich\ compare to other solvers that can
handle real-world regular expressions, including greedy/lazy
quantifiers and capturing groups?
\\
\textbf{R2:} How does \ostrich\ perform in the context of symbolic execution,
the primary application of string constraint solving?

%
%%%%%%%%%%%%%%%%%%%%%%%%%%%%
\OMIT
{
\paragraph{For \textbf{R1}:} We implemented a translator from ECMAScript 11th Edition (ES11 for short) regular
expressions to \regexps, and integrated \ostrich\ into the symbolic
execution tool Aratha~\cite{aratha}. We then ran Aratha+\ostrich\ on
the regression test suite of \expose{}~\cite{DBLP:conf/spin/LoringMK17},
as well as some other JavaScript programs containing match or replace
functions extracted from Github. To verify the soundness of
Aratha+\ostrich, we compared the results with those produced by
\expose{}; we also checked the correctness of models computed by
\ostrich\ by concretely executing the JavaScript program under test on
the generated inputs, to confirm that the concrete execution indeed
follows the targeted path. The results are summarized in Table~\ref{tab:exp-r1};
no inconsistencies were observed in the experiments, showing that the
semantics in this paper are indeed suitable for capturing ES11
semantics.
}
%%%%%%%%%%%%%%%%%%%%%%%%%%%%
\smallskip
{\em For \textbf{R1}:} There are no standard string benchmarks
involving \regexps, and we are not aware of other constraint solvers
supporting capturing groups, neither among the SMT nor the CP
solvers. %(e.g.,
%\cite{DBLP:conf/cpaior/ScottFPS17,DBLP:conf/cp/AmadiniGS20}).
The closest related work is the algorithm implemented in \expose{}, which
applies Z3~\cite{Z3} for solving string constraints, but augments
it with a refinement loop to approximate the {\regexp}
semantics.\footnote{We considered replacing Z3 with \ostrich\ in
  \expose{} for the experiments. However, \expose{} integrates Z3 using its
  C API, and changing to \ostrich, with native support
  for capturing groups, would have required the rewrite of substantial
  parts of \expose{}.}
For \textbf{R1}, we compared \ostrich\ with \expose{}+Z3 on 98,117
\regexps\ taken from \cite{DMC+19}.

For each regular expression, we created four harnesses: two in
SMT-LIB, as inputs for \ostrich, and two in JavaScript, as inputs for
\expose{}+Z3. The two harnesses shown in Fig.~\ref{fig:harness} use one of the
regular expressions from \cite{DMC+19} (\verb!<re1>!) in combination with
the replace-all function to simulate typical string processing;
\verb!<re2>! is the fixed pattern \verb![a-z]+!, and \verb!<repl>! the
replacement string \verb!"$1"!. The three paths of the JavaScript
function~\verb!fun! correspond to the three queries in the SMT-LIB
script, so that a direct comparison can be made between the results of
the SMT-LIB queries and the set of paths covered by \expose{}+Z3. The other
two harnesses are similar to the ones in Fig.~\ref{fig:harness}, but
use the match function instead of replace-all, and contain four
queries and four paths, respectively.

\begin{table}[t]
  \small
  \begin{center}
  \begin{tabular}{|l@{~~}|*{6}{c}|*{5}{c}@{~~}|}
    \hline
     & 
    \multicolumn{6}{c|}{\textbf{\ostrich}} &
    \multicolumn{5}{c|}{\textbf{\expose{}+Z3}}
    \\
      & \multicolumn{6}{c|}{\# queries solved within 60s}
      & \multicolumn{5}{c|}{\# paths covered within 60s}
    \\
     & 0 & 1 & 2 & 3 & 4 & \#Err
     & 0 & 1 & 2 & 3 & 4
    \\\hline
    \textbf{Match}  & 422 & 249 & 751 & 386 & 95,175 & 1,134
    & ~3,333 & 9,274 & 36,916 & 48,594 & 0
    \\
     \emph{(98,117} & \multicolumn{6}{c|}{Average time: 1.57s}
    &\multicolumn{5}{c|}{Average time: 28.0s}
    \\
    \emph{~~benchm.)} & \multicolumn{6}{c|}{Total \#sat: 250,947, \#unsat: 132,662}
    & \multicolumn{5}{c|}{Total \#paths covered: 228,888}
    \\\hline
    \textbf{Replace} & 4,170 & 2,463 & 555 & 89,794 & --- & 1,135
    & ~5,281 & 18,221 & 69,059 & 5,556 & ---
    \\
    \emph{(98,117} & \multicolumn{6}{c|}{Average time: 6.62s}
    & \multicolumn{5}{c|}{Average time: 55.0s}
    \\
    \emph{~~bench.)} & \multicolumn{6}{c|}{Total \#sat: 259,354, \#unsat: 13,601}
    & \multicolumn{5}{c|}{Total \#paths covered: 173,007}
      \\\hline
  \end{tabular}
  \end{center}
  \caption{The number of queries answered by \ostrich, and number of
    paths covered by \expose{}+Z3, in \textbf{R1}. 
    Experiments were done on an AMD Opteron 2220 SE machine, running
    64-bit Linux and Java~1.8.  Runtime per benchmark was limited to
    60s wall-clock time, 2GB memory, and the number of tests
    executed concurrently by \expose{}+Z3 to 1.  Average time is
    wall-clock time per benchmark, timeouts count as 60s.}
  \label{tab:exp-r2}

	\begin{center}
	\begin{tabular}{|l@{\quad}|*{6}{c}|*{6}{c}|}
	\hline
	   & 
	  \multicolumn{6}{c|}{\textbf{Aratha+\ostrich}} &\multicolumn{6}{c|}{\textbf{\expose{}+Z3}}
	  \\
    & 
	  \multicolumn{6}{c|}{~\# paths covered within 120s~~} &\multicolumn{6}{c@{\quad}|}{~\# paths covered within 120s}
	  \\
	   & ~~0~  & ~1~ &  ~2~ & ~3~ & ~$\geq$4~ & ~\#Err~ &
	    ~~0~  & ~1~ &  ~2~ & ~3~ & ~$\geq$4~ & ~\#T.O.~
	  \\\hline
	  \textbf{\expose{}} &  14 & 9 & 9 & 2 & 15 & 2  & 14 &  9 & 9 & 2 & 15 & 6  
	  \\
	  \emph{(49 programs)} & \multicolumn{6}{c|}{Average time: \textbf{4.66}s} & \multicolumn{6}{c@{\quad}|}{Average time: 57.44s}\\
	  & \multicolumn{6}{c|}{Total \#paths covered:124}  & \multicolumn{6}{c@{\quad}|}{Total \#paths covered:121}	  
	  \\\hline
	  \textbf{Match} & 3 & 7 & 12 & 6 & 0 & 0  & 3 & 8 & 12 & 5 & 0 & 6
	  \\
	  \emph{(28 programs)} & \multicolumn{6}{c|}{Average time: \textbf{5.19}s} & \multicolumn{6}{c@{\quad}|}{Average time: 60.26s}\\
	  & \multicolumn{6}{c|}{Total \#paths covered: 49}  & \multicolumn{6}{c@{\quad}|}{Total \#paths covered: 47}	  
	  \\\hline
	  \textbf{Replace} & 12 & 20 & 6 & 0 & 0 & 0  & 15 & 21 & 2 & 0 & 0 & 23
	  \\
	  \emph{(38 programs)} & \multicolumn{6}{c|}{Average time: \textbf{4.14}s} & \multicolumn{6}{c@{\quad}|}{Average time: 95.34s}\\
	  & \multicolumn{6}{c|}{Total \#paths covered: 32}  & \multicolumn{6}{c@{\quad}|}{Total \#paths covered: 25}	  
	  \\\hline
	\end{tabular}
	\end{center}
	\caption{Results of Expose+Z3 and Aratha+{\ostrich} on Javascript programs for \textbf{R2}. Experiments were done on an Intel(R)-Core(TM)-i5-8265U-CPU-@1.60GHz cpu, running 64-bit Linux and Java 1.8. Runtime was limited to 120s wall-clock time. Average time is
    wall-clock time needed per benchmark, and counts timeouts as 120s. \#Err is the number of non-straight-line path constraints that OSTRICH fails to deal with and \#T.O is the number of timeouts. Note that some paths may have already been covered before T.O. }
	\label{tab:exp-r1}
\end{table}

The results of this experiment are shown in
Table~\ref{tab:exp-r2}. \ostrich\ is able to answer all four queries
in 95,175 of the match benchmarks (97\%), and all three queries in
89,794 of the replace-all benchmarks (91.5\%). The errors in 1,134
cases (resp., 1,135 cases) are mainly due to back-references in
\verb!<re1>!, which are not handled by \ostrich. \expose{}+Z3 can
cover 228,888 paths of the match problems in total (91.2\% of the
number of sat results of \ostrich), although the runtime of
\expose{}+Z3 is on average 18x higher than that of \ostrich. For
replace, \expose{}+Z3 can cover 173,007 paths (66.7\%), showing that
this class of constraints is harder; the runtime of \expose{}+Z3 is on
average 8x higher than that of \ostrich.  Overall, even taking into
account that \expose{}+Z3 has to analyze JavaScript code, as opposed
to the SMT-LIB given to \ostrich, the experiments show that \ostrich\
is a highly competitive solver for \regexps.

\smallskip
{\em For \textbf{R2}:} For this experiment,
we  integrated \ostrich\ into the symbolic
execution tool Aratha~\cite{aratha}.  We compare Aratha+{\ostrich} with
\expose{}+Z3 on the regression test suite of \expose{}~\cite{DBLP:conf/spin/LoringMK17},
as well as a collection of other JavaScript programs containing match or replace
functions extracted from Github. In
Table~\ref{tab:exp-r1}, we can see that Aratha+{\ostrich} can within
120s cover slightly more paths than \expose{}+Z3. Aratha+{\ostrich} can
discover feasible paths much more quickly than \expose{}+Z3, however: on
all three families of benchmarks, Aratha+{\ostrich} terminates on
average in less than 10s, and it discovers all paths 
within 20s. \expose{}+Z3 needs the full 120s for
35 of the programs (``T.O.'' in the table),
and it finds new paths until the end of the 120s. Since \expose{}+Z3
handles the replace-all operation by unrolling, it is not able to
prove infeasibility of paths involving such operations, and will
therefore not terminate on some programs.
%is faster than \expose{}+Z3 by a factor between 1.7 and 15, while being
%able to cover more paths. The smaller speedup compared to the results
%for \textbf{R2} can be explained by the fact that Aratha should spend time in processing Javascript code, besides using {\ostrich} to solve path constraints.
%SMT-LIB queries
%produced by Aratha are not pure string constraints, but also include
%ADTs, arrays, and bit-vectors. 
Overall, the experiments indicate that {\ostrich} is more
efficient than the CEGAR-augmented symbolic execution for dealing
with \regexps.

%%%%%%%%%%%%%%%%%%%%%%%%%%%%

\section{Related Work}
\label{sec-related}

%In this section, we will discuss related results. In particular, we will discuss
%(1) results on modelling and reasoning about
%\regexp{} constraints, and (2) results on string constraint solving.

%\subsection
\noindent{\bf Modelling and Reasoning about \regexp{}.}
%This paper is concerned with string constraint solving in general, but the focus is on the interplay of regular expressions in modern programming language and solving constraints involving complex string functions. Both of them are monumental research fields for which we will only discuss the work which are most pertinent to ours. 
%
Variants and extensions of regular expressions to capture their usage in programming languages have received attention %been investigated 
in both theory and practice. In formal language theory, regular expressions with capturing groups and backreferences were considered in \cite{CSY03,CN09} and also more recently in \cite{Freydenberger13,Schmid16,BM17b,FS19}, where the expressibility issues and decision problems were investigated. Nevertheless, some basic features of these regular expression, namely, the non-commutative union and the greedy/lazy semantics of Kleene star/plus, were not addressed therein. In the software engineering community, % have also received attention in the software engineering community. 
some empirical studies were recently reported for these regular expressions, including portability across different programing languages \cite{DMC+19} and DDos attacks \cite{SP18}, as well as how programmers write them in practice \cite{MDD+19}.

Prioritized finite-state automata and %prioritized finite-state 
transducers were proposed in \cite{BM17}. Prioritized finite-state transducers add indexed brackets to the input string in order to identify the matches of capturing groups. It is hard---if not impossible---to use prioritized finite-state transducers to model replace(all) function, e.g., swapping the first and last name as in Example~\ref{exmp-name-swap}. In contrast, {\PSST}s store the matches in string variables, which can then be referred to, allowing us to conveniently model the match and replace(all) function. 
Streaming string transducers were used in \cite{ZAM19} to solve the straight-line string constraints with concatenation, finite-state transducers, and regular constraints.

%\subsection
\smallskip
\noindent {\bf String Constraint Solving.}
As we discussed Section \ref{sec-intro}, there has been much research
focussing on string constraint solving algorithms, especially
in the past ten years. Solvers typically use a combination
of techniques to check the satisfiability of string constraints,
including word-based methods, automata-based methods, and unfolding-based methods
like the translation to bit-vector constraints.
We mention among others the following string solvers:
Z3 \cite{Z3}, CVC4 \cite{cvc4}, Z3-str/2/3/4 \cite{Z3-str,Z3-str2,Z3-str3,BerzishMurphy2021},
 ABC \cite{ABC}, Norn
\cite{Abdulla14}, Trau \cite{Z3-trau,AbdullaACDHRR18-trau,Abdulla17}, OSTRICH
\cite{CHL+19}, S2S \cite{DBLP:conf/aplas/LeH18}, Qzy \cite{cox2017model}, Stranger \cite{Stranger}, Sloth
\cite{HJLRV18,AbdullaA+19},
Slog \cite{fang-yu-circuits}, Slent \cite{WC+18}, Gecode+S \cite{DBLP:conf/cpaior/ScottFPS17}, G-Strings \cite{DBLP:conf/cp/AmadiniGST17}, HAMPI
\cite{HAMPI}, and S3 \cite{S3}. 
Most modern string solvers provide support of concatenation and regular 
constraints. The push (e.g. see
\cite{GB16,Vijay-length,HAMPI,Berkeley-JavaScript,LB16,S3})
towards incorporating other functions---e.g. length, 
string-number conversions, replace, replaceAll---in a string theory is an
important theme in the area, owing to the desire to be able to reason 
about complex real-world string-manipulating programs.
These functions, among others, are now part of the SMT-LIB Unicode Strings
standard.\footnote{See
\url{http://smtlib.cs.uiowa.edu/theories-UnicodeStrings.shtml}}

To the best of our knowledge, there is currently no solver that
supports \regexp\ features like greedy/lazing matching or capturing
groups (apart from our own solver \ostrich). This was remarked in
\cite{LMK19}, where the authors try to amend the situation by developing 
\expose{} --- a dynamic symbolic execution engine --- that maps path 
constraints in JavaScript to Z3. The strength of \expose{} is in a thorough
modelling of \regexp{} features, some of which (including backreferences) we do 
not cover in our string constraint language and string solver \ostrich{}. However,
the features that we do not cover are also rare in practice, according to
\cite{LMK19} --- in fact, around 75\% of all the \regexp{} expressions found in
their benchmarks across 415,487 NPM packages can be covered in our fragment.
The strength of \ostrich{} against \expose{} is in a substantial improvement in
performance (by 30--50 fold) and precision. \expose{} does not terminate 
even for simple examples (e.g. for Example \ref{exmp-name-swap} and Example 
\ref{ex:normalize}), which can be solved by our solver within a few seconds.
%we do cover a huge portion of \regexp{} that arise in practice.

%For a recent comparison of the solvers, we refer the
%reader to the survey \cite{Ama20}.

For string constraint solving in general, we refer the readers to the recent survey \cite{Ama20}. In this work, we consider a string constraint language which is undecidable in general, and propose a propagation-based calculus to solve the constraints. However, we also identified a straight-line fragment including concatenation, extract, replace(All) which turns to be decidable. Our decision procedure extends the backward-reasoning approach in \cite{CHL+19}, where only standard one-way and two-way finite-state transducers were considered.

%%%%%%%%%%%%%%%%%%%%%%%%%%%%

%!TEX root = main.tex

\section{Conclusion}\label{sec-conc}
The challenge of reasoning about string constraints with regular expressions
stems from functions like match and replace that exploit features like capturing groups, not to mention the subtle deterministic
(greedy/lazy) matching. Our results provide the first string 
solving method that natively supports and effectively handles \regexp{}, which 
is a large order of magnitude faster than the symbolic execution engine 
\expose{} \cite{LMK19} tailored to constraints with regular expressions, 
which is at the moment the only available method for reasoning about string 
constraints with regular expressions. Our solver \ostrich{} relies on two ingredients: 
(i) Prioritized Streaming String Transducers (used to capture subtle non-standard
semantics of \regexp{}, while being amenable to analysis), and
(ii) a sequent calculus that exploits nice closure and algorithmic properties of
\PSST, and performs a kind of propagation of regular constraints by means of 
taking post-images or pre-images. We have also carried out thorough empirical studies
to validate our formalization of \regexp{} as {\PSST} with respect to JavaScript
semantics, as well as to measure the performance of our solver.
Finally, although the satisfiability of the constraint language is %generally 
undecidable, we have also shown that our solver terminates (and therefore is
complete) for the %so-called 
straight-line fragment.

Several avenues for future work are obvious. Firstly, it would be interesting to
see how \expose{} could be used in combination with our solver \ostrich{}. This
would essentially lift \ostrich{} to a symbolic execution engine (i.e. working
at the level of programs).

Secondly, we could incorporate other features of \regexp{} that
are not in our framework, e.g., lookahead and backreferences.
%\mat{change}
To handle lookahead, we may consider alternating variants of PSSTs.
Alternating automata~\cite{CKS81} are effectively able to branch and run parallel checks on the input.
We will need to model the subtle interplay between lookahead and references.
Backreferences could be handled by allowing some inspection of variable contents during transducer runs.
There is some precedent for this in higher-order automata~\cite{M76,E91}, whose stacks non-trivially store and use data.
However, the pre-image of string functions supporting \regexp{} with backreferences will not be regular in general, and emptiness of intersection of \regexp{} with backreferences is undecidable~\cite{CN09}.
Decidability can be recovered in some cases~\cite{FS19}.
We may study these cases or look for incomplete algorithms.
%\mat{endchange}

%\color{red}
%% a discussion of lookahead etc. 
%\begin{remark}
%	Notice that $\extract$ and $\replace$ use regexes with non-classical semantics whereas %. It is previously not known wPre-image recognizability was not known; 
%	previous results considered classical regular expressions without references (a non-example is Example \ref{exmp-name-swap}). We exploit PSSTs to capture the new semantics and prove pre-image recognizability. Handling other features such as lookahead/backreferences is non-trivial, partly due to the undecidability of the resulting SL fragment (e.g., with backreferences). 
%	One may investigate alternating PSSTs or product constructions for lookahead, and extend PSSTs to inspect variable contents for backreferences, which are left as future work.
%\end{remark}
%\color{black}

Finally, since strings do not live in isolation in
a real-world program, there is a real need to also extend our work with other
data types, in particular the integer data type.

%For the future work, it is interesting to extend this work to deal with more advanced features of regular expressions, e.g., lookahead and lookbehind. It is also desirable to support additional string functions involving the integer data type. %, in addition to those involving RWREs.

%%%%%%%%%%%%%%%%%%%%%%%%%%%%

\begin{acks}
 %% acks environment is optional
 %% contents suppressed with 'anonymous'
 %% Commands \grantsponsor{<sponsorID>}{<name>}{<url>} and
 %% \grantnum[<url>]{<sponsorID>}{<number>} should be used to
 %% acknowledge financial support and will be used by metadata
 %% extraction tools.
 % This material is based upon work supported by the
 % \grantsponsor{GS100000001}{National Science
 %   Foundation}{http://dx.doi.org/10.13039/100000001} under Grant
 % No.~\grantnum{GS100000001}{nnnnnnn} and Grant
 % No.~\grantnum{GS100000001}{mmmmmmm}.  Any opinions, findings, and
 % conclusions or recommendations expressed in this material are those
 % of the author and do not necessarily reflect the views of the
 % National Science Foundation.
We thank Johannes Kinder and anonymous referees for their helpful feedback.
T. Chen is supported by \grantsponsor{}{Birkbeck BEI School}{} under Grant No.~\grantnum{}{ARTEFACT},  \grantsponsor{}{National Natural Science Foundation of China}{} under Grant No.~\grantnum{}{62072309}, \grantsponsor{}{The State Key Laboratory of Novel Software Technology, Nanjing University}{} under Grant No.~\grantnum{}{KFKT2018A16}.
    M. Hague and A. Flores-Lamas are supported by the
    \grantsponsor{GS501100000266}
                 {Engineering and Physical Sciences Research Council}
                 {http://dx.doi.org/10.13039/501100000266}
    under Grant No.~\grantnum{GS501100000266}{EP/T00021X/1}
    S.~Kan and A.~Lin is supported by the European Research Council (ERC) 
    under the European
    Union's Horizon 2020 research and innovation programme (grant agreement no
    759969).
    P.\ R\"ummer is supported by the Swedish Research Council (VR)
    under grant~2018-04727, by the Swedish Foundation for Strategic
    Research (SSF) under the project WebSec (Ref.\ RIT17-0011), by the
    Wallenberg project UPDATE, and by grants from Microsoft and Amazon
    Web Services.
    Z. Wu is supported by the
    \grantsponsor{}{National Natural Science Foundation of China}{}
    under Grant No.~\grantnum{}{61872340}.
\end{acks}

\newpage

\bibliography{main}

\newpage
%%% Appendix
\appendix

\section{Appendix}

\subsection{Construction of {\PSST} from {\regexp}} \label{app:reg2psst}

\paragraph{Case $e =\emptyset$ (see Figure~\ref{fig-reg2pfa-0})} $\cT_\emptyset = (\{q_{\emptyset, 0}\}, \Sigma, \{x_{\emptyset}\}, \delta_\emptyset, \tau_\emptyset, E_\emptyset, q_{\emptyset, 0}, (\emptyset, \emptyset))$, where there are no transitions out of $q_{\emptyset,0}$, namely, $\delta_\emptyset(q_{\emptyset, 0}, a) = ()$ for every $a \in \Sigma$, $\tau_\emptyset(q_{\emptyset, 0}) = ((); ())$, and $E_\emptyset$ is vacuous here.
%$\delta_\emptyset(q_{\emptyset, 0}, a) = ()$ for every $a \in \Sigma$, $\tau_\emptyset(q_{\emptyset, 0}) = ((); ())$, and $E_\emptyset$ is vacuous here. 

\paragraph{Case $e = \varepsilon$ (see Figure~\ref{fig-reg2pfa-0})} $\cT_\varepsilon = (\{q_{\varepsilon, 0}, f_{\varepsilon,0}\}, \Sigma, \{x_\varepsilon\}, \delta_\varepsilon, \tau_\varepsilon, E_\varepsilon, q_{\varepsilon,0}, (\{f_{\varepsilon,0}\}, \emptyset))$, 
where 
%$\delta_\varepsilon(q_{\varepsilon,0}, a) = \delta_\varepsilon(f_{\varepsilon,0}, a) = ()$ for every $a \in \Sigma$, 
$\tau_\varepsilon(q_{\varepsilon,0}) = ((f_{\varepsilon,0}); ())$, %for each transition $(q, a, q')$, 
and $E_\varepsilon(q_{\varepsilon,0}, \varepsilon, f_{\varepsilon,0})(x) = \varepsilon$. Note $F_2 = \emptyset$ here. 

\paragraph{Case $e = a$ (see Figure~\ref{fig-reg2pfa-0})} $\cT_a = (\{q_{a,0}, q_{a,1}, f_{a,0}\}, \Sigma, \{x_a\}, \delta_a, \tau_a, E_a, q_{a,0}, (\emptyset, \{f_{a,0}\}))$, where 
%$\delta_a(q_{a,0}, b) = ()$ for every $b \in \Sigma$, 
$\tau_a(q_{a,0}) = ((q_{a,1}); ())$, 
$\delta_a(q_{a,1}, a) = (f_{a,0})$, 
%$\delta_a(q_{a,1}, b) = ()$ for every $b \in \Sigma \setminus \{a\}$,
%
%$\tau_a(q_{a,1}) = ((); ())$, and $\tau_a(f_{a,0}) = ((); ())$, 
%
$E_a(q_{a,0}, \varepsilon, q_{a,1})(x_a) = \varepsilon$, and $E_a(q_{a,1}, a, f_{a,0})(x_a) =x_aa$. Note $F_1 = \emptyset$ here. 
\begin{figure}[ht]
	\centering
	%\rule{\linewidth}{0cm}
	\includegraphics[width = 0.4\textwidth]{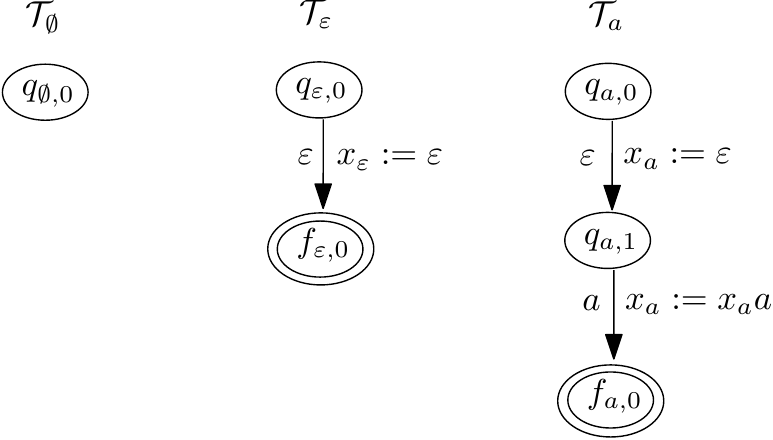}
	\caption{The PSST $\cT_{\emptyset}$, $\cT_{\varepsilon}$, and $\cT_{a}$ }
	\label{fig-reg2pfa-0}
\end{figure}

\paragraph{Case $e = [e_1^{+?}]$} Then $\cT_e$ is constructed from $\cT_{e_1}$ and $\cT^-_{[e^{\ast?}_1]}$, similarly to the aforementioned construction of $\cT_{[e_1^{+}]}$.

%%%%%%%%%%%%%%%%%%%%%%%%%%%%%%%%%%%%%%%%%%%%%%%%%%%%%%%%%%%%%%%%%%%%%%%%%%%%%%%%%%%%%%%%%%%%%%%%%%%%%
\paragraph{Case $e = [e_1^{\{m_1,m_2\}?}]$ for $1 \le m_1 < m_2$ (see Figure~\ref{fig-reg2pfa-5})} Then $\cT_e$ is constructed as the concatenation of $\cT^{\{m_1\}}_{e_1}$ and $\cT^{\{1,m_2-m_1\}?}_{e_1}$, where $\cT^{\{1,m_2-m_1\}?}_{e_1}$ is illustrated in Figure~\ref{fig-reg2pfa-5}, which is the same as $\cT^{\{1,m_2-m_1\}}_{e_1}$ in Figure~\ref{fig-reg2pfa-4}, except that the priorities of the $\varepsilon$-transition from $q^{(1)}_{e_1,0}$ to $f^\prime_0$ has the highest priority and  the priorities of the $\varepsilon$-transitions out of each $f^{(i)}_{e_1,2} \in F^{(i)}_{e_1,2}$ to $f^\prime_1$ are swapped.
\begin{figure}[ht]
	\centering
	%\rule{\linewidth}{0cm}
	\includegraphics[width = 0.8\textwidth]{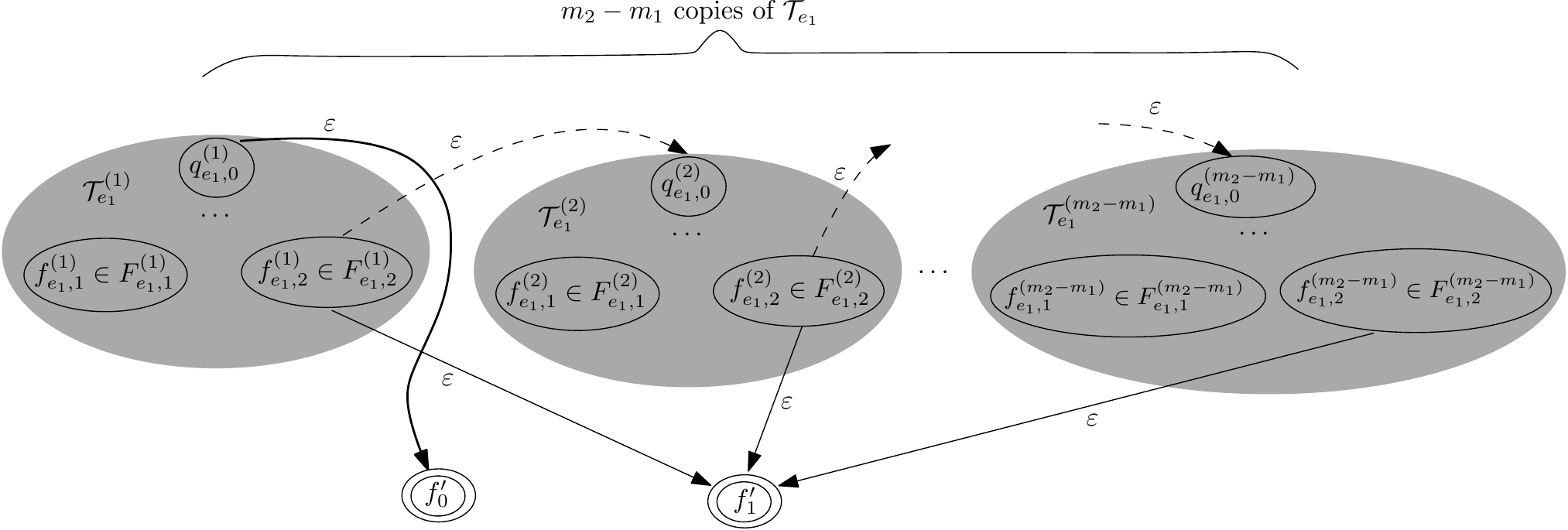}
	\caption{The PSST $\cT^{\{1,m_2-m_1\}?}_{e_1}$}
	\label{fig-reg2pfa-5}
\end{figure}

\subsection{From $\extract$, $\replace$ and $\replaceall$ to PSSTs}\label{appendix:sec-extract-replace-to-psst}

\noindent{\bf Lemma~\ref{lem-str-fun-to-psst}}.
    The satisfiability of $\strline$ reduces to the satisfiability of boolean combinations of formulas of the form $z=x \concat y$, $y=\cT(x)$, and $x \in \cA$, where $\cT$ is a PSST and $\cA$ is an FA.

    The proof is in two steps: first we remove $\$0$, $\refbefore$, and $\refafter$, then we encode the remaining string functions with PSSTs.

    \subsubsection{Removing Special References}

    The first step in our proof is to remove the special references $\$0$, $\refbefore$, and $\refafter$ from the replacement strings.
    These can be replaced in a series of steps, leaving only PSST transductions and replacement strings with only simple references ($\$i$).
    We will just consider $\replaceall$ as $\replace$ is almost identical.

    First, to remove $\$0$, suppose we have a statement
    $y := \replaceall_{\pat, \rep}(x)$
    with $\$0$ in $\rep$.
    We simply substitute
    $y := \replaceall_{\pat', \rep'}(x)$
    where
        $\pat' = (\pat)$, and
        $\rep' = \rep[\$1 / \$0, \$2 / \$1, \ldots, \$(k+1) / \$k]$.
    That is, we make the complete match an explicit (first) capture, which shifts the indexes of the remaining capturing groups by 1.

    Now suppose we have a statement
    $y := \replaceall_{\pat, \rep}(x)$
    with $\refbefore$ or $\refafter$ in $\rep$.
    We replace it with the following statements, explained below, where
    $y_1, \ldots, y_5$
    are fresh variables.
    \[
        \begin{array}{l}
            y_1 := \replaceall_{(\pat), \langle\$1\rangle}(x); \\
            y_2 := \psst_{\langle}(y_1); \\
            y_3 := \psst_{\mathrm{rev}}(y_2); \\
            y_4 := \psst_{\rangle}(y_3); \\
            y_5 := \psst_{\mathrm{rev}}(y_4); \\
            y := \replaceall_{\pat', \rep'}(y_5)
        \end{array}
    \]

    The first step is to mark the matched parts of the string with $\langle$ and $\rangle$ brackets (where $\langle$ and $\rangle$ are not part of the main alphabet).
    This is achieved by the first $\replaceall$.

    Next, we use a PSST $\psst_{\langle}$ that passes over the marked word.
    This is a copyful PSST that simply stores the word read so far into a variable $X$, except for the $\langle$ and $\rangle$ characters.
    It also has an output variable $O$, to which it also copies each character directly, except $\langle$.
    When it encounters $\langle$ it appends to $O$ the string
    $\langle X \langle$.
    That is, it puts the entire string preceding each $\langle$ into the output, surrounded by $\langle$ at the start and end.
    This is copyful since $X$ will be copied to both $X$ and $O$ in this step.
    For example, suppose the input string were
    $a b \langle c \rangle d \langle e \rangle f$,
    the output of $\psst_\langle$ would be
    $a b \langle \underline{a b} \langle c \rangle d \langle \underline{a b c d} \langle e \rangle f$.
    We have underlined the strings inserted for readability.

    The next step is to do the same for $\rangle$.
    To achieve this we first reverse the string so that a PSST can read the end of the string first.
    A similar transduction to $\psst_\langle$ is performed before the string is reversed again.
    In our example the resulting string is
    $a b \langle \underline{a b} \langle c \rangle \underline{d e f} \rangle d \langle \underline{a b c d} \langle e \rangle \underline{f} \rangle f$.

    Finally, we have
    $y := \replaceall_{\pat', \rep'}(y_5)$
    where
    $\pat' =
         \langle (\Sigma^{*?}) \langle
         \pat
         \rangle (\Sigma^{*?}) \rangle$, and
    \[
        \rep' = \rep[
            \$1 / \refbefore,
            \$2 / \$1,
            \ldots
            \$(k+1) / \$k,
            \$(k+2) / \refafter
        ] \ .
    \]
    That is, by inserting the preceding and succeeding text directly next to each match, we can use simple references $\$i$ instead of $\refbefore$ and $\refafter$.

    \subsubsection{Encoding string functions as PSSTs}

    Once $\$0$, $\refbefore$, and $\refafter$ are removed from the replacement strings, then the string functions can be replaced by PSSTs.

    \begin{lemma}
        For each string function $f = \extract_{i,e}$, $\replace_{\pat, \rep}$, or $\replaceall_{\pat, \rep}$ without $\refbefore$ or $\refafter$ in the replacements strings, a PSST $\cT_f$ can be constructed such that
        $$\cR_{f} = \{(w, w') \mid w'= f(w)\}.$$
    \end{lemma}
    \begin{proof}
    The $\extract_{i,e}$ can function be defined by a PSST $\cT_{i,e}$ obtained from the PSST $\cT_e$ (see Section~\ref{sect:regextopsst}) by removing all the string variables, except the string variable $x_{e'}$, where $e'$ is the subexpression of $e$ corresponding to the $i$th capturing group,  and setting the output expression of the final states as $x_{e'}$.

        Next, we give the construction of the PSST for $\replaceall_{\pat, \rep}$ where all the references in $\rep$ are of the form $\$i$.
        
        Recall $\rep = w_1 \$i_1 w_2 \cdots w_k \$i_k w_{k+1}$. Let $e'_{i_1},\ldots, e'_{i_k}$ denote the subexpressions of $\pat$ corresponding to the $i_1$th, $\ldots$, $i_k$th capturing groups of $\pat$.
        Then $\cT_{\replaceall_{\pat, \rep}} = (Q_\pat \cup \{q'_0\}$, $\Sigma$, $X'$, $\delta'$, $\tau', E', q'_0, F')$ where
        \begin{itemize}
            \item $q'_0 \not \in Q_\pat$,

            \item  $X' = \{x_0\} \cup X_\pat$,
            \item $F'(q'_0) = x_0$, and $F'(q')$ is undefined for every $q' \in Q_\pat$,
            \item $\delta'$ comprises the transitions in $\delta_\pat$, and the transition $\delta'(q'_0, a) = (q'_0)$ for $a \in \Sigma$,

            \item  $\tau'$ comprises the transitions in $\tau_\pat$, the transitions $\tau'(q'_0) = ((q_{\pat, 0}); ())$, $\tau'(f_{\pat, 1}) = ((q'_0); ())$ and $\tau'(f_{\pat, 2}) = ((q'_{0}); ())$ for $f_{\pat, 1} \in F_{\pat, 1}$ and $f_{\pat, 2} \in F_{\pat, 2}$,

%%%%%%%%%%%%%%%%%%%%%%%%%%%%%%%%
\OMIT{
                \begin{itemize}
                    \item $\delta'(q'_0, a) = (q'_0)$ for every $a \in \Sigma$, and $\tau'(q'_0) = ((q_0); ())$,
                    \item for every $q \in Q \setminus \{f_0\}$ and $a \in \Sigma$, $\delta'(q, a) = \delta(q, a)$ and $\tau'(q) = \tau(q)$,
                    \item $\delta'(f_0, a) = ()$ for every $a \in \Sigma$ and $\tau'(f_0) = ((q'_0); ())$,
                \end{itemize}
}
%%%%%%%%%%%%%%%%%%%%%%%%%%%%%%%%
                %
            \item $E'$ inherits $E_\pat$, and includes the assignments $E'(q'_0, a, q'_0)(x_0)  = x_0 a$ for $a \in \Sigma$, $E'(f, \varepsilon, q'_0) (x_0) = x_0\rep[x_{e'_{i_1}}/\$i_1, \ldots, x_{e'_{i_k}}/i_k]$ and $E'(f, \varepsilon, q'_0) (x) = \nullchar$ for every  $f \in F_{\pat,1} \cup F_{\pat, 2}$ and $x \in X_\pat$.
%%%%%%%%%%%%%%%%%%%%%%%%%%%%%%%%
\OMIT{
                \begin{itemize}
                    \item for every transition $(q, a, q')$ with $a \in \Sigma^\varepsilon$ in $\cA_\pat$, $E(q, a, q')(x_0) = x_0$,
                    \item for every transition $(q, a, q')$ with $a \in \Sigma^\varepsilon$ and every $j \in [k]$,  if $(q, a, q')$ occurs in ${\sf Sub}_{e'_{i_j}}[\cA_\pat]$, then $E(q, a, q')(x_j) = x_ja$, otherwise, $E(q, a, q')(x_j) = x_j$,
                    \item  for every $a \in \Sigma$ and $j \in [k]$, $E(q'_0, a, q'_0)(x_0) = x_0a$ and $E(q'_0, a, q'_0)$$(x_j) = x_j$,
                    \item $E(q'_0, \varepsilon, q_0)(x_j) = x_j$ for every $j \in [k] \cup \{0\}$,
                    \item $E(f_0, \varepsilon, q'_0)(x_0) = x_0 \rep[x_1/\$i_1,\ldots, x_k/\$i_k]$, moreover, for every $j \in [k]$, $E(f_0, \varepsilon, q'_0)(x_j) = \varepsilon$, where $ \rep[x_1/\$i_1,\ldots, x_k/\$i_k]$ denotes the string term obtained from $\rep$ by replacing every occurrence of $\$i_1,\cdots, \$i_k$ with $x_1,\cdots,x_k$ respectively.
                \end{itemize}
}
%%%%%%%%%%%%%%%%%%%%%%%%%%%%%%%%
                %
        \end{itemize}
        \end{proof}

        The construction of the PSST for $\replace_{\pat, \rep}$ is similar and illustrated in Fig.~\ref{fig-psst-replace}. The details are omitted.
        \begin{figure}[ht]
            \centering
            \includegraphics[width=\textwidth]{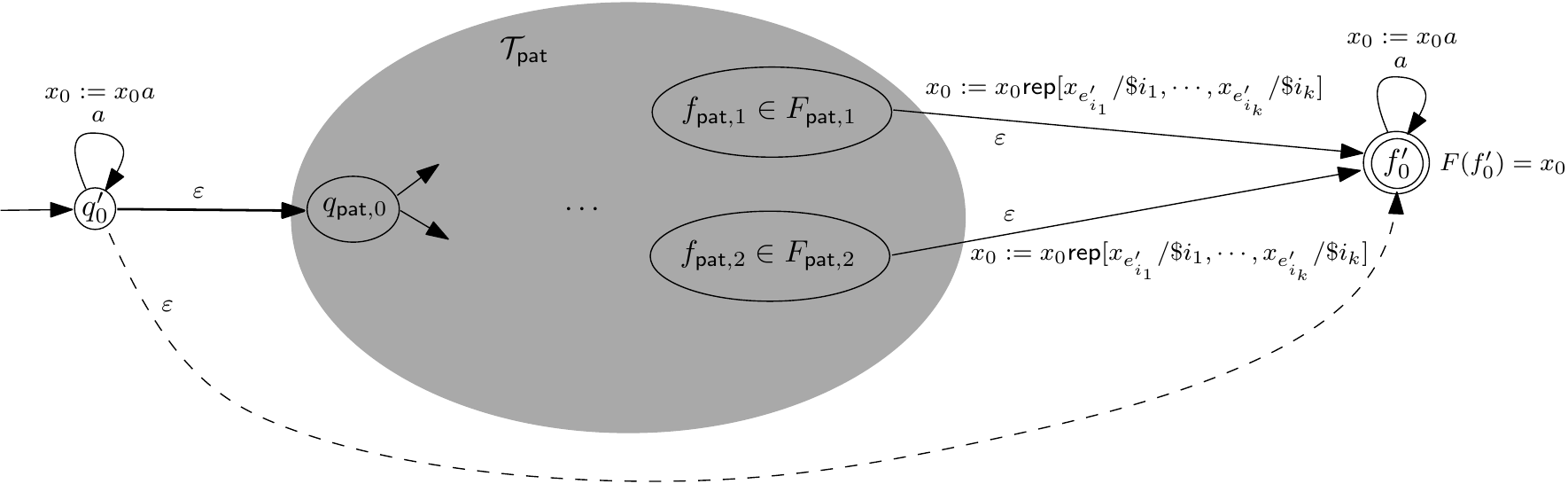}
            \caption{The PSST $\cT_{\replace_{\pat,\rep}}$}
            \label{fig-psst-replace}
        \end{figure}

%%%%%%%%%%%%%%%%%%%%%%%%%%%%%%%%%%%%%%
%%%%%%%%%%%%%%%%%%%%%%%%%%%%%%%%%%%%%%
\subsection{Proof of Lemma~\ref{lem:psst_preimage}}\label{app-pre-image}

%Construction of an FA $\cB$ for $\cR^{-1}_\cT(\Lang(\cA))$

\medskip

\noindent {\bf Lemma~\ref{lem:psst_preimage}}.
\emph{  Given a PSST $\psst = (Q_T, \Sigma$, $X$, $\delta_T$, $\tau_T$, $E_T$,  $q_{0, T}$, $F_T$) and an \FA{} $\Aut
  = (Q_A, \Sigma, \delta_A, q_{0, A}, F_A)$, we can compute an \FA{} $\cB = (Q_B,
  \Sigma, \delta_B, q_{0, B}, F_B)$ in exponential time  such that $\Lang(\cB) = \cR^{-1}_{\cT}(\Lang(\Aut))$.
}

\medskip

We prove Lemma~\ref{lem:psst_preimage} in the sequel.

        Let $\psst = (Q_T, \Sigma$, $X, \delta_T, \tau_T, E_T,  q_{0, T}, F_T)$ be a PSST  and $\Aut
        = (Q_A, \Sigma$, $\delta_A$, $q_{0, A}$, $F_A)$ be an \FA{}. Without loss of generality, we assume that $\Aut$ contains no $\varepsilon$-transitions. For convenience, we use $\cE(\tau_T)$ to denote $\{(q, q') \mid q' \in \tau_T(q)\}$. For convenience, for $a \in \Sigma$, we use $\delta^{(a)}_A$ to denote the  relation $\{(q, q') \mid (q, a, q') \in \delta_A\}$.

        To illustrate the intuition of the proof of Lemma~\ref{lem:psst_preimage}, let us start with the following natural idea of firstly constructing a PFA $\cB$ for the pre-image: $\cB$ simulates a run of $\psst$ on $w$, and, for each $x \in X$, records an $\Aut$-abstraction of the string stored in $x$, that is, the set of state pairs $(p, q) \in Q_A \times Q_A$ such that starting from $p$, $\Aut$ can reach $q$ after reading the string stored in $x$. Specifically, the states of $\cB$ are of the form $(q, \rho)$ with $q \in Q$ and $\rho \in (\cP(Q_A \times Q_A ))^{X}$. Moreover, the priorities of $\cB$ inherit those of $\psst$. The PFA $\cB$ is then transformed to an equivalent FA by simply dropping all priorities. We refer to this FA as $\cB'$.

        Nevertheless, it turns out that this construction is flawed: A string $w$ is in $\cR^{-1}_{\cT}(\Lang(\Aut))$ iff the (unique) accepting run of $\cT$ on $w$ produces an output $w'$ that is accepted by $\Aut$. However, a string $w$ is accepted by $\cB'$ iff \emph{there is a run of $\cT$ on $w$, not necessarily of the highest priority}, producing an output $w'$ that is accepted by $\Aut$.

            While the aforementioned natural idea does not work,  we choose to construct an FA $\cB$ that simulates the \emph{accepting} run of $\psst$ on $w$, and, for each $x \in X$, records an $\Aut$-abstraction of the string stored in $x$, that is, the set of state pairs $(p, q) \in Q_A \times Q_A$ such that starting from $p$, $\Aut$ can reach $q$ after reading the string stored in $x$.
            To simulate the accepting run of $\psst$, it is necessary to record all the states accessible through the runs of higher priorities to ensure the current run is indeed the accepting run of $\psst$ (of highest priority). Moreover, $\cB$ also remembers the set of $\varepsilon$-transitions of $\cT$ after the latest non-$\varepsilon$-transition to ensure that no transition occurs twice in a sequence of $\varepsilon$-transitions of $\cT$.

            Specifically, each state of $\cB$ is of the form $(q, \rho, \Lambda, S)$, where $q \in Q_T$, $\rho \in (\cP(Q_A \times Q_A ))^{X}$, $\Lambda \subseteq \cE(\tau_T)$, and $S \subseteq Q_T$.
            For a state $(q, \rho, \Lambda, S)$, our intention for $S$ is that the states in it are those that can be reached in the runs of higher priorities than the current run, by reading the same sequence of letters and applying the $\varepsilon$-transitions as many as possible. Note that when recording in $S$ all the states accessible through the runs of higher priorities, we do not take the non-repetition of $\varepsilon$-transitions into consideration since if a state is reachable by a sequence of $\varepsilon$-transitions where some $\varepsilon$-transitions are repeated, then there exists also a sequence of non-repeated $\varepsilon$-transitions reaching the state.
            Moreover, when simulating an $a$-transition of $\cT$ (where $a \in \Sigma$) at a state $(q, \rho, \Lambda, S)$, suppose $\delta_T(q, a) = (q_1, \cdots, q_m)$ and $\tau_T(q) = (P_1, P_2)$, then $\cB$ nondeterministically chooses $q_i$ and goes to the state $(q_i, \rho', \emptyset, S')$, where
            \begin{itemize}
                \item $\rho'$ is obtained from $\rho$ and $E_T(q, \sigma, q_i)$,
                \item $\Lambda$ is reset to $\emptyset$,
                \item all the states obtained from $S$ by applying  an $a$ transition should be \emph{saturated by $\varepsilon$-transitions} and put into $S'$, more precisely, all the states reachable from $S$ by first applying an $a$-transition, then a sequence of $\varepsilon$-transitions, should be put into $S'$,
                \item moreover, all the states obtained from $q_1,\cdots, q_{i-1}$ (which are of higher priorities than $q_i$) by saturating with $\varepsilon$-transitions should be put into $S'$,
                \item finally, all the states obtained from those in $P'_1 = \{q' \in P_1 \mid (q, q') \not \in \Lambda\}$ (which are of higher priorities than $q_i$) by saturating with non-$\Lambda$ $\varepsilon$-transitions first (i.e. the $\varepsilon$-transitions that do not belong to $\Lambda$), and applying an $a$-transition next, finally saturating with $\varepsilon$-transitions again, should be put into $S'$, (note that according to the semantics of PSST, the $\varepsilon$-transitions in $\Lambda$ should be avoided when defining $P'_1$ and saturating the states in $P'_1$ with $\varepsilon$-transitions).
            \end{itemize}
            %For technical reasons, when constructing $\cB$, we assume that this saturation happens when a state is added to $S$ for the first time. Therefore, at a state $(q, \rho, \Lambda, S)$, all the states reachable from the states in $S$ by sequences of $\varepsilon$-transitions in $\cT$ have already been in $S$.

            The above construction  does not utilize the so-called \tmtextit{copyless} property (i.e. for each transition $t$ and each variable $x$, $x$ appears at most once on the right-hand side of the assignment for $t$) \cite{AC10,AD11},
            thus it works for general, or \textit{copyful}, PSSTs \cite{FR17}.
            It can be noted that the powerset in
            $\rho \in (\cP(Q_A \times Q_A ))^{X}$
            is required to handle copyful transductions as the contents of a variable may be used in many different situations, each requiring a different abstraction.
            If the PSST is copyless, we can instead use
            $\rho \in (Q_A \times Q_A)^{X}$.
            That is, each variable is used only once, and hence only one abstraction pair is needed.
            The powerset construction in the transitions can be replaced by a non-deterministic choice of the particular pair of states from $Q_A$ that should be kept.
            This avoids the construction being exponential in the size of $A$, which in turn avoids the tower of exponential blow-up in the backwards reasoning.

                    For the formal construction of $\cB$, we need some additional notations.
                    \begin{itemize}
                        \item For $S \subseteq Q_T$, $\delta^{(ip)}_T(S, a) = \{q'_1 \mid \exists q_1 \in S, q'_1 \in \delta_T(q_1, a)\}$.
                        \item For $q \in Q_T$,  if $\tau_T(q) = (P_1, P_2)$, then $\tau^{(ip)}_T(\{q\})=S$ such that $S = P_1 \cup P_2$.
                            Moreover, for $S \subseteq Q_T$, we define $\tau^{(ip)}_T(S) = \bigcup \limits_{q \in S} \tau^{(ip)}_T(\{q\})$. We also use $\big(\tau^{(ip)}_T\big)^\ast$ to denote the $\varepsilon$-closure of $\cT$, namely, $\big(\tau^{(ip)}_T\big)^\ast(S) = \bigcup \limits_{n \in \Nat} \big(\tau^{(ip)}_T\big)^{n}(S)$, where $\big(\tau^{(ip)}_T\big)^{0}(S) = S$, and for $n \in \Nat$, $\big(\tau^{(ip)}_T\big)^{n+1}(S) = \tau^{(ip)}_T\big(\big(\tau^{(ip)}_T\big)^{n}(S)\big)$.
                        \item For $S \subseteq Q_T$ and $\Lambda \subseteq  \cE(\tau_T)$, we use $\big(\tau^{(ip)}_T \backslash \Lambda\big)^\ast(S)$ to denote the set of states reachable from $S$ by sequences of $\varepsilon$-transitions where {\it no} transitions $(q, \varepsilon, q')$ such that $(q, q') \in \Lambda$ are used.
                            %
                            %We also use $(\tau^{(ip)}_T)^\ast$ to denote the reflexive transitive closure of $\tau^{(ip)}_T$. \tl{here $(\tau^{(ip)}_T)$ is defined as a function... you mean function composition?}
                            %\zhilei{I think we can just use the term 'epsilon closure' here?}
                            %\item For $\sigma \in \Sigma$ and $S \subseteq Q_T$,  we use $\tau^+_T[a, S]$ to denote the set of states that can be obtained from
                            %
                        \item For $\rho \in (\cP(Q_A \times Q_A ))^{X}$ and $s \in X \rightarrow (X \cup \Sigma)^{\ast}$, we use $s(\rho)$ to denote $\rho'$ that is obtained from $\rho$ as follows: For each $x \in X$, if $s(x) = \varepsilon$, then $\rho'(x) = \{(p, p) \mid p \in Q_A\}$, otherwise, let $s(x) = b_1 \cdots b_\ell$ with $b_i \in \Sigma \cup X$ for each $i \in [\ell]$, then $\rho'(x) = \theta_1 \circ \cdots \circ \theta_\ell$, where $\theta_i = \delta^{(b_i)}_A$ if $b_i \in \Sigma$, and $\theta_i = \rho(b_i)$ otherwise, and $\circ$ represents the composition of binary relations.
                    \end{itemize}

                    We are ready to present the formal construction of $\cB =  (Q_B$, $\Sigma$, $\delta_B$, $q_{0, B}, F_B)$.
                    \begin{itemize}
                        \item $Q_B = Q_T \times (\cP(Q_A \times Q_A ))^{X} \times \cP(\cE(\tau_T)) \times \cP(Q_T)$,
                            %(Intuitively, the letter $\sigma$ in $(q, \sigma, \rho, S) \in Q_B$ means the next letter to be read at $q$, with $\bot$ represents the end of the input.)

                        \item $q_{0, B} = (q_{0,T}, \rho_{\varepsilon}, \emptyset, \emptyset)$ where $\rho_{\varepsilon} (x) = \{(q, q) \mid q \in Q\}$ for each $x \in X$,

                        \item $\delta_{B}$ comprises
                            \begin{itemize}
                                    %\item the tuples $(q'_0, \varepsilon, ((q_{0,T},\sigma), \rho_{\varepsilon}, \emptyset))$ where $\sigma \in \Sigma$, $\rho_{\varepsilon} (x) = \{(q, q) \mid q \in Q\}$ for each $x \in X$,
                                    %
                                \item the tuples $((q, \rho, \Lambda, S), a, (q_i, \rho', \Lambda', S'))$ such that
                                    %there exists $s \in \left((X \cup \Sigma\right)^*)^X)$ satisfying
                                    \begin{itemize}
                                        \item $a \in \Sigma$,
                                            %$a' \in \Sigma \cup \{\bot\}$,
                                            %
                                        \item $\delta_T (q, a) = (q_1, \ldots, q_i, \ldots, q_m)$,
                                        \item $s = E((q, a, q_i))$,
                                        \item $\rho' = s(\rho)$,
                                        \item $\Lambda' = \emptyset$, (Intuitively, $\Lambda$ is reset.)
                                        \item let $\tau_T(q) = (P_1, P_2)$, then $S' = \big(\tau^{(ip)}_T\big)^\ast\big(\{ q_1$, $\ldots$, $q_{i - 1} \} \cup \delta^{(ip)}_T\big(S \cup \big(\tau^{(ip)}_T \setminus \Lambda\big)^\ast(P'_1), a\big)\big)$, where $P'_1 = \{q' \in P_1 \mid (q, q') \not \in \Lambda\}$;
                                            %(Note that according to the semantics of PSSTs, when computing the set of states reachable from $q$ through an $\varepsilon$-transition to some $q' \in P_1$ first and a sequence of $\varepsilon$-transitions starting from $q'$ next, the transitions $(q'', \varepsilon, q''')$ with $(q'', q''') \in \Lambda$ should be excluded. )
                                            %
                                    \end{itemize}
                                \item the tuples $((q, \rho, \Lambda, S), \varepsilon, (q_i, \rho', \Lambda', S'))$ such that
                                    %there exists $s \in \left((X \cup \Sigma\right)^*)^X$ satisfying
                                    \begin{itemize}
                                            %, $a' \in \Sigma \cup \{\bot\}$,
                                            %
                                        \item $\tau_T(q) = ((q_1, \ldots, q_i, \ldots, q_m); \cdots)$,
                                        \item $(q, q_i) \not \in \Lambda$,

                                        \item $s = E(q, \varepsilon, q_i)$,
                                        \item $\rho' = s(\rho)$,
                                        \item $\Lambda' = \Lambda \cup \{(q, q_i)\}$,
                                        \item $S' =  S \cup \big(\tau^{(ip)}_T \backslash \Lambda \big)^\ast(\{ q_j \mid j \in [i-1], (q, q_j) \not \in \Lambda \})$;
                                            %
                                            %$\rho'(x) = \theta_\ell$ such that $\theta_0 = \{(p,p) \mid p \in Q_A\}$, and for each $i \in [\ell]$, if $b_i \in \Sigma$, then $\theta_i = \{(p, p') \mid (p, p'') \in \theta_{i-1}, (p'', b_i, p') \in \delta_A \mbox{ for some } p''\}$, otherwise, $\theta_i = \theta_{i-1} \cdot \rho(x)$.
                                    \end{itemize}

                                \item the tuples $((q, \rho, \Lambda, S), \varepsilon, (q_i, \rho', \Lambda', S'))$ such that
                                    %there exists $s \in \left((X \cup \Sigma\right)^*)^X$ satisfying
                                    \begin{itemize}
                                            %\item $a \in \Sigma$,
                                            %$a' \in \Sigma \cup \{\bot\}$,
                                            %
                                        \item $\tau_T (q) = ((q'_1, \ldots, q'_n); (q_1, \ldots, q_i, \ldots, q_m))$,
                                        \item $(q, q_i) \not \in \Lambda$,
                                        \item $s = E(q, \varepsilon, q_i)$,
                                        \item $\rho' = s(\rho)$,
                                        \item $\Lambda' = \Lambda \cup \{(q, q_i)\}$,
                                        \item $S' = S \cup \{q\} \cup \big(\tau^{(ip)}_T \backslash \Lambda \big)^\ast\big(\big\{q'_j \mid j \in [n], (q, q'_j) \not \in \Lambda \big\} \cup \big\{q_j \mid j \in [i-1], (q, q_j) \not \in \Lambda \big\} \big)$. (Note that here we include $q$ into $S'$, since the non-$\varepsilon$-transitions out of $q$ have higher priorities than the transition $(q, \varepsilon, q_i)$.)
                                            %
                                            %$\rho'(x) = \theta_\ell$ such that $\theta_0 = \{(p,p) \mid p \in Q_A\}$, and for each $i \in [\ell]$, if $b_i \in \Sigma$, then $\theta_i = \{(p, p') \mid (p, p'') \in \theta_{i-1}, (p'', b_i, p') \in \delta_A \mbox{ for some } p''\}$, otherwise, $\theta_i = \theta_{i-1} \cdot \rho(x)$.
                                    \end{itemize}
                            \end{itemize}
                        \item
                            Moreover, $F_B$ is the set of states $(q, \rho, \Lambda, S) \in Q_B$ such that
                            \begin{enumerate}
                                \item $F_T (q)$ is defined,
                                \item for every $q' \in S$, $F_T (q')$ is not defined,
                                \item if $F_T(q) = \varepsilon$, then $q_{0, A}  \in F_A$, otherwise,
                                    let $F_T(q) = b_1 \cdots b_\ell$ with $b_i \in \Sigma \cup X$ for each $i \in [\ell]$, then $(\theta_1 \circ \cdots \circ \theta_\ell) \cap (\{q_{0,A}\} \times F_A) \neq \emptyset$, where for each $i \in [\ell]$, if $b_i \in \Sigma$, then $\theta_i = \delta^{(b_i)}_A$, otherwise, $\theta_i = \rho(b_i)$.
                            \end{enumerate}
                    \end{itemize}

\newcommand\scon{S}
\newcommand\svar{x}
\newcommand\alphabet{\Sigma}

\newcommand\rebrac[1]{\left[#1\right]}
\newcommand\fbrac[1]{\left\langle#1\right\rangle}

\newcommand\idxi{i}
\newcommand\idxj{j}
\newcommand\idxk{k}
\newcommand\numof{n}

\newcommand\tiles{\Theta}
\newcommand\nontiles{\overline{\Theta}}
\newcommand\hrel{H}
\newcommand\vrel{V}
\newcommand\inittile{\tile^0}
\newcommand\fintile{f}
\newcommand\tile{t}
\newcommand\vartile{d}
\newcommand\linlen{\numof}
\newcommand\tileheight{h}
\newcommand\expheight{m}
\newcommand\spacer{\#}
\newcommand\isnum[2]{R^{#2}_{#1}}
\newcommand\lmark{\langle}
\newcommand\rmark{\rangle}
\newcommand\sftrue[1]{\top_{#1}}
\newcommand\sffalse[1]{\bot_{#1}}
\newcommand\sfvalue{v}

\newcommand\tilesnum[1]{\tiles_{#1}}
\newcommand\hrelnum[1]{\hrel_{#1}}
\newcommand\vrelnum[1]{\vrel_{#1}}
\newcommand\inittilenum[1]{\inittile_{#1}}
\newcommand\fintilenum[1]{\fintile_{#1}}
\newcommand\nmax[1]{N_{#1}}
\newcommand\tenc[2]{[#2]_{#1}}

\newcommand\fullrow[1]{
    \tenc{\expheight}{1} \tile^{#1}_1
        \ldots
        \tenc{\expheight}{\nmax{\expheight}}
            \tile^{#1}_{\nmax{\expheight}}
}

\subsection{Tower-Hardness of String Constraints with Streaming String Transductions}
\label{sec:tower-hard}

We show that the satisfiability problem for $\strlinesl$ is Tower-hard.

\begin{theorem}
The satisfiability problem for $\strlinesl$ is Tower-hard.
\end{theorem}

Our proof will use tiling problems over extremely wide corridors.
We first introduce these tiling problems, then how we will encode potential solutions as words.
Finally, we will show how $\strlinesl$ can verify solutions.

\subsubsection{Tiling Problems}

A *tiling problem* is a tuple
$\tup{\tiles, \hrel, \vrel, \inittile, \fintile}$
where
    $\tiles$ is a finite set of tiles,
    $\hrel \subseteq \tiles \times \tiles$ is a horizontal matching relation,
    $\vrel \subseteq \tiles \times \tiles$ is a vertical matching relation, and
    $\inittile, \fintile \in \tiles$ are initial and final tiles respectively.

A solution to a tiling problem over a $\linlen$-width corridor is a sequence
\[
    \begin{array}{c}
        \tile^1_1 \ldots \tile^1_\linlen \\
        \tile^2_1 \ldots \tile^2_\linlen \\
        \ldots \\
        \tile^\tileheight_1 \ldots \tile^\tileheight_\linlen
    \end{array}
\]
where
$\tile^1_1 = \inittile$,
$\tile^\tileheight_\linlen = \fintile$,
and for all
$1 \leq \idxi < \linlen$
and
$1 \leq \idxj \leq \tileheight$
we have
$\tup{\tile^\idxj_\idxi, \tile^\idxj_{\idxi+1}} \in \hrel$
and for all
$1 \leq \idxi \leq \linlen$
and
$1 \leq \idxj < \tileheight$
we have
$\tup{\tile^\idxj_\idxi, \tile^{\idxj+1}_\idxi} \in \vrel$.
Note, we will assume that $\inittile$ and $\fintile$ can only appear at
the beginning and end of the tiling respectively.

Tiling problems characterise many complexity classes~\cite{BGG97}. In
particular, we will use the following facts.

\begin{itemize}
\item
    For any $\linlen$-space Turing machine, there exists a tiling problem
    of size polynomial in the size of the Turing machine, over a corridor
    of width $\linlen$, that has a solution iff the $\linlen$-space Turing
    machine has a terminating computation.

\item
    There is a fixed
    $\tup{\tiles, \hrel, \vrel, \inittile, \fintile}$
    such that for any width $\linlen$ there is a unique solution
    \[
      \begin{array}{c}
          \tile^1_1 \ldots \tile^1_\linlen \\
          \tile^2_1 \ldots \tile^2_\linlen \\
          \ldots \\
          \tile^\tileheight_1 \ldots \tile^\tileheight_\linlen
      \end{array}
    \]
    and moreover $\tileheight$ is exponential in $\linlen$. One such
    example is a Turing machine where the tape contents represent a binary
    number. The Turing machine starts from a tape containing only $0$s
    and finishes with a tape containing only $1$s by repeatedly
    incrementing the binary encoding on the tape. This Turing machine can
    be encoded as the required tiling problem.
\end{itemize}

\subsubsection{Large Numbers}

The crux of the proof is encoding large numbers that can take values
between $1$ and $\expheight$-fold exponential.

A linear-length binary number could be encoded simply as a sequence of bits
\[
    b_0 \ldots b_\linlen \in \set{0,1}^\linlen \ .
\]
To aid with later constructions we will take a more oblique approach.
Let
$\tup{\tilesnum{1}, \hrelnum{1}, \vrelnum{1}, \inittilenum{1}, \fintilenum{1}}$
be a copy of the fixed tiling problem from the previous section for
which there is a unique solution, whose length must be exponential in
the width. In the future, we will need several copies of this problem,
hence the indexing here. Note, we assume each copy has disjoint tile
sets. Fix a width $\linlen$ and let $\nmax{1}$ be the corresponding
corridor length. A \emph{level-1} number can encode values from $1$ to
$\nmax{1}$. In particular, for $1 \leq \idxi \leq \nmax{1}$ we define
\[
    \tenc{1}{\idxi} = \tile^\idxi_1 \ldots \tile^\idxi_\linlen
\]
where
$\tile^\idxi_1 \ldots \tile^\idxi_\linlen$
is the tiling of the $\idxi$th row of the unique solution to the tiling problem.

A \emph{level-2} number will be derived from tiling a corridor of width
$\nmax{1}$, and thus the number of rows will be doubly-exponential. For
this, we require another copy $\tup{\tilesnum{2}, \hrelnum{2},
\vrelnum{2}, \inittilenum{2}, \fintilenum{2}}$ of the above tiling
problem. Moreover, let $\nmax{2}$ be the length of the solution for a
corridor of width $\nmax{1}$. Then for any $1 \leq \idxi \leq \nmax{2}$ we
define
\[
    \tenc{2}{\idxi} =
        \tenc{1}{1} \tile^\idxi_1
        \tenc{1}{2} \tile^\idxi_2
        \ldots
        \tenc{1}{\nmax{1}} \tile^\idxi_{\nmax{1}}
\]
where
$\tile^\idxi_1 \ldots \tile^\idxi_{\nmax{1}}$
is the tiling of the $\idxi$th row of the unique solution to the tiling
problem. That is, the encoding indexes each tile with it's column
number, where the column number is represented as a level-1 number.

In general, a *level-$\expheight$* number is of length
$(\expheight-1)$-fold exponential and can encode numbers
$\expheight$-fold exponential in size. We use a copy
$\tup{\tilesnum{\expheight},
      \hrelnum{\expheight},
      \vrelnum{\expheight},
      \inittilenum{\expheight},
      \fintilenum{\expheight}}$
of the above tiling problem and use a corridor of width
$\nmax{\expheight-1}$. We define $\nmax{\expheight}$ as the length of
the unique solution to this problem. Then, for any $1 \leq \idxi \leq
\nmax{\expheight}$ we have
\[
    \tenc{\expheight}{\idxi} =
        \tenc{\expheight-1}{1} \tile^\idxi_1
        \tenc{\expheight-1}{2} \tile^\idxi_2
        \ldots
        \tenc{\expheight-1}{\nmax{\expheight-1}}
            \tile^\idxi_{\nmax{\expheight-1}}
\]
where
$\tile^\idxi_1 \ldots \tile^\idxi_{\nmax{\expheight-1}}$
is the tiling of the $\idxi$th row of the unique solution to the tiling problem.

Note that we can define regular languages to check that a string is a
large number. In particular
\[
    \isnum{\expheight}{\linlen} =
    \begin{cases}
        \rebrac{\tilesnum{1}}^\linlen & \expheight = 1 \\
        \rebrac{
            \isnum{\expheight-1}{\linlen}
            \tilesnum{\expheight}
        }^\ast  & \expheight > 1 \ .
    \end{cases}
\]

\subsubsection{Hardness Proof}

We show that the satisfiability problem for $\strlinesl$ is
Tower-hard. We first introduce the basic framework of solving a hard
tiling problem. Then we discuss the two phases of transductions required
by the reduction. These are constructing a large boolean formula, and
then evaluating the formula. This two phases are described in separate
sections.

\paragraph{The Framework}

The proof is by reduction from a tiling problem over an
$\expheight$-fold exponential width corridor. In general, solving such
problems is hard for $\expheight$-ExpSpace.

Let $\nmax{\expheight}$ be the width of the corridor. Fix a tiling
problem
\[
    \tup{\tiles, \hrel, \vrel, \inittile, \fintile} \ .
\]

We will compose an $\strlinesl$ formula $\scon$ with a free
variable $\svar$. If $\scon$ is satisfiable, $\svar$ will contain a
string encoding a solution to the tiling problem. In particular, the
value of $\svar$ will be of the form
\[
    \begin{array}{c}
        \fullrow{1} \spacer \\
        \fullrow{2} \spacer \\
        \ldots \\
        \fullrow{\tileheight} \spacer \ .
    \end{array}
\]
That is, each row of the solution is separated by the $\spacer$ symbol.
Between each tile of a row is it's index, encoded using the large number
encoding described in the previous section.

The formula $\scon$ will use a series of replacements and assertions to verify
that the tiling encoded by $\svar$ is a valid solution to the tiling problem.
We will give the formula in three steps.

We will define the alphabet to be
\[
    \alphabet = \tiles \cup \nontiles
\]
where $\tiles$ is the set of tiles, and $\nontiles$ is the set of characters required to encode large numbers, plus $\spacer$.

The first part is
\[
    \begin{array}{l}
        \ASSERT{x \in
            \rebrac{
                \rebrac{
                    \isnum{\expheight}{\linlen} \tiles
                }^\ast
                \spacer
            }^\ast
        }; \\
        \ASSERT{x \in \isnum{\expheight}{\linlen} \inittile}; \\
        \ASSERT{x \in \alphabet^\ast \fintile \spacer}; \\
        \ASSERT{x \in
            \rebrac{
                \rebrac{
                    \sum\limits_{\tup{\tile_1, \tile_2} \in \hrel}
                        \isnum{\expheight}{\linlen} \tile_1
                        \isnum{\expheight}{\linlen} \tile_2
                }^\ast
                \rebrac{\isnum{\expheight}{\linlen} \tiles}^?
                \spacer
            }^\ast
        }; \\
        \ASSERT{x \in
            \rebrac{
                \rebrac{\isnum{\expheight}{\linlen} \tiles}
                \rebrac{
                    \sum\limits_{\tup{\tile_1, \tile_2} \in \hrel}
                        \isnum{\expheight}{\linlen} \tile_1
                        \isnum{\expheight}{\linlen} \tile_2
                }^\ast
                \rebrac{\isnum{\expheight}{\linlen} \tiles}^?
                \spacer
            }^\ast
        }; \\
    \end{array}
\]

The first asserts simply verify the format of the value of $\svar$ is as
expected and moreover, the first appearing element of $\tiles$ in the string is
$\inittile$, and the last element is $\fintile$.

The final two assertions check the horizontal tiling relation.  In particular,
the first checks that even pairs of tiles are in $\hrel$, while the second
checks odd pairs are in $\hrel$.

The main challenge is checking the vertical tiling relation. This is
done by a series of transductions operating in two main phases. The
first phase rewrites the encoding into a kind of large Boolean formula,
which is then evaluated in the second phase.

\paragraph{Constructing the Large Boolean Formula}

The next phase of the formula is shown below and explained afterwards.
For convenience, we will describe the construction using transductions.
After the explanation, we will describe how to achieve these transductions
using $\replaceall$.
\[
    \begin{array}{l}
        \svar^1_\expheight = \ap{\psst^1_\expheight}{\svar}; \\
        \svar^2_\expheight = \ap{\psst^2_\expheight}{\svar^1_\expheight}; \\
        \svar^3_\expheight = \ap{\psst^3_\expheight}{\svar^2_\expheight}; \\
        \svar^1_{\expheight-1}
            = \ap{\psst^1_{\expheight-1}}{\svar^3_\expheight}; \\
        \svar^2_{\expheight-1}
            = \ap{\psst^2_{\expheight-1}}{\svar^1_{\expheight-1}}; \\
        \svar^3_{\expheight-1}
            = \ap{\psst^3_{\expheight-1}}{\svar^2_{\expheight-1}}; \\
        \ldots \\
        \svar^1_1 = \ap{\psst^1_1}{\svar^3_2}; \\
        \svar^2_1 = \ap{\psst^2_1}{\svar^1_1}; \\
        \svar^3_1 = \ap{\psst^3_1}{\svar^2_1}; \\
        \svar_0 = \ap{\psst_0}{\svar^3_1} \ .
    \end{array}
\]

The Boolean formula is constructed by rewriting the encoding stored in
$\svar$. We need to check the vertical tiling relation by comparing
$\tile^\idxi_\idxj$ with $\tile^{\idxi+1}_\idxj$. However, these are
separated by a huge number of other tiles, which also need to be checked
against their counterpart in the next row.

The goal of the transductions is to "rotate" the encoding so that
instead of each tile being directly next to its horizontal counterpart,
it is directly next to its vertical counterpart. Our transductions do
not quite achieve this goal, but instead place the tiles in each row next to
potential vertical counterparts. The Boolean formula contains large
disjunctions over these possibilities and use the indexing by large
numbers to pick out the correct pairs.

The idea is best illustrated by showing the first three transductions,
$\psst^1_\expheight$, $\psst^2_\expheight$, and $\psst^3_\expheight$.
We start with
\[
    \begin{array}{c}
        \fullrow{1} \spacer \\
        \fullrow{2} \spacer \\
        \ldots \\
        \fullrow{\tileheight} \spacer \ .
    \end{array}
\]
The transducer $\psst^1_\expheight$ saves the row it is currently reading.
Then, when reading the next row, it outputs each index and tile of the current
row followed by a copy of the last row. The output is shown below. We use a
disjunction symbol to indicate that, after the transduction, the tile should
match one of the tiles copied after it. Between each pair of a tile and a
copied row, we use the conjunction symbol to indicate that every disjunction
should have one match. The result is shown below. To aid readability, we
underline the copied rows. The parentheses $\fbrac{}$ are also inserted to aid future parsing.
\[
    \begin{array}{c}
        \fbrac{
            \tenc{\expheight}{1} \tile^2_1
                \lor
                \underline{\fullrow{1}}
        }
        \land
        \ldots
        \land
        \fbrac{
            \tenc{\expheight}{\nmax{\expheight}} \tile^2_{\nmax{\expheight}}
                \lor
                \underline{\fullrow{1}}
        } \\
        \land \ldots \land \\
        \fbrac{
            \tenc{\expheight}{1} \tile^\tileheight_1
                \lor
                \underline{\fullrow{\tileheight-1}}
        }
        \land
        \ldots
        \land
        \fbrac{
            \tenc{\expheight}{\nmax{\expheight}}
                \tile^\tileheight_{\nmax{\expheight}}
                \lor
                \underline{\fullrow{\tileheight-1}}
        } \ .\\
    \end{array}
\]

After this transduction, we apply $\psst^2_\expheight$. This
transduction forms pairs of a tile, with all tiles following it from the
previous row (up to the next $\land$ symbol). This leaves us with a
conjunction of disjunctions of pairs. Inside each disjunct, we need to
verify that one pair has matching indices and tiles that satisfy the
vertical tiling relation $\vrel$. The result of the second transduction
is shown below.
\[
    \begin{array}{c}
        \fbrac{
            \tenc{\expheight}{1} \tile^2_1
            \tenc{\expheight}{1} \tile^1_1
            \lor
            \ldots
            \lor
            \tenc{\expheight}{1} \tile^2_1
            \tenc{\expheight}{\nmax{\expheight}} \tile^1_{\nmax{\expheight}}
        }
            \land
            \ldots
            \land
            \fbrac{
                \tenc{\expheight}{\nmax{\expheight}} \tile^2_{\nmax{\expheight}}
                \tenc{\expheight}{1} \tile^1_1
                \lor
                \ldots
                \lor
                \tenc{\expheight}{\nmax{\expheight}} \tile^2_{\nmax{\expheight}}
                \tenc{\expheight}{\nmax{\expheight}} \tile^1_{\nmax{\expheight}}
            } \\
        \land \ldots \land \\
        \fbrac{
            \tenc{\expheight}{1} \tile^\tileheight_1
            \tenc{\expheight}{1} \tile^{\tileheight-1}_1
            \lor
            \ldots
            \lor
            \tenc{\expheight}{1} \tile^\tileheight_1
            \tenc{\expheight}{\nmax{\expheight}}
                \tile^{\tileheight-1}_{\nmax{\expheight}}
        }
            \land
            \ldots
            \land
            \fbrac{
                \tenc{\expheight}{\nmax{\expheight}}
                    \tile^\tileheight_{\nmax{\expheight}}
                \tenc{\expheight}{1} \tile^{\tileheight-1}_1
                \lor
                \ldots
                \lor
                \tenc{\expheight}{\nmax{\expheight}}
                    \tile^\tileheight_{\nmax{\expheight}}
                \tenc{\expheight}{\nmax{\expheight}}
                    \tile^{\tileheight-1}_{\nmax{\expheight}}
            } \ .\\
    \end{array}
\]
Notice that we now have each tile in a pair with its vertical neighbour,
but also in a pair with every other tile in the row beneath. The indices
can be used to pick out the right pairs, but we will need further
transductions to analyse the encoding of large numbers.

To simplify matters, we apply $\psst^3_\expheight$. This transduction removes
the tiles from the string, retaining each pair of indices where the tiles
satisfy the vertical tiling relation. When the tiling relation is not
satisfied, we insert $\sffalse{\expheight}$. We use $\lmark$, $\spacer$, and
$\rmark$ to delimit the indices. We are left with a string of the form
\[
    \bigwedge \bigvee \lmark
        \tenc{\expheight}{\idxi} \spacer \tenc{\expheight}{\idxj}
    \rmark
    \lor \sffalse{\expheight} \lor \cdots \lor \sffalse{\expheight} \ .
\]
We will often elide the $\sffalse{\expheight}$ disjuncts for clarity. They
will remain untouched until the formula is evaluated in the next section.

We consider a pair $\lmark \tenc{\expheight}{\idxi} \spacer
\tenc{\expheight}{\idxj} \rmark$ to evaluate to true whenever $\idxi = \idxj$.
The truth of the formula can be computed accordingly. However, it's not
straightforward to check whether $\idxi = \idxj$ as they are large numbers. The
key observation is that they are encoded as solutions to indexed tiling
problems, which means we can go through a similar process to the transductions
above.

First, recall that $\tenc{\expheight}{\idxi}$ is of the form
\[
    \tenc{\expheight-1}{1} \vartile^\idxi_1
    \tenc{\expheight-1}{2} \vartile^\idxi_2
    \ldots
    \tenc{\expheight-1}{\nmax{\expheight-1}}
        \vartile^\idxi_{\nmax{\expheight-1}}
\]
where we use $\vartile$ to indicate tiles instead of $\tile$.

We apply three transductions $\psst^1_{\expheight-1}$,
$\psst^2_{\expheight-1}$, and $\psst^3_{\expheight-1}$. The first copies
the first index of each pair directly after the tiles of the second
index. That is, each pair
\[
    \lmark \tenc{\expheight}{\idxi} \spacer \tenc{\expheight}{\idxj} \rmark
\]
is rewritten to
\[
    \fbrac{
        \tenc{\expheight-1}{1} \vartile^\idxj_1
        \lor
        \tenc{\expheight}{\idxi}
    }
    \land
    \ldots
    \land
    \fbrac{
        \tenc{\expheight-1}{\nmax{\expheight-1}}
            \vartile^\idxj_{\nmax{\expheight-1}}
        \lor
        \tenc{\expheight}{\idxi}
    } \ .
\]
Then we apply a similar second transduction: each disjunction is
expanded into pairs of indices and tiles. The result is
\[
    \begin{array}{c}
        \fbrac{
            \tenc{\expheight-1}{1} \vartile^\idxj_1
            \tenc{\expheight-1}{1} \vartile^\idxi_1
            \lor
            \ldots
            \lor
            \tenc{\expheight-1}{1} \vartile^\idxj_1
            \tenc{\expheight-1}{\nmax{\expheight-1}}
                \vartile^\idxi_{\nmax{\expheight-1}}
        } \\
        \land
        \ldots
        \land \\
        \fbrac{
            \tenc{\expheight-1}{\nmax{\expheight-1}}
                \vartile^\idxj_{\nmax{\expheight-1}}
            \tenc{\expheight-1}{1} \vartile^\idxi_1
            \lor
            \ldots
            \lor
            \tenc{\expheight-1}{\nmax{\expheight-1}}
                \vartile^\idxj_{\nmax{\expheight-1}}
            \tenc{\expheight-1}{\nmax{\expheight-1}}
                \vartile^\idxi_{\nmax{\expheight-1}}
        } \ .
    \end{array}
\]
The third transduction replaces with $\sffalse{\expheight-1}$ all pairs where
we don't have $\vartile^\idxj_\idxk = \vartile^\idxi_{\idxk'}$ (recall, we need
to check that $\idxi = \idxj$ so the tiles at each position should be the
same). As before, for a single pair, this leaves us with a string formula of
the form
\[
    \bigwedge \bigvee \lmark
        \tenc{\expheight-1}{\idxi'} \spacer \tenc{\expheight-1}{\idxj'}
    \rmark
    \lor \sffalse{\expheight-1} \lor \cdots \lor \sffalse{\expheight-1} \ .
\]
Again, we will elide the $\sffalse{\expheight-1}$ disjuncts for clarity as they
will be untouched until the formula is evaluated.  Recalling that there are
many pairs in the input string, the output of this series of transductions is a
string formula of the form
\[
    \bigwedge \bigvee \bigwedge \bigvee \lmark
        \tenc{\expheight-1}{\idxi'} \spacer \tenc{\expheight-1}{\idxj'}
    \rmark \ .
\]

We repeat these steps using $\psst^1_{\expheight-2}$,
$\psst^2_{\expheight-2}$, $\psst^3_{\expheight-2}$ all the way down to
$\psst^1_1$, $\psst^2_1$, $\psst^3_1$. We are left with a string formula of the
form
\[
    \bigwedge \bigvee \cdots \bigwedge \bigvee \lmark
        \tenc{1}{\idxi'} \spacer \tenc{1}{\idxj'}
    \rmark \ .
\]
Recall each $\tenc{1}{\idxi'}$ is of the form
\[
    \vartile^{\idxi'}_1 \ldots \vartile^{\idxi'}_\linlen \ .
\]
The final step interleaves the tiles of the two numbers. The result is a string
formula of the form
\[
    \bigwedge \bigvee \cdots \bigwedge \bigvee \bigwedge \vartile \vartile' \ .
\]
This is the formula that is evaluated in the next phase.

To complete this section we need to implement the above transductions using
$\replaceall$.

First, consider $\psst^1_\expheight$.  We start with
\[
    \begin{array}{c}
        \fullrow{1} \spacer \\
        \fullrow{2} \spacer \\
        \ldots \\
        \fullrow{\tileheight} \spacer \ .
    \end{array}
\]
We are aiming for
\[
    \begin{array}{c}
        \fbrac{
            \tenc{\expheight}{1} \tile^2_1
                \lor
                \fullrow{1}
        }
        \land
        \ldots
        \land
        \fbrac{
            \tenc{\expheight}{\nmax{\expheight}} \tile^2_{\nmax{\expheight}}
                \lor
                \fullrow{1}
        } \\
        \land \ldots \land \\
        \fbrac{
            \tenc{\expheight}{1} \tile^\tileheight_1
                \lor
                \fullrow{\tileheight-1}
        }
        \land
        \ldots
        \land
        \fbrac{
            \tenc{\expheight}{\nmax{\expheight}}
                \tile^\tileheight_{\nmax{\expheight}}
                \lor
                \fullrow{\tileheight-1}
        } \ .\\
    \end{array}
\]
We use two $\replaceall$s. The first uses $\refbefore$ to do the main work of
copying the previous row into the current row a huge number of times. In fact,
$\refbefore$ will copy too much, as it will copy everything that came before,
not just the last row. The second $\replaceall$ will cut down the contents of
$\refbefore$ to only the last row. That is, we first apply
$\replaceall_{\pat_1, \rep_1}$ and then $\replaceall_{\pat_2, \rep_2}$ where
\[
    \begin{array}{rcl}
        \pat_1 &=& (\tile) \\
        \rep_1 &=& \$1 \triangleleft \refbefore \triangleright
    \end{array}
\]
and $\triangleleft$ and $\triangleright$ are two characters not in $\alphabet$,
and, letting $\alphabet_\spacer = \alphabet \setminus \set{\spacer}$,
\[
    \begin{array}{rcl}
        \pat_2 &=& \triangleleft
                \alphabet_\spacer^\ast \spacer (
                    \alphabet_\spacer^\ast
                ) \spacer \alphabet_\spacer^\ast
            \triangleright \\
        \rep_2 &=& \lor \$1 \ .
    \end{array}
\]
That is, the first replace adds after each tile the entire preceding string,
delimited by $\triangleleft$ and $\triangleright$. The second replace picks out
the final row of each string between $\triangleleft$ and $\triangleright$ and
adds the $\lor$. Notice that the second replace does not match anything between
$\triangleleft$ and $\triangleright$ on the first row. In fact, we need another
$\replaceall$ to delete the first row. That is $\replaceall_{\pat_3, \rep_3}$
where
\[
    \begin{array}{rcl}
        \pat_3 &=& \rebrac{
            \alphabet \cup \set{\triangleleft, \triangleright}
            }^\ast \triangleright \alphabet_\spacer^\ast \spacer \\
        \rep_3 &=& \varepsilon \ .
    \end{array}
\]
Notice, the pattern above matches any row containing at least one
$\triangleright$. This means only the first row will be deleted as delimiters
have already been removed from the other rows.  To complete the step, we
replace all $\spacer$ with $\land$ and insert the parenthesis $\fbrac{}$ using another $\replaceall$ (and a concatenation at the beginning and the end of the string).

The transduction $\psst^2_\expheight$ uses similar techniques to the above and
we leave the details to the reader. The same is true of the other similar
transductions $\psst^1_{\idxi}$ and $\psst^2_{\idxi}$.

Transduction $\psst^3_\expheight$ (and similarly the other $\psst^2_\idxi$)
replaces all pairs
\[
    \tenc{\expheight}{\idxi} \tile^2_\idxi
    \tenc{\expheight}{\idxi + 1} \tile^1_{\idxi+1}
\]
that do not satisfy the vertical tiling relation with $\sffalse{\expheight}$,
and rewrites them to
\[
    \lmark
        \tenc{1}{\idxi'} \spacer \tenc{1}{\idxj'}
    \rmark
\]
if the vertical tiling relation is matched. This can be done in two steps:
first replace the non-matches, then replace the matches. To replace the
non-matches we use $\replaceall_{\pat_1, \rep_1}$ where
\[
    \begin{array}{rcl}
        \pat_1 &=& \sum\limits_{\tup{\tile_1, \tile_2} \notin \vrel}
            \isnum{\expheight}{\linlen} \tile_1
            \isnum{\expheight}{\linlen} \tile_2 \\
        \rep_1 &=& \sffalse{\expheight} \ .
    \end{array}
\]
For the matches we use $\replaceall_{\pat_2, \rep_2}$ where
\[
    \begin{array}{rcl}
        \pat_2 &=& \sum\limits_{\tup{\tile_1, \tile_2} \in \vrel}
            (\isnum{\expheight}{\linlen}) \tile_1
            (\isnum{\expheight}{\linlen}) \tile_2 \\
        \rep_1 &=& \langle \$1 \spacer \$2 \rangle \ .
    \end{array}
\]

The final transduction takes a string of the form
\[
    \bigwedge \bigvee \cdots \bigwedge \bigvee \lmark
        \tenc{1}{\idxi'} \spacer \tenc{1}{\idxj'}
    \rmark
\]
where each $\tenc{1}{\idxi'}$ is of the form
\[
    \vartile^{\idxi'}_1 \ldots \vartile^{\idxi'}_\linlen \ .
\]
We need to interleave the tiles of the two numbers, giving a string of the form
\[
    \bigwedge \bigvee \cdots \bigwedge \bigvee \bigwedge \vartile \vartile' \ .
\]
This can be done with a single $\replaceall_{\pat, \rep}$ where
\[
    \begin{array}{rcl}
        \pat &=&
            \langle
                (\tilesnum{1}) \ldots (\tilesnum{1})
            \spacer
                (\tilesnum{1}) \ldots (\tilesnum{1})
            \rangle \\
        \rep &=&
            \langle
                \$1 \$(\linlen + 1)
                \land \cdots \land
                \$\linlen \$(2\linlen)
            \rangle \ . \\
    \end{array}
\]

\paragraph{Evaluating the Large Boolean Formula}

The final phase of $\scon$ evaluates the Boolean formula and is shown below.
Again we write the formula using transductions and explain how they can be done
with $\replaceall$.
\[
    \begin{array}{c}
        \svar^\land_0 = \ap{\psst_0}{\svar_0}; \\
        \svar^\lor_1 = \ap{\psst^\land_0}{\svar^\land_0}; \\
        \svar^\land_1 = \ap{\psst^\lor_1}{\svar^\lor_1}; \\
        \svar^\lor_2 = \ap{\psst^\land_1}{\svar^\land_1}; \\
        \svar^\land_2 = \ap{\psst^\lor_2}{\svar^\lor_2}; \\
        \svar^\land_3 = \ap{\psst^\land_2}{\svar^\land_2}; \\
        \ldots \\
        \svar^\lor_\expheight
            = \ap{\psst^\land_{\expheight-1}}{\svar^\land_{\expheight-1}}; \\
        \svar^\land_\expheight
            = \ap{\psst^\lor_\expheight}{\svar^\lor_{\expheight}}; \\
        \svar_f
            = \ap{\psst^\land_\expheight}{\svar^\lor_{\expheight}}; \\
        \ASSERT{\svar_f \in \pat_f}
    \end{array}
\]

The first transducer $\psst_1$ reads the string formula
\[
    \bigwedge \bigvee \cdots \bigwedge \bigvee \bigwedge \vartile \vartile' \ .
\]
copies it to its output, except replacing each pair
$\vartile \vartile'$
with $\sftrue{1}$ if $\vartile = \vartile'$ and with $\sffalse{1}$ otherwise.
This is requires two simple $\replaceall$ calls.

The remaining transductions evaluate the innermost disjunction or conjunction as
appropriate (the parenthesis $\fbrac{}$ are helpful here). For example
$\psst^\lor_1$
replaces the innermost
$\bigvee \sfvalue$
with $\sftrue{1}$ if $\sftrue{1}$ appears somewhere in the disjunction and
$\sffalse{1}$ otherwise.
This can be done by greedily matching any sequence of characters from
$\set{\sftrue{1}, \sffalse{1}, \lor}$
that contains at least one $\sftrue{1}$ and replacing the sequence with $\sftrue{1}$,
then greedily matching any remaining sequence of
$\set{\sffalse{1}, \lor}$
and replacing it with $\sffalse{1}$.
The evaluation of conjunctions works similarly, but inserts $\sftrue{2}$ and $\sffalse{2}$ in the move to the next level of evaluation.

The final assert checks that $\svar_f$ contains only the character $\sftrue{\expheight+1}$
and fails otherwise.

This completes the reduction.

\subsection{Exponential Space Copyless Algorithm}
\label{sec:appendix}

We argue that, when the PSSTs are copyless, satisfiability for $\strlinesl$ can be decided in exponential space.

The algorithm in the proof of Theorem~\ref{thm-main} for $\strlinesl$ applies a subset of the proof rules in Table~\ref{tab:calculus}. These proof rules branch on disjunctions and backward propagation through transductions and concatenations. We can explore all branches through the proof tree, storing the sequence of branches chosen in polynomial space. Thus, we can consider each branch independently.

A branch consists of a backwards propagation through PSSTs and concatenations. We will argue that each state of the automata constructed can be stored in exponential space.
Since these states are combinations of states of PSSTs and finite automata constructed by earlier stages of the algorithm, it is possible to calculate the next states from the current state on the fly.
Thus, if the states can be stored in exponential space, the full algorithm will only require exponential space.

We first consider the case where we have PSSTs and FAs rather than string functions using \regexp{}.
Let $n$ be the size of the largest PSST or FA.
Let $x$ be the maximum number of variables in the PSST.
Finally, let $\ell$ be the length of the longest branch of transductions and concatenations in the proof tree (which is linear in the size of the constraint).

The FA in the pre-image of a concatenation all have the same size as the output automaton $\cA$.
The FA $\cB$ in the pre-image of a PSST $\cT$ with output automaton $\Aut$ is an automaton such that
$\Lang(\cB) = \cR^{-1}_{\cT}(\Lang(\Aut))$.
It has states of the form
$(q, \rho, \Lambda, S)$,
where
    $q$ is a state of $\cT$,
    $\rho$ is a function from variables of $\cT$ to pairs of states of $\Aut$,
    $\Lambda \subseteq \cE(\tau_T)$, and
    $S \subseteq Q_T$.
Note, because we assume $\cT$ is copyless, $\rho$ is a function to pairs of states, not to sets of pairs of states.
The space needed to store a state of of $\cB$ is hence
$\bigO(s + 2 x s + 2 n)$
where $s$ is the space required to store a state of $\Aut$.
Consequently, after $\ell$ backwards propagations, we can store the states of the automaton in space
$\bigO(2^\ell x^\ell n)$.
That is, exponential space.
This remains true when the PSST and FA may be exponential in the size of the $\regexp$.

\end{document}